\documentclass[11pt,a4paper]{article}
\pdfoutput=1 
\usepackage{fullpage}

\usepackage{amsmath}
\usepackage{amsfonts}
\usepackage{amssymb}
\usepackage{amsthm}
\usepackage{mathtools}
\allowdisplaybreaks

\usepackage[tt=false]{libertine}
\usepackage[libertine,vvarbb]{newtxmath}

\usepackage{float}
\usepackage{xspace}

\usepackage{bm}

\usepackage[hyphens]{url}
\usepackage{hyperref}
\usepackage[svgnames]{xcolor}
\hypersetup{colorlinks={true},urlcolor={blue},linkcolor={DarkBlue},citecolor=[named]{DarkGreen}}
\usepackage{microtype}
\usepackage[capitalise,nameinlink,noabbrev]{cleveref}
\usepackage[backend=biber, style=alphabetic, maxbibnames=999, maxalphanames=999, backref=true]{biblatex}
\usepackage{enumitem}

\usepackage{subcaption}
\usepackage{tikz}
\usetikzlibrary{patterns}
\usetikzlibrary{patterns.meta}
\usetikzlibrary{math}

\usepackage{doi}

\renewcommand{\paragraph}[1]{\medskip\noindent\textbf{#1.\;}}

\def \R{\mathbb R}
\def \N{\mathbb N}

\newcommand{\sset}[1]{\left\{ #1\right\}}
\newcommand{\ssets}[1]{\{ #1\}}
\newcommand{\fwh}[1]{\; \left| \; #1 \right.}

\newcommand{\card}[1]{\left| #1 \right|}
\newcommand{\cards}[1]{| #1 |}

\newcommand{\bigunion}{\bigcup}     

\newcommand{\ifif}{\Longleftrightarrow} 
\newcommand{\inters}{\cap}

\newcommand{\norm}[1]{\left\lVert #1\right\rVert}

\DeclareMathOperator*{\expectation}{\mathbb E}
\newcommand{\expect}[2][]{\expectation_{#1}\nolimits\left[#2\right]}

\DeclareMathOperator*{\probability}{\mathrm{Pr}}
\newcommand{\prob}[1]{\probability\left[#1\right]}

\DeclareMathOperator*{\argmax}{argmax}

\DeclareMathOperator*{\supportdistro}{\mathrm{supp}}
\newcommand{\support}[1]{\supportdistro\left(#1\right)}

\renewcommand\vec{\bm}

\newcommand{\ppad}{\textup{PPAD}\xspace}

\newcommand*{\poly}{\operatorname{poly}}

\def\pcircuit{\textup{\textsc{Pure-Circuit}}\xspace}
\newcommand{\PURE}{\textup{\textsf{PURIFY}}\xspace}
\newcommand{\NOT}{\textup{\textsf{NOT}}\xspace}

\newcommand{\AND}{\textup{\textsf{AND}}\xspace}
\newcommand{\val}[1]{\boldsymbol{\mathrm{A}}[#1]}
\newcommand{\valonly}{\boldsymbol{\mathrm{A}}}
\newcommand{\garbo}{\ensuremath{\bot}\xspace}
\newcommand{\cbidder}{\ensuremath{\textsf{const}}\xspace}
\newcommand{\variables}{\ensuremath{X}\xspace}
\newcommand{\uvar}{\ensuremath{x}\xspace}
\newcommand{\vvar}{\ensuremath{y}\xspace}
\newcommand{\wvar}{\ensuremath{z}\xspace}

\theoremstyle{definition}
\newtheorem{definition}{Definition}

\theoremstyle{plain}
\newtheorem{theorem}{Theorem}[section]
\newtheorem{lemma}[theorem]{Lemma}
\newtheorem{corollary}[theorem]{Corollary}

\Crefname{claim}{Claim}{Claims}
\newtheorem{inftheorem}{Informal Theorem}

\theoremstyle{definition}
\newtheorem{remark}{Remark}
\newtheorem{example}{Example}
\newtheorem{question}{Open Question}

\title{On the Computation of Equilibria in Discrete First-Price Auctions\thanks{Aris Filos-Ratsikas was supported by the UK Engineering and Physical Sciences Research Council (EPSRC) grant EP/Y003624/1. Charalampos Kokkalis was supported by an EPSRC DTA Scholarship (Reference EP/W524384/1). Some preliminary work on this project was performed during a visit of Alexandros Hollender to the University of Edinburgh under the SICSA Distinguished Visiting Fellowship scheme.\newline 
\indent A preliminary version of this paper appeared in EC'24.}}

\author{
\begin{tabular}{c c}
& \\ \textbf{Aris Filos-Ratsikas} & \textbf{Yiannis Giannakopoulos}\\
\small{University of Edinburgh, United Kingdom} & \small{University of Glasgow, United Kingdom} \\
\href{mailto:Aris.Filos-Ratsikas@ed.ac.uk}{\small{\texttt{aris.filos-ratsikas@ed.ac.uk}}} & \href{mailto:yiannis.giannakopoulos@glasgow.ac.uk}{\small{\texttt{yiannis.giannakopoulos@glasgow.ac.uk}}}\\
& \\
\textbf{Alexandros Hollender} & \textbf{Charalampos Kokkalis}\\
\small{University of Oxford, United Kingdom} & \small{University of Edinburgh, United Kingdom} \\
\href{mailto:alexandros.hollender@cs.ox.ac.uk}{\small{\texttt{alexandros.hollender@cs.ox.ac.uk}}} & \href{mailto:charalampos.kokkalis@ed.ac.uk}{\small{\texttt{charalampos.kokkalis@ed.ac.uk}}}
\end{tabular}}

\date{}

\begin{document}
\maketitle

\begin{abstract}
We study the computational complexity of computing Bayes-Nash equilibria in first-price auctions with discrete value distributions and discrete bidding space, under general subjective beliefs. 
It is known that such auctions do not always have pure equilibria. In this paper, we prove that the problem of deciding their existence is NP-complete, even for approximate equilibria. 
On the other hand, it can be shown that \emph{mixed} equilibria are guaranteed to exist; however, their computational complexity has not been studied before. We establish the PPAD-completeness of computing a mixed equilibrium and we complement this by an efficient algorithm for finding symmetric approximate equilibria in the special case of iid priors. %
En route to these results, we develop a computational equivalence framework between continuous and discrete first-price auctions, which can be of independent interest, and which allows us to transfer existing positive and negative results from one setting to the other. 
Finally, we show that correlated equilibria of the auction can be computed in polynomial time.
\end{abstract}

\clearpage
\setcounter{tocdepth}{2}
\tableofcontents

\clearpage
\section{Introduction}

The first-price auction is arguably the simplest and most intuitive auction format: there is one good for sale, which gets allocated to the highest bidder (breaking ties uniformly at random), for a payment equal to her submitted bid. In its primitive forms, this type of auction dates back to ancient times \parencite{cassady1967auctions}, and became quite popular during the 17th century for the sale of paintings and estates. In recent times, the first-price auction has been central to the sale of online advertising space, known as the multi billion-dollar industry of ``ad impressions'' (e.g., see \parencite{balseiro2021robust})\footnote{In fact, the sale of ad impressions accounted for 58\% of Google's revenue in 2022, for an amount of \$162.45 billion; see \parencite{statista}.}. Indeed, most digital marketplaces, known as \emph{ad exchanges} ---including Google's Ad Manager--- are currently using first-price auctions for the sale of their ad impressions \parencite{paes2020competitive,despotakis2021first,conitzer2022multiplicative}. This mass migration to the first-price auction from other auction formats, such as the celebrated second-price auction \parencite{vickrey1961counterspeculation}, was coined ``the first-price movement'' and was mainly attributed to the auction's simplicity for the users, but also to increased revenue for the publishers. Indeed, in a survey of 100 publishers in November 2018, 78\% reported an increase in revenue due to migrating to the first-price auction~\parencite{digiday}. 

Contrary to its famous second-price counterpart, the first-price auction provides incentives to the participants to strategically choose their bids. For example, consider a bidder that is willing to pay \$10 for the good. If the bidder believes that it is unlikely for any other bidder to be willing to bid an amount equal or close to \$10, she would \emph{underbid}, aiming to win the good at a lower price. This induces a \emph{strategic game}, in which the bidders choose their bids aiming to maximize their personal utilities. Inherent in this game is the notion of a \emph{belief} of a bidder for the values (the aforementioned ``willingness to buy'') of the other bidders, and hence, in game-theoretic terms, this is a \emph{Bayesian game} of incomplete information. These games were formally introduced in the seminal work of~\textcite{harsanyi1967games} and capture situations in which the bidders are uncertain about the type of opponents they are facing in the game. The stable outcomes of these games, where no bidder wishes to unilaterally change her strategy, are known as \emph{Bayes-Nash equilibria} (or simply, \emph{equilibria}).

In the context of auctions, incomplete information is typically captured by probability distributions (also known as \emph{priors}) over the possible values that the other bidders have for the item. Interestingly, \citeauthor{vickrey1961counterspeculation}, in his seminal paper in 1961, defined and analysed a Bayesian game for the first-price auction, even before \citeauthor{harsanyi1967games}'s formal introduction of the concept. Following his work, a plethora of papers in the economics literature studied the equilibria of first-price auction; see, e.g., \parencite{vickrey1961counterspeculation,griesmer1967toward,maskin1985auction,maskin2000equilibrium,maskin2003uniqueness,chwe1989discrete,plum1992characterization,lebrun1996existence,lebrun1999first,lebrun2006uniqueness,lizzeri2000uniqueness,Athey2001}. The focus of these works is primarily on showing the existence of equilibria for different variants of the auction, their (non)-uniqueness, and in relatively few cases how to ``find'' them, i.e., deriving closed-form solutions that characterize them. 

In the computer science community, \textcite{fghlp2021_sicomp} recently put forward the study of the \emph{computational complexity} of finding these equilibria; that is, devising polynomial-time algorithms to compute them, or proving computational hardness results for appropriate complexity classes. They studied a setting with (subjective) continuous priors and discrete bids, for which the existence of a (pure) equilibrium is guaranteed by the work of \textcite{Athey2001}. Among other results, they showed that computing an equilibrium of the auction is PPAD-complete in general, but polynomial-time solvable in special cases. The class PPAD was defined by \textcite{papadimitriou1994complexity} and is believed to contain problems that are computationally hard to solve, including finding a mixed Nash equilibrium in general games \parencite{daskalakis2009complexity,chen2009settling}, among many others. Following up on \parencite{fghlp2021_sicomp}, \textcite{chen2023complexity} considered the same setting with \emph{common} priors and provided a PPAD-completeness result, but under a different tie-breaking rule, rather than uniform tie-breaking. 

A common feature of the aforementioned works is that they assume the priors to be \emph{continuous}, meaning that the values are drawn from continuous distributions over an infinite domain. While this assumption is quite common in the auction literature, considering auctions with \emph{discrete} values has clear merits as well. For example, such priors are often formed based on past bidding behaviour, or can be interpreted as the prices of reselling the good\footnote{Experimental works in auctions, e.g.~\parencite{coppinger1980incentives,cox1988theory}, refer to the values explicitly as ``resale values'' when communicating them to the participants.}; in either case, they are inherently of a discrete nature \parencite{escamocher2009existence}. Additionally, in experiments designed to evaluate the incentive properties of the auction, the values are ultimately chosen from a discrete set, e.g., see \parencite{coppinger1980incentives,cox1982theory,cox1983tests,cox1988theory,harrison1989theory}. Furthermore, cryptographic implementations of the auction, which are based on \emph{perturbations} of the values, are contingent on the discreteness in order to reason about equilibria in the perturbed space \parencite{miltersen2009privacy}. Finally, the setting of discrete values and discrete bids makes the problem more amenable to computational approaches, as the representation of the inputs (the priors) and the outputs (the strategies) is straightforward. In contrast, the issue of representation in continuous value spaces is a rather intricate one, and requires additional monotonicity assumptions on the strategies; see the discussion in \parencite{fghlp2021_sicomp} for more details. Other works that consider discrete priors in auctions include \parencite{cai2012algorithmic,fu2014optimal,papadimitriou2015optimal,koccyiugit2018robust,albert2022mechanism}.

Motivated by the above, we consider the problem of computing equilibria in first-price auctions with discrete values and discrete bids. We phrase the following general question:

\begin{question}
What is the computational complexity of finding equilibria in first-price auctions with discrete values and discrete bids?
\end{question}

\subsection{Our Results and Discussion}

\subsubsection{Pure equilibria} Contrary to the case of continuous values, it is known that in a first-price auction with discrete values, \emph{pure} Bayes-Nash equilibria may not even exist \parencite{maskin1985auction,escamocher2009existence}. This motivates the computational question of \emph{deciding} their existence. To this end, we provide our first main result:

\begin{inftheorem}\label{th:NP-completeness-informal}
    The problem of deciding the existence of a pure Bayes-Nash equilibrium, in a first-price auction with discrete subjective priors and discrete bids, is NP-complete.
\end{inftheorem}
The term ``subjective'' in the statement of \cref{th:NP-completeness-informal} refers to the fact that different bidders may have different beliefs about the values of other bidders. Subjective priors are also present in most of the results in \parencite{fghlp2021_sicomp}, as well as in the definition of general Bayesian games (e.g., see \cite{harsanyi1967games,myerson1997book,Jehle2001a}. Related notions of subjectivity have been considered rather widely in the literature, e.g., see \cite{hahn1973notion,fudenberg1986limit,battigalli1992learning,battigalli1988conjectural,kalai1993rational,rubinstein1994rationalizable,witkowski2012peer,frongillo2016geometric}. An interesting special case of subjective priors is the well-known \emph{independent private values (IPV)} model, where the prior distribution for every bidder $i$ is the same from the perspective of any other bidder $j$ (see, e.g., \parencite{maskin1985auction}).

\subsubsection*{$\varepsilon$-approximate Equilibria} In the context of game theory, it is often useful to talk about $\varepsilon$-approximate equilibria, i.e., sets of strategies where any unilateral deviation could improve a player's utility by at most $\varepsilon$ (here, in the additive sense). With regard to computational complexity in particular, this might be necessary, as exact equilibria in certain settings might involve irrational quantities in their description, making them a priori impossible to represent on a computer; see \parencite{fghlp2021_sicomp} for a discussion related to the first-price auction, and also \parencite{goldberg2011survey} for a more general context. In our setting with discrete values and bids, this is not an issue: all pure equilibria are rational. Still, in light of the aforementioned non-existence and NP-hardness results, one might wonder: ``Could $\varepsilon$-approximate equilibria always exist, and could we perhaps compute them in polynomial time?''. It turns out that the answer to both of these questions is ``no'', at least when $\varepsilon$ is a reasonably small constant.
In fact, the formal version of~\cref{th:NP-completeness-informal} is an \emph{inapproximability} result for $\varepsilon$-approximate equilibria, which we present after we establish a strengthening of the non-existence results of \parencite{maskin1985auction,escamocher2009existence} to approximate equilibria as well.  

\subsubsection{Mixed Equilibria} Motivated by the non-existence of pure equilibria for discrete values, we then turn our attention to \emph{mixed} Bayes-Nash equilibria of the first-price auction, in which agents are allowed to randomize among different bids; note that previous works \parencite{fghlp2021_sicomp,chen2023complexity} focus exclusively on the complexity of pure equilibria. The existence of such mixed equilibria is guaranteed, due to known existence results in general Bayesian games; see also our discussion in \cref{sec:correlated-equilibria}. This renders the (mixed) equilibrium-finding problem a \emph{total} search problem; i.e., one that always has a solution, and thus a candidate for membership in the aforementioned class PPAD. Our second main result establishes that the problem is, in fact, PPAD-complete.

\begin{inftheorem}\label{th:PPAD-completeness-informal}
    The problem of computing a mixed Bayes-Nash equilibrium of a first-price auction with discrete subjective priors and discrete bids is PPAD-complete.
\end{inftheorem}

We remark that, similar to \cref{th:NP-completeness-informal}, the PPAD-hardness part of \cref{th:PPAD-completeness-informal} also applies even to $\varepsilon$-approximate equilibria, for a constant $\varepsilon$, whereas the PPAD-membership applies to approximate equilibria for the level of accuracy required by the representation (i.e., inversely-exponential in the input size).

One could conceivably exploit the nature of the aforementioned general existence proofs, and combine them with the general membership technique of~\textcite{Daskalakis2006}, to obtain the PPAD-membership part of our result in \cref{th:PPAD-completeness-informal}. Here we take a different route, by establishing an interesting connection between the discrete and continuous auction settings. In particular, we prove a \emph{computational equivalence} between $\varepsilon$-approximate \emph{mixed} equilibria of the discrete first-price auction and $\varepsilon'$-approximate \emph{pure} equilibria of the continuous variant. This equivalence could serve as a useful tool in the future: the gist of our framework is that one can effectively think of pure equilibria in the continuous setting and mixed equilibria in the discrete setting as two sides of the same coin, allowing for translating results between the two settings.  

For example, in this paper we will use this equivalence to derive the aforementioned PPAD-membership as a corollary of the corresponding result for pure equilibria~\parencite{fghlp2021_sicomp}, by transforming the discrete auction into a continuous one, and computing a pure equilibrium there. The other direction of the equivalence yields a PPAD-hardness result, which is however somewhat unsatisfactory, as it only applies to \emph{monotone} equilibria. This is a by-product of the fact that, as mentioned earlier, pure equilibria of the continuous auction need to be monotone for reasons related to their representation \parencite{fghlp2021_sicomp}. For that reason, we provide a stand-alone, stronger PPAD-hardness result for all equilibria (not necessarily monotone) via a direct reduction from the PPAD-complete problem \textsc{Pure Circuit} of \textcite{deligkas2022pure}. This allows also our PPAD-hardness to hold even for instances with very simple bidding spaces.

\subsubsection{Positive Results} The NP- and PPAD-hardness results of the aforementioned theorems showcase the computational challenges in the quest for finding equilibria for general prior distributions. In quest for positive results, we first consider a natural restriction of these priors, to be independent and identically distributed (iid), and the question of computing \emph{symmetric} mixed Bayes-Nash equilibria, i.e., equilibria in which all the bidders use the same bidding strategy. For this case we provide a PTAS (polynomial-time approximation scheme); that is, an algorithm that finds an $\varepsilon$-approximate equilibrium for any $\varepsilon$, in time polynomial in the description of the auction, but possibly exponential in $1/\varepsilon$. In particular, if we aim for an $\varepsilon$-approximate equilibrium with $\varepsilon$ constant, then this yields a polynomial-time algorithm. We state the corresponding theorem informally below:

\begin{inftheorem}\label{th:informal-ptas}
The problem of computing a symmetric $\varepsilon$-approximate mixed Bayes-Nash equilibrium of a first-price auction with iid priors admits a PTAS. 
\end{inftheorem}
Our proof is based on formulating the equilibrium computation problem as a system of polynomial inequalities, which can be solved within inversely-exponential precision using known results from the literature (e.g., see \parencite{grigor1988solving}). To ensure that the running time of our algorithm is polynomial, we need to make sure that the number of variables in the system is small; in our case, that would correspond to the value and bidding spaces being of fixed size. For the former, we exploit the nature of monotone strategies and devise a succinct representation which we refer to as \emph{support-representation}. For the latter, we prove a ``shrinkage lemma'', which enables us to work on a substantially smaller bidding space of size $O(1/\varepsilon)$, find an $\varepsilon'$-approximate equilibrium there, and translate it to an $(\varepsilon+\varepsilon')$-equilibrium of the original auction. 

Finally, we consider \emph{correlated} equilibria, a more general equilibrium notion due to \textcite{aumann74}. In Bayesian games, there are several conceivable definitions of correlated equilibria; see, e.g., \parencite{aumann1987correlated,bergemann2013robust,bergemann2016bayes,forges-5defs-ce,forges-revisited,forges2023correlated} and \parencite[Sec. 6.3]{myerson1997book}. In this paper we adopt a standard notion, defined via the \emph{type-agent representation},\footnote{Some works refer to this representation as the \emph{agent normal form representation}, e.g., see \parencite{hartline2015no,ahunbay2024uniqueness}. We use the term ``type-agent representation'' instead, following \textcite[Sec.~2.8]{myerson1997book}.} a canonical normal-form game representation of Bayesian games that preserves the underlying fundamental equilibrium structure of the auction (e.g., see \cite[Sec.~2.8]{myerson1997book}, \cite[Sec.~7.2.3]{Jehle2001a}), also adopted in \parencite{hartline2015no,ahunbay2024uniqueness}. 
Our contribution here, which is presented in \cref{sec:correlated-equilibria}, is showing that (exact) correlated equilibria of the first-price auction can be computed in polynomial time.

\subsection{Further Related Work}

As we mentioned in the introduction, \textcite{fghlp2021_sicomp} were the first to study the complexity of equilibrium computation in first-price auctions with continuous priors and discrete bids, providing a PPAD-completeness result for the case of subjective priors. In follow-up work, \textcite{chen2023complexity} considered the IPV setting and proved a PPAD-completeness under an alternative, somewhat involved tie-breaking rule for the possible winners, rather than the standard uniform tie-breaking rule. Our shrinkage lemma is inspired by a conceptually similar idea in \parencite{chen2023complexity}, but, contrary to their work, it refers to auctions with discrete values (rather than continuous), subjective priors (rather than IPV), and mixed equilibria (rather than pure equilibria).

The setting with discrete values (and discrete bids), which we study in the present paper, was first studied by \cite{escamocher2009existence}, who provided preliminary results limited mainly to $2$ bidders with bivalued iid distributions. In a conceptually related paper, \textcite{wang2020bayesian} considered the computation of equilibria in first-price auctions with discrete values and continuous bids under a non-standard tie-breaking rule, which was used primarily by \textcite{maskin2000equilibrium} as an intermediate step to prove results for uniform tie-breaking. Other works that study the complexity of equilibrium computation in Bayesian auctions, and Bayesian games in general, include \parencite{Gottlob:2007aa,Conitzer:2008aa,Papadimitriou2008a,cai2014simultaneous}.

\section{Preliminaries}

In a (discrete, Bayesian) \emph{first-price auction (DFPA)}, there is a set
$N=\{1,2,\ldots,n\}$ of \emph{bidders} and one item for sale. Each
bidder $i$ has a \emph{value} $v_i\in V_i$ for the item and submits a \emph{bid} $b_i \in B$.
Sets $V_1,V_2,\dots,V_n,B$ are \emph{finite} subsets of $[0,1]$ and are called the \emph{value spaces} of the bidders and the \emph{bidding space} of the auction, respectively.

The item is allocated to the highest bidder, who has to submit a payment equal to her bid. In case of a tie for the highest bid, the winner is determined according to the
\emph{uniform tie-breaking} rule.
That is, for a \emph{bid profile} $\vec{b}=(b_1,\ldots,b_n)$, the \emph{ex-post utility} of bidder $i$ with value $v_i$ is defined as
\begin{equation}
\label{eq:ex_post_utilities}
\tilde{u}_i(\vec{b};v_i) \coloneq 
\begin{cases}
\frac{1}{\cards{W(\vec{b})}}(v_i-b_i), & \text{if}\;\; i\in W(\vec{b}), \\
0, & \text{otherwise}, 
\end{cases}
\qquad\text{where}\;\; W(\vec{b})=\argmax_{j\in N} b_j
\end{equation}

\paragraph{Bayesian priors}
Each bidder $i\in N$ has a \emph{subjective belief} for the values of each of the other bidders $j\in N\setminus\ssets{i}$, in the form of a \emph{prior} distribution $F_{i,j}$ over $V_{j}$. This induces a product distribution $\vec{F}_{-i}\coloneq\times_{j\neq i} F_{i,j}$ for the values $\vec{v}_{-i}=(v_1,\dots,v_{i-1},v_{i+1},\dots, v_n)\in \times_{j\in N\setminus\ssets{i}} V_j\eqqcolon\vec{V}_{-i}$ of the other bidders. In
other words, from the perspective of bidder $i$, the values $v_j$ for $j\neq i$ are
drawn \emph{independently} from distributions $F_{i,j}$.

We will also be interested in special cases of these Bayesian priors, namely: 
\begin{itemize}
    \item[---] \emph{Independent Private Values (IPV)}, where $F_{i,j} = F_{i',j}$ for all $j \in N$ and $i, i' \in N\setminus\ssets{j}$. In this case, notation can be simplified by using $F_j$ instead of $F_{i,j}$. This can be interpreted as the value profile $\vec{v}\in \times_{i\in N} V_i$ being drawn from the product distribution $\vec{F}=\times_{j\in N} F_j$.
    \item[---] \emph{Identical Independent Values (iid)}, which is a special case of IPV above where $V_{i}=V_{i'}$ and $F_{i}=F_{i'}$ for all $i,i'\in N$. In other words, bidder values are iid according to a common distribution $F$.
\end{itemize}

\subsection{The Auction Game}
\label{sec:model}

The DFPA described above gives rise to a Bayesian game of incomplete information, where each
bidder $i$ chooses her bid based on her own (true) value $v_i$ and her beliefs
$\vec{F}_{-i}$ about the other bidders. 
A (mixed) \emph{strategy} of bidder $i$ is a function $\beta_i:
V_i \rightarrow \Delta(B)$ mapping values to distributions over bids.\footnote{See, e.g., \cite[Sec.~3.9]{myerson1997book}, \cite[Sec.~7.2.3]{Jehle2001a}, \parencite{milgrom1985distributional,lucier2017equilibria}.}$^{,}$\footnote{Formally, for any finite set $X$ we define $\Delta(X)$ as the simplex 
$$\Delta(X)\coloneq \sset{p\in[0,1]^X\fwh{\sum_{x\in X} p(x)=1}},$$ 
where $p(x)$ can be interpreted as the mass that the distribution assigns to element $x\in X$. We will denote the support of a distribution $p\in \Delta(X)$ with $\support{p}\coloneqq \sset{x\in X \fwh{p(x)>0}}$.} 
\emph{Pure} strategies correspond to the special case where a mixed strategy
$\beta_i$ always assigns full mass on single bids; that is, for all $v_i\in V_i$
there exists a $b_i \in B$ such that
$\beta_i(v_i)(b_i)=1$ and
$\beta_i(v_i)(b)=0$ for all $b\neq b_i$. Therefore,
for simplicity, we will sometimes represent pure strategies directly as functions
$\hat\beta_i:V_i \to B$ from values to bids.

Given a strategy profile $\vec{\beta}_{-i}\in \times_{j\in N\setminus\ssets{i}}\Delta(B)^{V_j}$ of the other bidders, the (interim) \emph{utility} of a bidder $i$ with value $v_i$, when bidding $b\in B$, is given by
\begin{align}
u_i(b,\vec{\beta}_{-i};v_i) 
    &\coloneq\expect[\vec v_{-i}\sim \vec{F}_{-i}]{\expect[\vec b_{-i}\sim \vec{\beta}_{-i}(\vec{v}_{-i})]{\tilde{u}_i(b,\vec{b}_{-i};v_i)}} \notag\\
    &=  \sum_{\vec{v}_{-i}\in \vec{V}_{-i}}\left(\prod_{j\in N\setminus\ssets{i}} f_{i,j}(v_j)\right)\sum_{\vec{b}_{-i}\in B^{N\setminus\ssets{i}}}\left(\prod_{j\in N\setminus\ssets{i}} \beta_j(v_j)(b_j)\right)\tilde{u}_i(b,\vec{b}_{-i};v_i).\label{eq:DFPA-utility-interim-mixed}
\end{align}
where $\vec{\beta}_{-i}(\vec{v}_{-i})$ is a shorthand for the product distribution
$\times_{j\in N\setminus\ssets{i}}\beta_j(v_j)$, and $f_{i,j}$ denotes the probability mass function (pmf) of $F_{i,j}$. Recall that $\beta_j(v_j)(b_j)$ denotes the probability that bidder $j \neq i$ submits bid $b_j$
when having value $v_j$.
This can be viewed as the bidder computing her utility as follows: (a) she draws a value
$v_j$ for each bidder $j \neq i$ from her corresponding subjective prior $F_{i,j}$; (b) she uses the strategy ``rules'' $\vec{\beta}_{-i}$
of the others to map the values obtained in (a) to actual bid distributions over $B$.

 For
convenience, we also extend, in the natural way, the definition above to handle
the case where bidder $i$ randomizes over her bids; that is, for
$\vec{\gamma}\in\Delta(B)$ we define:
\begin{equation}
\label{eq:DFPA-utility-interim-mixed-randomized-bidding}
u_i(\vec{\gamma},\vec{\beta}_{-i};v_i) 
\coloneq \expect[b\sim \vec{\gamma}]{u_i(b,\vec{\beta}_{-i};v_i)}
= \sum_{b\in B} \vec{\gamma}(b) u_i(b,\vec{\beta}_{-i};v_i).
\end{equation}

\paragraph{Equilibria}
We proceed to define the two main (exact and approximate) equilibrium notions that we will study in this paper. 

\begin{definition}[$\varepsilon$-approximate mixed Bayes-Nash equilibrium of the DFPA]\label{def:approx-mixed-bayes-nash-equilibrium}
Let $\varepsilon \geq 0$.
A (mixed) strategy profile $\vec{\beta}=(\beta_1, \ldots, \beta_n)$ is an (interim) $\varepsilon$-approximate mixed Bayes-Nash equilibrium (MBNE) of the DFPA if for any bidder $i \in N$ and any value $v_i \in V_i$, 
\begin{equation}
\label{eq:MBNE-def-condition-full}
u_i(\beta_i(v_i),\vec{\beta}_{-i};v_i) \geq u_i(\vec{\gamma},\vec{\beta}_{-i};v_i) - \varepsilon \qquad \text{for all}\;\; \vec{\gamma}\in \Delta(B).
\end{equation}
We will refer to a $0$-approximate MBNE as an \emph{exact} MBNE.
\end{definition}
In the special case where all bidders choose the same strategies, i.e.\ $V_i=V_i'$ and $\beta_i=\beta_{i'}$ for all $i,i'\in N$, the (approximate) equilibrium will be called \emph{symmetric}.
\begin{remark}
    \label{note:MBNE-def-pure-deviation}
    It is straightforward to check that condition~\eqref{eq:MBNE-def-condition-full} can be equivalently stated to just range over all \emph{pure} deviations $\gamma\in B$ (instead of all \emph{mixed} ones $\vec{\gamma}\in\Delta(B)$), without affecting~\cref{def:approx-mixed-bayes-nash-equilibrium}. 
\end{remark}

It is known (see, e.g., \cite[Theorem~7.3]{Jehle2001a}) that finite Bayesian
games (and thus, DFPA as well) always have at least one exact MBNE. Therefore, existence of $\varepsilon$-approximate MBNE is
also guaranteed for any $\varepsilon > 0$. 

Similarly, one can define the notion of pure equilibria:

\begin{definition}[$\varepsilon$-approximate pure Bayes-Nash equilibrium of the DFPA]\label{def:approx-pure-bayes-nash-equilibrium}
Let $\varepsilon \geq 0$.
A pure strategy profile $\hat{\vec{\beta}}=(\hat\beta_1, \ldots, \hat\beta_n)$ is an (interim) $\varepsilon$-approximate pure Bayes-Nash equilibrium (PBNE) of the DFPA if for any bidder $i \in N$ and any value $v_i \in V_i$, 
\[
u_i(\hat\beta_i(v_i),\hat{\vec{\beta}}_{-i};v_i) \geq u_i(b,\hat{\vec{\beta}}_{-i};v_i) - \varepsilon \qquad \text{for all}\;\; b\in B.
\]
\end{definition}
In contrast to mixed equilibria, discrete first-price auctions do not, in general, have exact pure equilibria~\cite{escamocher2009existence}. In~\cref{th:tight_existence_approx}, we extend this nonexistence result to \emph{approximate} (pure) equilibria.

\paragraph{No overbidding} Following the literature of first-price auctions in economics \parencite{maskin2000equilibrium,maskin2003uniqueness,lebrun2006uniqueness}, as well as in computer science \parencite{fghlp2021_sicomp,wang2020bayesian,escamocher2009existence,chen2023complexity}, we will also make the standard no overbidding assumption. In our context, this translates into strategies never assigning positive probability to some bid $b_i$ which is larger than the bidder's value $v_i$. Note that, if we allow the bidders to abstain from the auction, by adopting what is known as the \emph{null bid} assumption in the literature (see~\parencite{maskin1985auction,Athey2001}), i.e.\ that $0\in B$, then any overbidding strategy is weakly dominated by a strategy that transfers any bidding probability mass from a $b_i > v_i$ to the zero bid.

\paragraph{Representation} For our model to be fit for computational purposes, we have to determine how exactly the inputs and the outputs of the associated computational problems will be represented. Similarly to \parencite{fghlp2021_sicomp,chen2023complexity}, we will assume that the bidding space $B$ is given explicitly as part of the input in the form of rational numbers. The value spaces $V_i$ will also be given explicitly, in the same way. 
The value distributions $F_{i,j}$ will be described explicitly via their probability mass functions. That is, we are given a list of rationals $\{f_{i,j}(v)\}_{v\in V_j}$ representing the probability that bidder $j$ has true value $v$, from the perspective of bidder $i$. 
The mixed strategies in the output are described explicitly by (rational) numbers $\ssets{p_{i}(v,b)}\in[0,1]$ representing the probability that bidder $i$ submits bid $b\in B$ when her true value is $v\in V_i$. Note that, for the special case of pure strategies it is $\ssets{p_{i}(v,b)}\in\ssets{0,1}$.  

\subsection{Efficient Computation of Bidder Utilities}\label{sec:key-properties}
We conclude our preliminaries with the following lemma, which will be useful for several of our results throughout the paper. The lemma establishes that the expected utilities, given mixed bidding strategies of the bidders, can be computed in polynomial time. We remark that \textcite[Lemma~3.2]{fghlp2021_sicomp} proved a similar lemma for the case of pure strategies (and continuous value spaces). Our proof is similar, but extra care is needed in order to handle the case of mixed strategies; we use similar notation to \parencite{fghlp2021_sicomp}, for consistency.

\begin{lemma}
    \label{lemma:DFPA-utilities-efficient}
    Fix a DFPA. For any bidder $i\in N$ and any true value $v_i\in V_i$,
        the utility\footnote{See~\eqref{eq:DFPA-utility-interim-mixed-randomized-bidding} and \eqref{eq:DFPA-utility-interim-mixed} for the utility definition.} $u_i(\vec\gamma,\vec{\beta}_{-i};v_i)$ of $i$, given as input any distribution $\vec{\gamma}$ over her bids and any mixed bidding strategy profile $\vec{\beta}_{-i}$ of the other bidders, is computable in polynomial time.
\end{lemma}

\begin{proof}
    Without loss of generality, we can reorder the bidders such that the one whose utility we are calculating is the last, i.e., $n=\card{N}$. Let $H_n(b,\vec{\beta}_{-n})$ denote the probability of bidder $n$ winning when bidding $b\in B$ against the mixed strategies $\vec{\beta}_{-n}$ of the remaining bidders (based on bidder $n$'s subjective beliefs). We can write bidder $n$'s utility when bidding $b$ as $u_n(b,\vec{\beta}_{-n};v)=(v-b)H_n(b,\vec{\beta}_{-n})$, so it suffices to show how to efficiently compute $H_n(b,\vec{\beta}_{-n})$. We can express this as:
    \begin{equation}
    \label{eq:H-function}
        H_n(b,\vec{\beta}_{-n})=\sum_{r=0}^{n-1}\frac{1}{r+1}T_n(b,n-1,r)
    \end{equation}
    where, for $0\leq r \leq \ell \leq n-1$, we use $T_n(b,\ell,r)$ to denote the probability that \emph{exactly} $r$ out of the first $\ell$ bidders bid $b$, and the remaining $\ell - r$ bidders all bid below $b$. Once again, this probability is based on bidder $n$'s subjective priors. Next, for a given bidder $j$, let:
    \begin{align}
       g_{j,b} &:= \displaystyle \mathop{\mathbb{E}}_{v_{j} \sim F_{n,j}} \left[ \Pr_{\xi \sim \beta_j(v_j)}[\xi = b] \right] = \sum_{v_j \in V_j}\beta_j(v_j)(b)f_{n,j}(v_j) \label{eq:def-utility-function-g}\\
        G_{j,b} &:= \displaystyle \mathop{\mathbb{E}}_{v_{j} \sim F_{n,j}} \left[ \Pr_{\xi \sim \beta_j(v_j)}[\xi < b] \right] = \sum_{\substack{b' \in B, \\ b' < b}} g_{j,b'} \label{eq:def-utility-function-G}
    \end{align}
    where $f_{n,j}$ is the probability mass function corresponding to bidder $n$'s beliefs about bidder $j$'s value. What we have defined above is the probability (perceived from the perspective of bidder $n$) that bidder $j$ bids \emph{exactly} $b$ (denoted as $g_{j,b}$) and below $b$ (denoted as $G_{j,b}$). The computation of these takes into account two sources of randomness; first, we draw a value $v_j$ for bidder $j$ from the prior distribution $F_{n,j}$, and then we pick a bid $\xi$, given the distribution over bids that $j$'s strategy represents at the drawn value $\beta_j(v_j)$. It is easy to see that these quantities can be computed in time $O(|V||B|)$ with access to the pmf of the prior distributions as well as the distributions over bids corresponding to the mixed strategies.
    We can now express the term $T_n(b,n-1,r)$ from \eqref{eq:H-function} using the newly defined quantities of \cref{eq:def-utility-function-g,eq:def-utility-function-G}:
    \begin{equation}
    \label{eq:inefficient-T}
        T_n(b,n-1,r) = \sum_{\substack{S \subseteq [n-1] \\ \card{S}=r}} \prod_{j \in S}g_{j,b} \cdot \prod_{j \notin S}G_{j,b}     
    \end{equation}
    where we have used the notation $[n-1]$ to indicate the set of the integers from $1$ to $n-1$ inclusive.

    Note that \eqref{eq:inefficient-T} cannot be computed efficiently in any obvious way, since the number of summands can be exponential in $n$. To overcome this issue, we proceed with a dynamic programming algorithm for computing $T_n(b,n-1,r)$, in a similar manner to ~\parencite[Lemma 3.2]{fghlp2021_sicomp}.
    The probabilities $T_n(b, \ell, k)$ can be computed from $g_{j,b}$ and $G_{j,b}$ (which were defined in \cref{eq:def-utility-function-g,eq:def-utility-function-G} respectively), via dynamic programming, conditioning on bidder $\ell$'s bid in the following way:
    \begin{align*}
    	T_n(b,0,0) &=1;              &\\
    	T_n(b,\ell,k) &=0,&\quad\text{for}\;\; k>\ell;\\
    	T_n(b,\ell+1,0  ) &=T_n(b,\ell,0)G_{\ell+1,b};&\\
    	T_n(b,\ell+1,k+1) &=T_n(b,\ell,k)g_{\ell+1,b} + T_n(b,\ell,k+1)G_{\ell+1,b}; &\text{for}\;\; k\leq\ell.
    \end{align*}
    Thus, all values of $T_n(b,n-1,k)$, for $k=0,\ldots,n-1$, can be computed with a total number of $O(n^2)$ recursive calls, so that $H_n(b,\vec{\beta}_{-n})$ can be computed in polynomial time.
    This directly implies that the expected utility $u_n(\vec{\beta} ;v)$ of bidder $n$, given as input a mixed strategy profile $\vec{\beta}$, is efficiently computable, by virtue of it being a convex combination of at most $|B|$ pure strategy utilities.
\end{proof}

\section{Pure Equilibria}

In this section, we present our results for the pure Bayes-Nash equilibria (PBNE) of first-price auctions. First, it is already known from \parencite{escamocher2009existence} that in the case of discrete values, PBNE may not exist, even for very simple instances of two bidders and iid priors with support size $2$. In the following theorem, we strengthen this impossibility result to $\varepsilon$-approximate PBNE, for a reasonably large constant value of $\varepsilon$. We delegate the proof of the theorem to \cref{sec:existence_discrete_negative}.

\begin{theorem}
    \label{th:tight_existence_approx}
    For any $\varepsilon<\frac{1}{72}$, there exist discrete first-price auctions,
    even with only two bidders and iid priors with support size $2$, that do not have $\varepsilon$-approximate pure equilibria. 
    \end{theorem}

Driven by the non-existence result of \cref{th:tight_existence_approx}, we consider the computational complexity of deciding the existence of PBNE in a DFPA. Our main result of this section is that, when the prior distributions are subjective, the problem is NP-complete, even for approximate equilibria. 
The following statement is the formal version of \cref{th:NP-completeness-informal}. In the interest of space and the flow of our presentation, we provide only a high-level sketch here. The full proof of our theorem is delegated to \cref{app:np-hardness}.

\begin{theorem}
    Consider any $\varepsilon < \frac{1}{180}$. The problem of deciding the existence of an $\varepsilon$-PBNE of a DFPA is NP-complete. 
\end{theorem}
\begin{proof}[Proof sketch]
We begin by showing membership in NP, demonstrating that we can verify any positive certificate in polynomial time. Given a pure strategy profile $\hat{\vec{\beta}}$, it suffices to efficiently verify that strategy $\hat{\beta}_i$ is an $\varepsilon$-best-response for each bidder $i \in N$ and value $v_i \in V_i$. To this end, we can quantify over all, polynomially many in the size of the input, possible values of $i \in N,v_i \in V_i, b\in B$, and compute the corresponding utilities. By \cref{lemma:DFPA-utilities-efficient}, this can be done in polynomial time.

We now turn to the NP-hardness of the problem. We reduce from the well-known NP-complete Circuit Satisfiability problem (\textsc{Circuit-SAT}) \parencite{Garey1979a}. For ease of notation, our reduction will use the bidding space $B=\{ 0,1,2,3 \}$ and the value space (common for all bidders) $V=\{ 0, \frac{9}{4}, 10 \}$. With this setup, we will prove the result for $\varepsilon < 1/18$.  
Dividing everything by the largest entry (in this case $10$), we bring the problem back to the desired domain $[0,1]$, while keeping all comparisons between utilities in the proof intact. Hence, we obtain the result for our version of the problem for $\varepsilon<1/180$.\smallskip

\noindent The proof can be outlined in the following steps:

\begin{enumerate}[leftmargin=*]
    \item[---] We define the set of valid strategies (mapping from values to bids) $\{ s_0, s_1 \}$ that correspond to the boolean values \emph{false} and \emph{true}, respectively.\smallskip

    \item[---] We introduce 3 gadgets (\emph{projection}, \emph{OR}, and \emph{NOT}), which are used to simulate the \emph{OR} and \emph{NOT} gates of the circuit, as well as a \emph{split node}, which is used to copy the value of its input to two outputs. Note that any boolean circuit can be converted in polynomial time to another one that only uses these gates. The \emph{projection} gadget is used in order to correctly map any bidder that plays a strategy different than $s_0$ and $s_1$ back to the set of valid strategies. By introducing bidders with appropriate subjective prior distributions, we establish that all the induced bidders can play $\varepsilon$-best-responses simultaneously if and only if the boolean values specified by their strategies satisfy the definitions of the corresponding gates of the circuit. \smallskip

    \item[---] For the simulation step, we assumed that all input bidders played a valid strategy. To ensure this, we introduce the \emph{input} gadget. This can be seen as introducing an extra bidder $i'$ for each bidder $i$ corresponding to an input of the circuit, with subjective priors resembling a \emph{projection} gate from $i$ to $i'$ and vice versa. We show in our analysis that this guarantees that $i$ and $i'$ can simultaneously play $\varepsilon$-best-responses if and only if they both play the same, valid, strategy. \smallskip

    \item[---] To finish off the construction, it remains to encompass the notion of the satisfiability of the circuit, meaning distinguishing the cases where the output can get the value \emph{true}. To achieve this, we introduce the final gadget of our reduction, the \emph{output} gadget. This introduces two new bidders to the DFPA, $k$ and $\ell$, which initially only have subjective beliefs that put positive probability mass to some value greater than $0$ for each other. We show that they cannot simultaneously play $\varepsilon$-best-responses. Then, we describe $k$'s prior distribution for the bidder $i$ corresponding to the output of the circuit, such that in the resulting DFPA it is possible for all 3 bidders ($i,k,\ell$) to simultaneously play $\varepsilon$-best-responses if and only if $i$ plays strategy $s_1$. Therefore, combining this with the previous steps, the resulting DFPA has an $\varepsilon$-PBNE for $\varepsilon < \frac{1}{18}$ if and only if the circuit is satisfiable.\qedhere
\end{enumerate}
\end{proof}

\section{Mixed Equilibria} 
\label{sec:mixed-equilibria}

In this section, we switch our focus to the more general class of mixed equilibria. At first glance, one might be inclined to believe that since both the value space and the bidding space are discrete, the existence of MBNE follows immediately by Nash's Theorem \parencite{nash1951non}. However, this is not the case, since the DFPA game is a Bayesian game of incomplete information. MBNE in (general) Bayesian games have been extensively discussed in the literature (see, e.g., \parencite{milgrom1985distributional,aumann1961mixed}, \parencite[Chapter 7.2.3]{Jehle2001a}, and \parencite[Chapter 2.8]{myerson1997book}). A common way to prove their existence is via transforming the Bayesian game into an appropriately constructed normal form game, and invoking Nash's Theorem there; see, e.g., \parencite[Theorem 7.3]{Jehle2001a}. 
This normal form game formulation is the \emph{type-agent representation} that we mentioned in our introduction, and will be critical for our discussion on correlated equilibria later, in \cref{sec:correlated-equilibria}. However, it is not needed for any of our results of the present section. Note that these general existence results for mixed equilibria in Bayesian games directly imply the existence of a MBNE in DFPA as well, rendering the problem of finding a MBNE, a total search problem.  

We will also provide the proof of \cref{th:PPAD-completeness-informal}, namely the PPAD-completeness of computing MBNE of the DFPA. En route to this result, we will present an interesting \emph{computational equivalence theorem} between approximate mixed Bayes-Nash equilibria of the DFPA with discrete values and approximate pure Bayes-Nash equilibria of the DFPA with continuous values, the setting studied in \parencite{fghlp2021_sicomp,chen2023complexity}. This equivalence allows us to translate positive and negative results in either direction; importantly, for the purpose of this section, it will help us establish the PPAD-membership of computing a MBNE of the DFPA. In principle, the equivalence theorem could also be used to get a PPAD-hardness result for the same problem, albeit somewhat weak in nature. Instead, we will provide a stronger, stand-alone PPAD-hardness result for computing a MBNE, which does not rely on the equivalence theorem.

\subsection{Further Definitions}

Before we present our equivalence theorem, we need to introduce some necessary terminology for the setting of first-price auctions with continuous priors (and discrete bids), as well as some further definitions that will be useful for the subsequent results. 

\paragraph{Continuous first-price auctions} A continuous first-price auction (CFPA) is defined very similarly to DFPA, they key difference being that the value spaces $V_i$ are continuous intervals of real numbers and the prior distributions $F_{i,j}$ are continuous. The bidding space $B=\{b_1,\ldots,b_m\}$ is still discrete. For the purpose of this paper, it is enough to consider distributions with piecewise constant densities, represented explicitly by the endpoints and the height of the intervals. As we mentioned in the introduction, the representation of strategies in a CFPA is not as straightforward as in the case of a DFPA. For this reason, one needs to focus on \emph{monotone strategies}, meaning that bidders submit higher bids on higher values. Using this, a strategy can be represented as a sequence of \emph{jump points} $\alpha_1,\ldots,\alpha_{m-1}$, where $\alpha_i \in V_i$ is the largest value for which bidder $i$ bids $b_i$ or lower. The notion of a pure Bayes-Nash equilibrium (PBNE) of a CFPA is defined similarly to a DFPA. We refer the reader to \parencite[Sec.~2]{fghlp2021_sicomp} for more details.   \medskip

\noindent It will be useful to define a notion of monotonicity, appropriate for mixed strategies:

\begin{definition}[Monotone mixed strategies in a DFPA]
    \label{def:monotonicty-discrete}
Consider a DFPA with value space $V_i$ for each bidder $i$.
A mixed strategy $\beta_i\in \Delta(B)^{V_i}$ (of bidder $i$) will be called \emph{monotone} if
$$
\max \support{\beta_i(v)} \leq \min \support{\beta_i(v')}
\qquad
\text{for all}\;\; v,v' \in V_i\;\; \text{with}\;\; v < v'. 
$$
A strategy profile will be called monotone, if the strategies of all bidders in it are monotone.
\end{definition}

Next, we present the notion of $\varepsilon$-well-supported MBNE of a DFPA, a more refined notion of equilibrium approximation which we will use in the results of this section. 

\begin{definition}[$\varepsilon$-well-supported mixed Bayes-Nash equilibrium of the DFPA]\label{def:well-supported-bayes-nash-equilibrium}
Let $\varepsilon \geq 0$.
A (mixed) strategy profile $\vec{\beta}=(\beta_1, \ldots, \beta_n)$ is an (interim) $\varepsilon$-well-supported mixed Bayes-Nash equilibrium (MBNE) of the DFPA if for any $i \in N$, any $v_i \in V_i$, and any $b \in \support{\beta_i}$:
\begin{equation}
\label{eq:WS-MBNE-def-condition-full}
u_i(b,\vec{\beta}_{-i};v_i) \geq u_i(\vec{\gamma},\vec{\beta}_{-i};v_i) - \varepsilon \qquad \text{for all}\;\; \vec{\gamma}\in \Delta(B).
\end{equation}
\end{definition}
Informally, in an $\varepsilon$-well-supported MBNE, the support of a strategy only contains bids that approximately (within $\varepsilon$) maximize the bidder's utility.\footnote{\textcite{chen2023complexity} use the term ``$\varepsilon$-well-supported'' to refer to \emph{interim} equilibria, see our \cref{def:approx-mixed-bayes-nash-equilibrium}. To avoid any misconceptions, we provide a short discussion in \cref{app:interim}.} 

Finally, we define a technical notion, which will be useful for us in the following:

\begin{definition}[Interaction degree of a DFPA]
A DFPA has \emph{interaction degree} bounded by $d \in \{1, \dots, n-1\}$, if for any bidder $i \in N$, there exists a set $J \subseteq N \setminus \{i\}$ of bidders with $|J| \leq d$ such that for any bidder $j \in N \setminus (J \cup \{i\})$ we have $0 \in V_j$ and $f_{i,j}(0) = 1$.
\end{definition}

Note that any DFPA has interaction degree bounded by $n-1$. Intuitively, the interaction degree provides a bound on how many bidders are perceived as always having zero value, from the perspective of some other bidder. The following lemma shows that, as long as the interaction degree is constant, $\varepsilon$-well-supported MBNE can be obtained in polynomial time from $\varepsilon'$-approximate MBNE of the DFPA, when both $\varepsilon$ and $\varepsilon'$ are constant.

\begin{lemma}\label{lem:NE2WSNE}
Consider any DFPA with interaction degree bounded by $d \in \{1, \dots, n-1\}$. Then, for any $\varepsilon \in [0,1]$, a (monotone) $\varepsilon$-well-supported MBNE can be obtained in polynomial time given any (monotone) $\varepsilon^2/(8d)$-approximate MBNE. Furthermore, if the auction is iid, then this also holds if we restrict our attention to symmetric equilibria.
\end{lemma}
\begin{proof}
Let $\vec{\beta}$ be a $\delta$-approximate MBNE of the auction, for some $\delta > 0$. Consider any $\gamma \geq \delta$, to be fixed later. For any $i \in N$ and any $v_i \in V_i$, let
$$G := \big\{b \in B: u_i(b,\vec{\beta}_{-i};v_i) \geq \max_{b' \in B} u_i(b',\vec{\beta}_{-i};v_i) - \gamma \big\}$$
be the set of bids that yield optimal utility up to $\gamma$. Then, we must have
\begin{equation}\label{eq:wsne-good-bids}
\beta_i(v_i)(B \setminus G) := \sum_{b \in B \setminus G} \beta_i(v_i)(b) < \delta/\gamma
\end{equation}
Indeed, if we had $\beta_i(v_i)(B \setminus G) \geq \delta/\gamma$, then
\begin{equation*}
\begin{split}
u_i(\beta_i(v_i),\vec{\beta}_{-i};v_i) &= \sum_{b \in B \setminus G} \beta_i(v_i)(b) \cdot u_i(b,\vec{\beta}_{-i};v_i) + \sum_{b \in G} \beta_i(v_i)(b) \cdot u_i(b,\vec{\beta}_{-i};v_i)\\
&< \beta_i(v_i)(B \setminus G) \cdot \left(\max_{b' \in B} u_i(b',\vec{\beta}_{-i};v_i) - \gamma \right) + \beta_i(v_i)(G) \cdot \max_{b' \in B} u_i(b',\vec{\beta}_{-i};v_i)\\
&= \max_{b' \in B} u_i(b',\vec{\beta}_{-i};v_i) - \gamma \cdot \beta_i(v_i)(B \setminus G)\\
&\leq \max_{b' \in B} u_i(b',\vec{\beta}_{-i};v_i) - \delta
\end{split}
\end{equation*}
contradicting the fact that $\vec{\beta}$ is a $\delta$-approximate MBNE.

Next, we construct a strategy profile $\vec{\beta}'$ from $\vec{\beta}$. For every $i \in N$ and every $v_i \in V_i$ we proceed as follows. For any $b \in B \setminus G$ we set
$$\beta_i'(v_i)(b) := 0$$
and for any $b \in G$ we set
$$\beta_i'(v_i)(b) := \frac{\beta_i(v_i)(b)}{\sum_{b' \in G} \beta_i(v_i)(b')}$$
where the denominator is strictly positive by \eqref{eq:wsne-good-bids}. Thus, $\beta_i'(v_i)$ is a well-defined distribution over $B$ that only puts positive mass on bids in $G$. Since $\beta_i(v_i)$ is non-overbidding, so is $\beta_i'(v_i)$. Furthermore, we have $\support{\beta_i'(v_i)} \subseteq \support{\beta_i(v_i)}$, and thus $\vec{\beta}'$ is monotone, if $\vec{\beta}$ is monotone. If the auction is iid and $\vec{\beta}$ is symmetric, then $\vec{\beta}'$ will also be symmetric.

Now, we have that for all $i \in N$ and $v_i \in V_i$, the total variation distance between $\beta_i(v_i)$ and $\beta_i'(v_i)$ satisfies
$$\textup{TV}\left(\beta_i(v_i), \beta_i'(v_i)\right) := \frac{1}{2} \|\beta_i(v_i) - \beta_i'(v_i)\|_1 \leq \frac{1}{2} \cdot 2 \cdot \beta_i(v_i)(B \setminus G) < \delta/\gamma$$
by \eqref{eq:wsne-good-bids}. Fix some $i \in N$. Let $D_j$ and $D_j'$ denote the distributions over $B$ obtained by drawing $v_j \in V_j$ according to $F_{i,j}$, and then $b$ according to $\beta_j(v_j)$, respectively $\beta_j'(v_j)$. We have
\begin{equation*}
\begin{split}
\textup{TV}(D_j,D_j') &= \frac{1}{2} \sum_{b \in B} \left| \sum_{v_j \in V_j} f_{i,j}(v_j) \left(\beta_j(v_j)(b) - \beta_j'(v_j)(b)\right) \right|\\
&\leq \sum_{v_j \in V_j} f_{i,j}(v_j) \frac{1}{2} \sum_{b \in B} \left| \beta_j(v_j)(b) - \beta_j'(v_j)(b) \right|\\
&\leq \sum_{v_j \in V_j} f_{i,j}(v_j) \delta/\gamma = \delta/\gamma.
\end{split}
\end{equation*}
Next, define the product distributions $D := \times_{j \in N \setminus \{i\}} D_j$ and $D' := \times_{j \in N \setminus \{i\}} D_j'$ over $B^{n-1}$. Since the auction has interaction degree bounded by $d$, there exists a set $J \subseteq N \setminus \{i\}$ with $|J| \leq d$ such that for all $j \in N \setminus (J \cup \{i\})$ we have $0 \in V_j$ and $f_{i,j}(0) = 1$. By the no-overbidding assumption, it must be that for all $j \in N \setminus (J \cup \{i\})$, $\beta_j(0)(0) = 1 = \beta_j'(0)(0)$, and thus $D_j = D_j'$. As a result
$$\textup{TV}(D,D') \leq \sum_{j \in N \setminus \{i\}} \textup{TV}(D_j,D_j') = \sum_{j \in J} \textup{TV}(D_j,D_j') \leq d \cdot \delta/\gamma.$$
Now, fix some bid $b \in B$. For any $k \in \{0,1,\dots,n-1\}$, let $E_k \subseteq B^{n-1}$ denote the event that exactly $k$ bids are equal to $b$ and the rest are strictly smaller. Recall the definition of the $H$ functions in
\cref{lemma:DFPA-utilities-efficient}.
Then, we can write
\begin{equation*}
\begin{split}
H_i(b,\vec{\beta}_{-i}) &= \sum_{k=0}^{n-1} \frac{1}{k+1} \probability_D[E_k]\\
&= \frac{1}{n} \sum_{k=0}^{n-1} \probability_D[E_k] + \sum_{\ell=1}^{n-1} \left(\frac{1}{\ell} - \frac{1}{\ell+1} \right) \sum_{k=0}^{\ell-1} \probability_D[E_k]\\
&= \frac{1}{n} \probability_{D}\left[\bigcup_{k=0}^{n-1} E_k\right] + \sum_{\ell=1}^{n-1} \left(\frac{1}{\ell} - \frac{1}{\ell+1} \right) \probability_{D}\left[\bigcup_{k=0}^{\ell-1} E_k\right]
\end{split}
\end{equation*}
where we used the fact that the events $E_k$ are all disjoint. We also obtain the analogous expression for $H_i(b,\vec{\beta}_{-i}')$ and $D'$. It follows that
\begin{equation*}
\begin{split}
& |H_i(b,\vec{\beta}_{-i}) - H_i(b,\vec{\beta}_{-i}')|\\
\leq \quad & \frac{1}{n} \left| \probability_{D}\left[\bigcup_{k=0}^{n-1} E_k\right] - \probability_{D'}\left[\bigcup_{k=0}^{n-1} E_k\right] \right| + \sum_{\ell=1}^{n-1} \left(\frac{1}{\ell} - \frac{1}{\ell+1} \right) \left| \probability_{D}\left[\bigcup_{k=0}^{\ell-1} E_k\right] - \probability_{D'}\left[\bigcup_{k=0}^{\ell-1} E_k\right] \right|\\
\leq \quad & \frac{1}{n} d\delta/\gamma + \sum_{\ell=1}^{n-1} \left(\frac{1}{\ell} - \frac{1}{\ell+1} \right) d\delta/\gamma = d\delta/\gamma
\end{split}
\end{equation*}
where we used $\textup{TV}(D,D') \leq d \delta/\gamma$.

As a result, for any $i \in N$, $v_i \in V_i$, and $b \in B$ we have
$$\left|u_i(b,\vec{\beta}_{-i};v_i) - u_i(b,\vec{\beta}_{-i}';v_i)\right| = \left| H_i(b,\vec{\beta}_{-i}) \cdot (v_i-b) - H_i(b,\vec{\beta}_{-i}') \cdot (v_i-b) \right| \leq d\delta/\gamma.$$
Finally, for any $i \in N$, $v_i \in V_i$, and $b \in \support{\beta_i'(v_i)} \subseteq G$ we have
$$u_i(b,\vec{\beta}_{-i};v_i) \geq \max_{b' \in B} u_i(b',\vec{\beta}_{-i};v_i) - \gamma$$
by construction of $G$, and as a result
$$u_i(b,\vec{\beta}_{-i}';v_i) \geq \max_{b' \in B} u_i(b',\vec{\beta}_{-i}';v_i) - \gamma - 2d\delta/\gamma.$$
Thus, $\vec{\beta}'$ is a $(\gamma + 2d\delta/\gamma)$-well supported MBNE. Setting $\gamma := \sqrt{2d\delta}$ and $\delta := \varepsilon^2/(8d)$ yields the result. Note that for $\varepsilon = 0$ the statement trivially holds.
\end{proof}

\subsection{Discrete and Continuous Auctions: A Computational Equivalence}\label{sec:equivalence}

We are now ready to prove our computational equivalence result. This is composed of \cref{lem:DFPA2CFPA,lem:CFPA2DFPA} below.

\begin{lemma}\label{lem:DFPA2CFPA}
Given $\delta \in (0,1)$ and a DFPA, we can construct in polynomial time a CFPA (with the same bidding space) such that for any $\varepsilon \geq 0$, we can transform in polynomial time any $\varepsilon$-PBNE of the CFPA into an $(\varepsilon + \delta)$-well-supported monotone MBNE of the DFPA. Furthermore, this reduction maps IPV (resp.\ iid) auctions to IPV (resp.\ iid) auctions, and symmetric equilibria to symmetric equilibria.
\end{lemma}

\begin{proof}
Let $\delta \in (0,1)$ be given. Consider a DFPA given by the value spaces $V_1, \dots, V_n$ and the distributions $F_{i,j}$. Without loss of generality, we can assume that $\delta$ satisfies
\begin{equation}\label{eq:wlog-delta}
2\delta \leq \min_{p,q \in V_1 \cup \dots V_n \cup B: p \neq q} |p-q|.
\end{equation}
Indeed, if this is not the case, we can replace $\delta$ by a smaller value in $(0,1)$ that satisfies this inequality.

We construct a CFPA with the same bidding space $B$ and with distributions $F_{i,j}'$. For every $i \neq j$, the piecewise constant density function of the distribution $F_{i,j}'$ is constructed from $F_{i,j}$ as follows.
\begin{itemize}
    \item[---] For every $v_j \in V_j$ with $v_j \neq 0$, we let the pdf have value $f_{i,j}(v_j)/\delta$ in interval $[v_j-\delta,v_j]$.
    \item[---] If $0 \in V_j$, we let the pdf have value $f_{i,j}(0)/\delta$ in interval $[0,\delta]$.
    \item[---] The pdf has value $0$ everywhere else in $[0,1]$.
\end{itemize}
For an illustration, see \cref{fig:disc-cont}.
This pdf is well-defined by \eqref{eq:wlog-delta} (also using the fact that $0 \in B$), and in particular the blocks of value do not overlap. Furthermore, any additional structure of the auction (namely IPV or iid) is kept intact.

Now consider any $\varepsilon$-PBNE $\vec{\beta}'$ of the CFPA. We construct a corresponding strategy profile $\vec{\beta}$ in the original DFPA as follows. For any $j \in N$, $\beta_j: V_j \to \Delta(B)$ is given by:
\begin{itemize}
    \item[---] For every $v_j \in V_j$ with $v_j \neq 0$, let $\beta_j(v_j)$ be the distribution of $\beta_j'(u)$ where $u$ is drawn uniformly at random in $[v_j-\delta,v_j]$. In other words, for all $b \in B$
    $$\beta_j(v_j)(b) = \lambda \big(\{u \in [v_j-\delta,v_j]: \beta_j'(u) = b\}\big)/\delta$$
    where $\lambda$ denotes the Lebesgue measure (in this case just the length of the interval).
    \item[---] If $0 \in V_j$, let $\beta_j(0)$ be the distribution of $\beta_j'(u)$ where $u$ is drawn uniformly at random in $[0,\delta]$.
\end{itemize}
Recall that $\beta_j': [0,1] \to B$ is a non-overbidding non-decreasing step-function strategy. As a result, and also by \eqref{eq:wlog-delta}, the mixed strategy $\beta_j$ will also be non-overbidding. Furthermore, the fact that $\beta_j'$ is non-decreasing immediately yields that $\beta_j$ is monotone. Finally, if $V_i = V_{i'}$ for all $i,i'$ and $\vec{\beta}'$ is symmetric, then so is $\vec{\beta}$.

Next, we note that the construction of $F_{i,j}'$ from $F_{i,j}$ and the construction of $\vec{\beta}$ from $\vec{\beta}'$ have been carefully devised to ensure that for every $i \neq j$ the two following distributions over $B$ coincide
\begin{itemize}
    \item[---] Pick $v_j \in V_j$ according to $F_{i,j}$, and then output $b \in B$ according to $\beta_j(v_j)$.
    \item[---] Pick $v \in [0,1]$ according to $F_{i,j}'$ and output $b = \beta_j'(v)$.
\end{itemize}

As a result, we have that
\begin{equation}\label{eq:discr2cont}
u_i(b,\vec{\beta}_{-i};v_i) = u_i'(b,\vec{\beta}_{-i}';v_i)
\end{equation}
for all $i \in N$, $b \in B$, and $v_i \in V_i$, where $u_i, u_i'$ denote the utilities in the DFPA and CFPA respectively.

It remains to prove that $\vec{\beta}$ is an $(\varepsilon + \delta)$-well-supported monotone MBNE of the DFPA. For this, it suffices to show that for any $i \in N$ and $v_i \in V_i$, we have that all $b \in B$ with $\beta_i(v_i)(b) > 0$ satisfy $u_i(b,\vec{\beta}_{-i};v_i) \geq u_i(b',\vec{\beta}_{-i};v_i) - \varepsilon - \delta$ for all $b' \in B$. Consider any such $i, v_i, b$ and note that by construction of $\beta_i$ from $\beta_i'$ there exists $v \in [0,1]$ with $|v-v_i| \leq \delta$ such that $\beta_i'(v) = b$. Since $\vec{\beta}'$ is an $\varepsilon$-PBNE $\vec{\beta}'$ of the CFPA, it follows that for all $b' \in B$
\begin{alignat*}{2}
&& u_i'(b,\vec{\beta}_{-i}';v) &\geq u_i'(b',\vec{\beta}_{-i}';v) - \varepsilon\\
\implies \quad && H_i(b,\vec{\beta}_{-i}') \cdot (v-b) &\geq H_i(b',\vec{\beta}_{-i}') \cdot (v-b') - \varepsilon\\
\implies \quad && H_i(b,\vec{\beta}_{-i}') \cdot (v_i-b) &\geq H_i(b',\vec{\beta}_{-i}') \cdot (v_i-b') - \varepsilon + \big(H_i(b',\vec{\beta}_{-i}') - H_i(b,\vec{\beta}_{-i}')\big) \cdot (v-v_i)\\
\implies \quad && u_i'(b,\vec{\beta}_{-i}';v_i) &\geq u_i'(b',\vec{\beta}_{-i}';v_i) - \varepsilon - \delta\\
\implies \quad && u_i(b,\vec{\beta}_{-i};v_i) &\geq u_i(b',\vec{\beta}_{-i};v_i) - \varepsilon - \delta
\end{alignat*}
where we used \eqref{eq:discr2cont}, $|v-v_i| \leq \delta$, and the terms $H_i(b,\vec{\beta}_{-i}') \in [0,1]$ which denote the probability for bidder $i$ to win the item when bidding $b$, while the other bidders follow the strategy profile $\vec{\beta}_{-i}'$ in the CFPA. This completes the proof.
\end{proof}

\begin{figure}
  \begin{minipage}[b]{0.4\linewidth}
    \centering
    \resizebox{200pt}{150pt}
    {\begin{tikzpicture}

        \draw[->] (0,0) -- (11.5,0);
        \draw[->] (0,0) -- (0,12);
        
        \draw[line width=0.95mm, blue] (0,0) -- (0, 2);
        \fill[color=blue] (0,2) circle[radius=3pt];
        \node[below] at (0, -0.15) {\huge $v_0$};
        \node[left] at (0, 2) {\huge $f(v_{0})$};
        \draw[olive] (1,0) -- (1, 2 / 0.8);
        \draw[olive] (1, 2 / 0.8) -- (0, 2 / 0.8);
        \node[above] at (0.5, 2 /0.8) {\Large $\delta$};
        \draw[->] (0.65, 2/0.8 + 0.25) -- (1 ,2/0.8 + 0.25);
        \draw[->] (0.4, 2/0.8 + 0.25) -- (0.1 ,2/0.8 + 0.25);
        \draw [olive] (0, 2 / 0.8) -- (0, 2);
        \node[right] at (0.9, 2 / 1.6) {\large $f(v_0)\delta^{-1}$};
        \draw[->] (1.5, 2 / 1.6 + 0.4) -- (1.5, 2 / 0.8 -0.2);
        \draw[->] (1.5, 2 / 1.6 - 0.4) -- (1.5, 0.2);
        \fill[pattern=north west lines, pattern color=olive] (0,0) rectangle (1, 2 / 0.8);
        
        \foreach \x/\y/\z in {4/4/1, 7/8/2, 11/6/3}
        {
        \draw[line width=0.95mm, blue] (\x,0) -- (\x, \y);
        \fill[color=blue] (\x,\y) circle[radius=3pt];
        \node[below] at (\x, -0.15) {\huge $v_{\z}$};
        \node[left] at (0, \y) {\huge $f(v_{\z})$};
        \draw[olive] (\x-1,0) -- (\x-1, \y / 0.8);
        \draw[olive] (\x-1, \y / 0.8) -- (\x, \y / 0.8);
        \node[above] at (\x-0.5, \y /0.8) {\Large $\delta$};
        \draw[->] (\x-0.35, \y /0.8 + 0.25) -- (\x - 0.1 ,\y /0.8 + 0.25);
        \draw[->] (\x-0.6, \y /0.8 + 0.25) -- (\x - 0.9 ,\y /0.8 + 0.25);
        \node[left] at (\x-1+0.08, \y / 1.6) {\large $f(v_{\z}) \delta^{-1}$};
        \draw[->] (\x -1.5, \y / 1.6 + 0.4) -- (\x -1.5, \y / 0.8 -0.2);
        \draw[->] (\x -1.5, \y / 1.6 - 0.4) -- (\x -1.5, 0.2);
        \draw [olive] (\x, \y / 0.8) -- (\x, \y);
        \fill[pattern=north west lines, pattern color=olive] (\x,0) rectangle (\x -1, \y / 0.8);
        }
        \node[left] at (0,11.8) {\huge $f(v)$};
        \end{tikzpicture}}
    \caption{Discrete $\to$ Continuous}
    \label{fig:disc-cont}
  \end{minipage}
  \qquad
  \begin{minipage}[b]{0.5\linewidth}
    \centering
    \resizebox{200pt}{150pt}{\begin{tikzpicture}

        \draw[->] (0,0) -- (12,0);
        \draw[->] (0,0) -- (0,12);
        
        \def\intervals{0, 0.5, 1, 1.5, 2, 2.5, 3, 4, 4.5, 5, 6, 7, 7.3, 8, 9, 9.15, 10, 11}
        \def\sizes{{1, 2, 0.5}}
        
        \foreach \x [count=\xi] in \intervals {
            \draw (\x,0.1) -- (\x,-0.1);
        }
        \definecolor{blizzardblue}{rgb}{0.67, 0.9, 0.93}
        
        \draw[color=blue] (0,0) rectangle (2.5,1.5);
        \fill[pattern=north west lines, pattern color=blizzardblue] (0,0) rectangle (2.5,1.5);
        \draw[color=blue] (2.5,0) rectangle (7.3,4);
        \fill[pattern=north west lines, pattern color=blizzardblue] (2.5,0) rectangle (7.3,4);
        \draw[color=blue] (7.3,0) rectangle (9.15,3);
        \fill[pattern=north west lines, pattern color=blizzardblue] (7.3,0) rectangle (9.15,3);
        \draw[color=blue] (9.15,0) rectangle (11,1);
        \fill[pattern=north west lines, pattern color=blizzardblue] (9.15,0) rectangle (11,1);
        
        \node[below] at (0,-0.1) {\huge $v_0$};
        \node[below] at (11,-0.1) {\huge $1$};
        
        \draw[color=red] (0,0) -- (0,0.4);
        \fill[color=red] (0,0.4) circle[radius=2pt];
        \draw[color=red] (0.5,0) -- (0.5,0.4);
        \fill[color=red] (0.5,0.4) circle[radius=2pt];
        \draw[color=red] (1,0) -- (1,0.4);
        \fill[color=red] (1,0.4) circle[radius=2pt];
        \draw[color=red] (1.5,0) -- (1.5,0.4);
        \fill[color=red] (1.5,0.4) circle[radius=2pt];
        \draw[color=red] (2,0) -- (2,0.4);
        \fill[color=red] (2,0.4) circle[radius=2pt];
        
        \draw[color=red] (2.5,0) -- (2.5,1.7);
        \fill[color=red] (2.5,1.7) circle[radius=2pt];
        \draw[color=red] (3,0) -- (3,3.4);
        \fill[color=red] (3,3.4) circle[radius=2pt];
        \draw[color=red] (4,0) -- (4,1.7);
        \fill[color=red] (4,1.7) circle[radius=2pt];
        \draw[color=red] (4.5,0) -- (4.5,1.7);
        \fill[color=red] (4.5,1.7) circle[radius=2pt];
        \draw[color=red] (5,0) -- (5,3.4);
        \fill[color=red] (5,3.4) circle[radius=2pt];
        \draw[color=red] (6,0) -- (6,3.4);
        \fill[color=red] (6,3.4) circle[radius=2pt];
        \draw[color=red] (7,0) -- (7,1.7*0.6);
        \fill[color=red] (7,1.7*0.6) circle[radius=2pt];
        \draw[color=red] (7.3,0) -- (7.3,2.5*0.7);
        \fill[color=red] (7.3,2.5*0.7) circle[radius=2pt];
        \draw[color=red] (8,0) -- (8,2.5);
        \fill[color=red] (8,2.5) circle[radius=2pt];
        \draw[color=red] (9,0) -- (9,2.5*0.15);
        \fill[color=red] (9,2.5*0.15) circle[radius=2pt];
        \draw[color=red] (9.15,0) -- (9.15,0.8*0.85);
        \fill[color=red] (9.15,0.8*0.85) circle[radius=2pt];
        \draw[color=red] (10,0) -- (10,0.8);
        \fill[color=red] (10,0.8) circle[radius=2pt];
        
        \node[above] at (3.5,0) {\Large $\delta$};
        \draw[->] (3.65,0.25) -- (3.9,0.25);
        \draw[->] (3.4,0.25) -- (3.1,0.25);
        
        \draw[->] (0,0) -- (12,0);

        \node[below] at (11.8,-0.1) {\huge $v$};
        \node[left] at (-0.1,11.8) {\huge $f(v)$};
        \end{tikzpicture}}
    \caption{Continuous $\to$ Discrete}
    \label{fig:cont-disc}
  \end{minipage}
\end{figure}

\begin{lemma}\label{lem:CFPA2DFPA}
Given $\delta \in (0,1)$ (in unary) and a CFPA, we can construct in polynomial time a DFPA (with the same bidding space) such that for any $\varepsilon \geq 0$, we can transform in polynomial time any $\varepsilon$-well-supported \textbf{monotone} MBNE of the DFPA into an $(\varepsilon + \delta)$-PBNE of the CFPA. Furthermore, this reduction maps IPV (resp.\ iid) auctions to IPV (resp.\ iid) auctions, and symmetric equilibria to symmetric equilibria.
\end{lemma}
\begin{proof}
    Let $\delta \in (0,1)$ be given in unary. Consider a CFPA given by continuous distributions $F_{i,j}'$ with piecewise constant density functions. Denote by $S_{i,j}$ the smallest set of points in $[0,1]$ such that $\{0, 1\} \subseteq S_{i,j}$ and the density function of $F_{i,j}'$ is constant between any two adjacent points in $S_{i,j}$. Let $S := \cup_{i,j \in N: i \neq j} S_{i,j}$, and note that $|S|$ is polynomial in the size of the CFPA instance. Now construct a set $V \subset [0,1]$ of polynomial size such that $S \subseteq V$ and any two adjacent points in $V$ are separated by distance at most $\delta$. Since $\delta \in (0,1)$ is given in unary, this is easily achieved, e.g., by taking the union of $S$ and a $\delta$-grid on $[0,1]$.

    We construct a DFPA with the same bidding space $B$ and with value spaces $V_i := V \setminus \{1\}$ for all $i \in N$. For each $i \neq j$, the discrete distribution $F_{i,j}$ over $V_j$ is constructed from $F_{i,j}'$ as follows. For any $v_j \in V_j$, let $v_j^+ \in V$ denote the next largest element in $V$. Note that this is well-defined since $V = V_j \cup \{1\}$. Now we define $F_{i,j}$ to be the distribution that draws $v_j \in V_j$ with probability equal to the probability that $v \in [v_j,v_j^+]$ when $v \in [0,1]$ is drawn according to $F_{i,j}'$. For an illustration of the induced distributions, see \cref{fig:cont-disc}. This indeed defines a valid distribution on $V_j$, since $0 \in V_j$ and $1 \in V$, and so the intervals $[v_j, v_j^+]$, $v_j \in V_j$, form a partition of $[0,1]$ (ignoring endpoints). Note that the additional structure of the auction (namely IPV or iid) is kept intact.
    
    Now consider any $\varepsilon$-well-supported monotone MBNE $\vec{\beta}$ of the DFPA. We construct a corresponding strategy profile $\vec{\beta}'$ in the original CFPA. For any $j \in N$, $\beta_j': [0,1] \to B$ is a step-function constructed as follows. We describe $\beta_j'$ separately in each interval $[v_j,v_j^+]$, $v_j \in V_j$. For each $b \in \support{\beta_j(v_j)}$, there is a segment of height $b$ and length $\beta_j(v_j)(b) \cdot |v_j^+-v_j|$. The segments are ordered inside $[v_j,v_j^+]$ in order of increasing $b$. This yields a non-decreasing step-function inside $[v_j,v_j^+]$. Since the intervals $[v_j,v_j^+]$ form a partition of $[0,1]$, this defines $\beta_j'$ as a step-function\footnote{In the case where $\beta_j'$ is assigned two different values at a common endpoint between two adjacent intervals, we can just pick the smallest one.} over $[0,1]$. Since $\beta_j$ was assumed to be monotone, it follows that $\beta_j'$ is a non-decreasing step-function over all of $[0,1]$. Furthermore, since $\beta_j$ is non-overbidding, we have that $\max \support{\beta_j(v_j)} \leq v_j$ and thus $\max \beta_j'([v_j,v_j^+]) \leq v_j$, which means that $\beta_j'$ is also non-overbidding. Finally, note that if $\vec{\beta}$ is symmetric, then so is $\vec{\beta}'$.
    
    The construction of $F_{i,j}$ from $F_{i,j}'$, and the construction of $\vec{\beta}'$ from $\vec{\beta}$ have been carefully devised to ensure that for every $i \neq j$ the two following distributions over $B$ coincide
    \begin{itemize}
        \item[---] Pick $v_j \in V_j$ according to $F_{i,j}$, and then output $b \in B$ according to $\beta_j(v_j)$.
        \item[---] Pick $v \in [0,1]$ according to $F_{i,j}'$ and output $b = \beta_j'(v)$.
    \end{itemize}
    
    As a result, we have that
    \begin{equation}\label{eq:discr2cont-bis}
    u_i(b,\vec{\beta}_{-i};v_i) = u_i'(b,\vec{\beta}_{-i}';v_i)
    \end{equation}
    for all $i \in N$, $b \in B$, and $v_i \in V_i$, where $u_i, u_i'$ denote the utilities in the DFPA and CFPA respectively.
    
    It remains to prove that $\vec{\beta}'$ is an $(\varepsilon + \delta)$-PBNE of the CFPA. Consider any $i \in N$ and any $v \in [0,1]$, and let $b = \beta_i'(v)$. We wish to show that $u_i'(b,\vec{\beta}_{-i}';v) \geq u_i'(b',\vec{\beta}_{-i}';v) - \varepsilon - \delta$ for all $b' \in B$. By construction of $\beta_i'$ from $\beta_i$, there exists $v_i \in V_i$ such that $v \in [v_i,v_i^+]$ and $b \in \support{\beta_i(v_i)}$. Furthermore, by construction of $V_i$ we also have $|v-v_i| \leq \delta$. Now, since $\vec{\beta}$ is an $\varepsilon$-well-supported MBNE of the DFPA, we have that for all $b' \in B$
    \begin{alignat*}{2}
    && u_i(b,\vec{\beta}_{-i};v_i) &\geq u_i(b',\vec{\beta}_{-i};v_i) - \varepsilon\\
    \implies \quad && u_i'(b,\vec{\beta}_{-i}';v_i) &\geq u_i'(b',\vec{\beta}_{-i}';v_i) - \varepsilon\\
    \implies \quad && H_i(b,\vec{\beta}_{-i}') \cdot (v_i-b) &\geq H_i(b',\vec{\beta}_{-i}') \cdot (v_i-b') - \varepsilon\\
    \implies \quad && H_i(b,\vec{\beta}_{-i}') \cdot (v-b) &\geq H_i(b',\vec{\beta}_{-i}') \cdot (v-b') - \varepsilon + \big(H_i(b',\vec{\beta}_{-i}') - H_i(b,\vec{\beta}_{-i}')\big) \cdot (v_i-v)\\
    \implies \quad && u_i'(b,\vec{\beta}_{-i}';v) &\geq u_i'(b',\vec{\beta}_{-i}';v) - \varepsilon - \delta
    \end{alignat*}
    where we used \eqref{eq:discr2cont-bis}, $|v-v_i| \leq \delta$, and the terms $H_i(b,\vec{\beta}_{-i}') \in [0,1]$ which denote the probability for bidder $i$ to win the item when bidding $b$, while the other bidders follow the strategy profile $\vec{\beta}_{-i}'$ in the CFPA. This completes the proof.
    
\end{proof}

\subsection{PPAD-completeness}
\label{sec:PPAD-completeness-mixed}

Next, we provide the proof of \cref{th:PPAD-completeness-informal}, i.e., the PPAD-completeness of computing approximate MBNE of the DFPA with subjective priors. The PPAD-membership follows directly from our equivalence result and the PPAD-membership of pure PBNE of the CFPA proven by \textcite{fghlp2021_sicomp}. Our result does not only establish that computing $\varepsilon$-approximate equilibria is in PPAD, but that the membership holds even for $\varepsilon$-well-supported MBNE (\cref{def:well-supported-bayes-nash-equilibrium}) which are also monotone (\cref{def:monotonicty-discrete}). The formal statement of our membership lemma is below:

\begin{lemma}\label{lem:PPAD-membership}
The problem of computing a monotone $\varepsilon$-well-supported MBNE of the DFPA with subjective priors is in PPAD.
\end{lemma}

\begin{proof}
First, we observe that the problem lies in the class TFNP of total search problems that have efficiently verifiable solutions \parencite{megiddo1991total}. This is because by \cref{lemma:DFPA-utilities-efficient}, we can compute the utility of a given strategy in polynomial time. Then, it suffices to check if a given strategy is utility-maximizing against all possible pure deviations, see also \cref{note:MBNE-def-pure-deviation}. For the PPAD-membership, we apply \cref{lem:DFPA2CFPA}, which reduces the problem of finding a monotone $\varepsilon$-well-supported MBNE of the DFPA to that of finding a (monotone) $\varepsilon$-approximate PBNE of the CFPA. The latter problem is in PPAD by \parencite[Theorem 4.4.]{fghlp2021_sicomp}. 
\end{proof}

We now turn our attention to the PPAD-hardness of the problem. A first observation is that our discrete-continuous equivalence result can already yield \emph{some} PPAD-hardness result for the DFPA: starting from the PPAD-hardness result of \textcite[Theorem 5.1.]{fghlp2021_sicomp}, one could use \cref{lem:CFPA2DFPA} to reduce from the continuous to the discrete version. However, that (coupled with \cref{lem:NE2WSNE}) would result in the PPAD-hardness of computing \emph{monotone} $\varepsilon$-approximate MBNE. Monotonicity is a desirable property for positive results (as it makes the results stronger), but we would ideally like to prove hardness for \emph{any} type of MBNE, monotone or not. Another caveat of the equivalence approach is that the PPAD-hardness result would require instances with relatively large bidding spaces. 

To avoid these shortcomings, we provide a direct PPAD-hardness result for the problem, via a reduction from the recently proposed PPAD-complete problem \textsc{PureCircuit} \parencite{deligkas2022pure}. Our reduction shows computational hardness for $\varepsilon$-approximate equilibria that need not be monotone, and for instances with very small bidding spaces, thus making the hardness result stronger. We remark that the proof of \cref{lem:ppad-hardness} below first establishes PPAD-hardness for $\varepsilon$-well-supported MBNE on instances with small interaction degree, and then invokes \cref{lem:NE2WSNE} to strengthen the result to the PPAD-hardness of $\varepsilon$-approximate MBNE. The proof is included in \cref{app:ppad-hardness}.

\begin{lemma}\label{lem:ppad-hardness}
There exists a constant $\varepsilon > 0$ such that computing an $\varepsilon$-approximate MBNE of the DFPA with subjective priors is PPAD-hard.
\end{lemma}

\section{Symmetric Equilibria in IID Auctions} \label{sec:symmetric}

Motivated by the computational hardness results in the previous sections, we now consider the natural special case of \emph{iid priors}, and the computation of \emph{symmetric mixed equilibria} in such auctions. Interestingly, positive results for this case were not known before, even for continuous auctions. We show that symmetric MBNE always exist in a DFPA and we provide a PTAS for their computation. 

Before we get to this result, we present the following useful lemma. We will make use of the lemma for iid priors, but it holds more generally for subjective priors. The lemma shows that one can, in polynomial time, translate an $\varepsilon$-approximate MBNE of the DFPA in a \emph{shrunk} bidding space (of smaller size), to an $\varepsilon'$-approximate equilibrium of the auction on the original bidding space. A process similar in nature was presented in \parencite[Proof of Thm 1.2, Step 1]{chen2023complexity} for the case of the CFPA in the IPV model, and \emph{ex-ante} \emph{pure} equilibria.\footnote{See \cref{app:interim} for the definition of the notion.} Our proof holds for subjective priors and interim mixed equilibria, and is simpler in nature. 

\begin{lemma}[Bidding Space Shrinkage Lemma]
    \label{lemma:shrinkage}
    Consider a DFPA with bidding space $B$ and let $M$ be a positive integer. We can construct a bidding space $B'\subseteq B$ with cardinality $\card{B'}\leq M$, in time polynomial in $M$ and the size of the input such that any $\varepsilon$-approximate MBNE of the auction restricted to the bidding space $B'$ is a $\left(\varepsilon+\frac{1}{M}\right)$-approximate MBNE in the original auction.
\end{lemma}
\begin{proof}
    Given a bidding space $B$ and a positive integer $M$, we construct the restricted bidding space as follows:
    \begin{equation*}
        \tilde{B}= \bigunion_{\ell = 0}^{M-1} \min \left\{ b \in B \fwh{\ell \frac{1}{M} \leq b \leq (\ell + 1) \frac{1}{M}}  \right\},
    \end{equation*}
    where we allow a slight abuse of notation, such that the minimum over an empty set is defined to be the empty set.
    To better understand how this smaller bidding space can approximate the original, it is useful to define the following notion of distance between sets, which is an instantiation of the Hausdorff distance:
    \begin{equation*}
        d(B,B') \coloneqq \max_{b \in B} \min_{b' \in B'} |b-b'| 
    \end{equation*}

    By the definition of $\tilde{B}$, we can see that in our case this is bounded by the size of each interval, hence:
    \begin{equation}
    \label{eq:set-distance-bound}
    d(B,\tilde{B})\leq \frac{1}{M}. 
    \end{equation}

    Let $\mathcal A$ be the original auction (with bidding space $B$) and $\tilde{\mathcal A}$ be the one corresponding
    to the new bidding space $\tilde{B}$. Additionally, let $\tilde{\beta}$ be an $\varepsilon$-approximate MBNE of $\tilde{\mathcal{A}}$. First, we show that for any (mixed) strategy profile $\tilde{\vec{\beta}}$ of $\tilde{B}$, any bidder $i$ with a value $v_i$, and any bid $b\in B$, there has to exist a $\tilde b\in \tilde{B}$ such that: 
    \begin{equation*}
        \label{eq:shrinkage-utility-distance-bound}
        \card{u_i(\tilde{b}, \tilde{\vec{\beta}}_{-i};v_i) - u_i(b; \tilde{\vec{\beta}}_{-i};v_i)} \leq d(B,\tilde{B}).
    \end{equation*}

    To see why this is true, we will bound the left-hand side in the following way:
    \begin{align}
        \label{eq: shrinkage-pure-bound}
        \notag \card{u_i(\tilde{b}, \tilde{\vec{\beta}}_{-i};v_i) - u_i(b; \tilde{\vec{\beta}}_{-i};v_i)} &= \card{(v_i-\tilde{b})H(\tilde{b},\tilde{\vec{\beta}}_{-i}) - (v_i-b) H(b, \tilde{\vec{\beta}}_{-i})} \\ 
        \notag &\leq \card{b-\tilde{b}}\max \{ H(b,\tilde{\vec{\beta}}_{-i}),  H(\tilde{b},\tilde{\vec{\beta}}_{-i})\} \\ 
        &\leq \card{b-\tilde{b}} \leq d(B, \tilde{B})
    \end{align}
    where we have used the definition of the utility function~\eqref{eq:H-function} and we have bounded the function $H$ trivially by $1$ (since it is a probability).

    Next, assume that $\tilde{\vec{\beta}}$ is an $\varepsilon$-approximate MBNE of $\tilde{\mathcal{A}}$. Then, by \cref{note:MBNE-def-pure-deviation}, it should be the case that, for all $i \in N, v_i \in V_i, \tilde{b} \in \tilde{B}$, the following holds:
    \begin{equation}
        \label{eq:shrinkage-MBNE-condition}
        u_i(\tilde{\vec{\beta}};v_i) \geq u_i(\tilde{b}, \tilde{\vec{\beta}}_{-i};v_i) - \varepsilon 
    \end{equation}

    Using the bounds derived in \cref{eq:set-distance-bound,eq: shrinkage-pure-bound}, we can substitute into \eqref{eq:shrinkage-MBNE-condition} to obtain:
    \begin{equation*}
        \label{eq:shrinkage-to-show}
        u_i(\tilde{\vec{\beta}};v_i) \geq u_i(b, \tilde{\vec{\beta}}_{-i};v_i) - \varepsilon - \frac{1}{M}
        \qquad \forall i \in N,\; \forall v_i \in V_i, \forall b \in B,
    \end{equation*}
    and this is precisely (due to \cref{note:MBNE-def-pure-deviation}) the sufficient condition for $\tilde{\vec{\beta}}$ to be an $(\varepsilon + \frac{1}{M})$-approximate MBNE of $\mathcal{A}$.

\end{proof}

We now present our positive results for the DFPA with iid priors and the computation of symmetric MBNE. We start with the following theorem which establishes the existence of symmetric and monotone equilibria for the auction in this case. 

\begin{theorem}
\label{th:existence-symmetric-equilibria-iid}
In every DFPA with iid priors there exists a symmetric and monotone (exact) MBNE.
\end{theorem}

Our proof will go via the variant of the auction with continuous values, again appealing to our computational equivalence result in \cref{sec:equivalence}. In particular, we will first show, employing Kakutani's fixed point theorem \parencite{kakutani1941generalization}, that the CFPA always admits an (exact) monotone PBNE. By \cref{lem:DFPA2CFPA} this will imply the existence of a monotone $\varepsilon$-approximate MBNE for the DFPA. 
Finally, taking the limit of the sequence of approximate equilibria as $\varepsilon \rightarrow 0$, we will obtain the existence of an exact monotone equilibrium. 

We remark that, to the best of our knowledge, general existence results for symmetric equilibria and iid priors have not been proven for either the DFPA or the CFPA. There is a plethora of results for symmetric equilibria of auctions where both the values and the bids are continuous (e.g., see \parencite{riley1981optimal, chawla2013auctions} and \parencite[Section 2.3]{krishna2009auction}), but those do not have implications for our setting. For the case of the CFPA, existence results have only been shown for special cases \parencite{chwe1989discrete}. The proof of \cref{th:existence-symmetric-equilibria-iid} is in \cref{sec:existence-symmetric-MBNE-iid-DFPA}.

We now proceed to present a PTAS (\emph{Polynomial-time approximation scheme}) for computing symmetric $\varepsilon$-approximate MBNE for iid priors; in our setting, this means an algorithm that on input the description of the DFPA and any $\varepsilon>0$, finds an $\varepsilon$-approximate monotone MBNE of the auction in time polynomial in the description of the auction and possibly exponential in $1/\varepsilon$. For any constant $\varepsilon>0$, this is a polynomial time algorithm for the problem. 

Our strategy for constructing this algorithm will be to express the equilibrium computation problem as a system of polynomial inequalities. By known results in the literature (e.g., see \cite{grigor1988solving}), such a system can be solved to any accuracy $\delta >0$ in time polynomial in $1/\log(\delta)$ and $(\kappa d)^{\beta^2}$, where $\kappa$ is the number of polynomials, $d$ is their maximum degree, and $\beta$ is the number of variables of the system. However, if we were to represent our equilibrium computation problem as such a system ``naively'', we would end up with a running time that is exponential in both $|V|$ (the common value space for all the bidders) and $|B|$. To avoid this, we exploit the monotonicity of the equilibrium to come up with a much more succinct representation of the strategies, which we refer to as \emph{support-representation}. This allows us to remove the exponential dependence on $|V|$, but $|B|$ still appears in the exponent of the running time. For that, we invoke \cref{lemma:shrinkage} to obtain a constant-sized bidding space, at the expense of some accuracy in the efficiency of our computed equilibrium. 

\begin{theorem}
    \label{th:symmetric-mixed-iid-algorithm}
    In every DFPA with $m$ bids and $n$ bidders having iid priors over $k$ values, a symmetric $\varepsilon$-approximate mixed equilibrium can be computed in time polynomial in $\log(1/\varepsilon)$ and $(nkm)^{m^2}$.
    \end{theorem}

\begin{proof}
Fix a DFPA with $N=[n]$ bidders, value space $V_i=V=\sset{v_1,v_2,\dots,v_k}$ (indexed by $j\in[k]$) for all bidders $i\in N$, bidding space $B=\sset{b_1,b_2,\dots,b_m}$ (indexed by $\ell\in[m]$), and iid \ priors with probability mass function $\sset{f(v_i)}_{i\in [k]}$.

\paragraph{\emph{Support-representation} of monotone strategies} Let $\beta_i:V\rightarrow \Delta(B)$ be a symmetric (mixed) strategy of some bidder $i$. Then, $\beta_i$ can be represented in the following, more succinct way, which we will call \emph{support-representation}.

We maintain a sequence of indices over the bidding set that keeps track of which values have an ``active'' support on the given bid. Formally, we define a 
\emph{monotone} sequence $\xi:[m]\rightarrow [k]$, such that $\xi(m)=k$, and helper functions
\begin{align*}
    J^{\text{mixed}} &\coloneqq \xi([m])\\
    J^{\text{pure}} &\coloneqq [k]\setminus J^{\text{mixed}}\\
    \ell^{-}(j) & \coloneqq \min\sset{\ell\in[m]\fwh{\xi(j)=\ell}}, &&\text{for}\;\; j\in J^{\text{mixed}}, \\
    \ell^{+}(j) & \coloneqq \min\sset{\max\sset{\ell\in[m]\fwh{\xi(j)=\ell}}+1,m}, &&\text{for}\;\; j\in J^{\text{mixed}},\\
    \ell^{-}(j) = \ell^{+}(j) = \ell(j) & \coloneqq \min\sset{\ell\in[m]\fwh{\xi(j)>\ell}}, &&\text{for}\;\; j\in J^{\text{pure}}, \\
    L(j) & \coloneqq [\ell^{-}(j),\ell^{+}(j)]\inters \N, &&\text{for}\;\; j\in [k].
\end{align*}
The semantics are that strategy $\beta_i$ maps value $v_j$ to a distribution
supported within $\ssets{b_\ell}_{\ell\in L(j)}$. In particular, under
$\beta_i$, bidder $i$ bids deterministically
$b_{\ell(j)}$ when having a value in $\ssets{v_j}_{j\in
J^{\text{pure}}}$.

\begin{example}
As an example of our representation, consider the sequence $\xi=(1,1,1,4,5,5,9)$ for a DFPA auction with $7$ bids $B=\sset{b_1,\dots,b_7}$ and $9$ values $V=\ssets{v_1,\dots,v_9}$. Then, $\xi$ corresponds to a monotone bidding rule $\beta:V\rightarrow\Delta(B)$ such that: 
\begin{itemize}
\item[---] $\beta(v_1)$ is supported within $\ssets{b_1,b_2,b_3,b_4}$
\item[---] $\beta(v_2)=\beta(v_3)= b_4$ deterministically
\item[---] $\beta(v_4)$ is supported within $\ssets{b_4,b_5}$
\item[---] $\beta(v_5)$ is supported within $\ssets{b_5,b_6,b_7}$
\item[---] $\beta(v_6)=\beta(v_7)= \beta(v_8)=\beta(v_9)=b_7$ deterministically.
\end{itemize}
\end{example}

To complement $\xi$, which is the ``structural'' component of our
representation, we also maintain a corresponding vector of probability
variables $\ssets{p_{j,\ell}}_{j\in J^{\text{mixed}},\ell\in L(j)}$, with the
semantics that $p_{j,\ell}$ is the probability that value $v_j$ assigns to
bid $b_{\ell}$ (under $\beta_i$). Notice that, for any given $\xi$, these variables are at most $2\card{B}=2k$ since $\card{J^{\text{mixed}}}=\card{\xi([m])}\leq k$ and $L(j)\inters L(j')\leq 1$ for any $j,j'\in J^{\text{mixed}}$ with $j\neq j'$. 

Finally, observe that if we have the aforementioned support-representation of a monotone strategy, we can easily extend it (in polynomial time) to the canonical representation of mixed strategies by defining the remaining probabilities (which are $0/1$ constants), in the following way
\begin{equation}
    \label{eq:iid-symmetric-probs-constants}
    p_{j,\ell}=
    \begin{cases}
        0, &\text{if}\;\; j\in J^{\text{mixed}} \;\land\; \ell\neq L(j),\\
        1, &\text{if}\;\; j\in J^{\text{pure}} \;\land\; \ell = \ell(j),\\
        0, &\text{if}\;\; j\in J^{\text{pure}} \;\land\; \ell \neq \ell(j).
    \end{cases}
\end{equation}

We now continue with describing our algorithm for computing an approximate symmetric equilibrium of our auction.

\paragraph{Guessing the support-representation structure}
Let $\vec{\beta}=(\beta,\beta,\dots,\beta)$ be an \emph{exact} MBNE of our
auction, where all bidders employ the same
monotone\footnote{See~\cref{def:monotonicty-discrete}.} (mixed) strategy
$\beta$. Recall that such an object is guaranteed to exist, due
to~\cref{th:existence-symmetric-equilibria-iid}.

The bidder strategies $\beta$ of such an equilibrium will have the same support-representation $\xi$. Since $\xi$ is a $m$-length sequence whose elements belong in $[k]$, we can exhaustively go over all such possible structural vectors that our equilibria can have in time $O(k^m)$.

\paragraph{The system of polynomial inequalities}
Once such a structure vector $\xi$ for the individual bidder strategies is fixed, the equilibrium $\vec\beta$ is fully determined by the probabilities $\ssets{p_{j,\ell}}$ of the support-representation. More precisely, $\vec\beta$ is an equilibrium if and only if the following constraints are satisfied:
\begin{itemize}
    \item[---] (Feasibility of strategies) For each $v_j\in V$, probabilities $\ssets{p_{j,\ell}}$ should give rise to a valid distribution over bids $b_\ell$, i.e.
    \begin{align}
        \label{eq:iid-symmetric-system-constraint-feasibility-probs-1}
        p_{j,\ell} &\geq 0 &&\forall j\in[k],\; \forall \ell\in[m] \\
        \label{eq:iid-symmetric-system-constraint-feasibility-probs-2}
        \sum_{\ell=1}^m p_{j,\ell} &=1 &&\forall j\in [k]
    \end{align}
    
    \item[---] (No overbidding) No bidder should bid higher that her true value with strictly positive probability, under strategy $\beta$:
    \begin{align}
        \label{eq:iid-symmetric-system-constraint-overbidding}
        p_{j,\ell}(v_j-b_\ell) &\geq 0 &&\forall j\in[k],\; \forall \ell\in[m].
    \end{align}

    \item[---] (No improvement) 
    For any possible true value, no bidder should have an incentive to unilaterally deviate from the (possibly randomized) bidding strategy dictated by $\beta$, to any alternative deterministic bid (see~\cref{def:approx-mixed-bayes-nash-equilibrium} for $\varepsilon=0$). As a matter of fact, since  $\vec{\beta}$ is symmetric, it is enough to consider condition~\eqref{eq:MBNE-def-condition-full} from the perspective of a single (arbitrary) bidder $i$. So, letting $u(\ell,j)\coloneqq u_i(b_\ell,\vec{\beta}_{-i};v_j)$, the equilibrium condition can be written as 
    \begin{align}
        \label{eq:iid-symmetric-system-constraint-equilibrium}
    \sum_{\ell'=1}^m p_{j,\ell'} u(\ell',j) &\geq u(\ell,j) &&\forall j\in[k],\; \forall \ell\in[m],
    \end{align}
    where, by using~\cref{eq:H-function,eq:def-utility-function-g,eq:def-utility-function-G,eq:inefficient-T} and the fact that the bidders are iid, we can express
    \begin{equation}
        \label{eq:utilities-expression-wrt-g-iid-symmetrc}
    u(\ell,j) = (v_j-b_\ell)\sum_{r=0}^{n-1}\frac{1}{r+1}\binom{n-1}{r}g_{\ell}^r G^{n-r-1}_{\ell},
    \end{equation}
    with 
    \begin{align}
        \label{eq:function-g-system-iid}
        g_{\ell} &= \sum_{j=1}^k p_{j,\ell} f(v_j)   &&\forall \ell\in[m]
    \end{align}
    denoting the probability that any given bidder $i'\neq i$ bids exactly $b_\ell$ (see~\eqref{eq:def-utility-function-g}) under profile $\vec{\beta}_{-i}$, and 
    \begin{align}
        \label{eq:function-G-system-iid}
    G_{\ell} = \sum_{\ell'=1}^{\ell-1} g_{\ell'} &&\forall \ell\in[m]
    \end{align}
    the probability that she bids strictly less than $b_\ell$ (see~\eqref{eq:def-utility-function-G}).
    Observe that in~\eqref{eq:utilities-expression-wrt-g-iid-symmetrc} the utility $u(\ell,j)$ is expressed as a polynomial, of degree at most $n-1$, over the variables $\sset{g_\ell,G_\ell}_{\ell \in [m]}$.
\end{itemize}

Let (P) denote the feasibility program, over variables $\ssets{p_{j,\ell},g_\ell,G_\ell}$, defined by putting together constraints \eqref{eq:iid-symmetric-system-constraint-feasibility-probs-1}, \eqref{eq:iid-symmetric-system-constraint-feasibility-probs-2}, \eqref{eq:iid-symmetric-system-constraint-overbidding}, \eqref{eq:iid-symmetric-system-constraint-equilibrium} via replacing the utilities by their expression in~\eqref{eq:utilities-expression-wrt-g-iid-symmetrc}, \eqref{eq:function-g-system-iid}, and~\eqref{eq:function-G-system-iid}. Then, any solution to (P) gives rise, via the support-representation structure $\xi$, to a symmetric, monotone, exact MBNE of our auction. 

Now observe that (P) is a system of $O(km)$ polynomial inequalities, of maximum degree $n$, over $O(m)$ variables; for the latter, recall that only (at most) $2m$ of the probability variables $p_{j,\ell}$ are ``real'' variables in (P), namely $\ssets{p_{j,\ell}}_{j\in J^{\text{mixed}},\ell\in L(j)}$, since all others are constants given by~\eqref{eq:iid-symmetric-probs-constants}. Therefore, an approximate solution to (P) can be found in time polynomial in $n$, $k$ and exponential in $m$. 
More precisely, let $\ssets{\tilde p_{j,\ell},\tilde g_\ell, \tilde G_\ell}$ be a set of values for the variables of (P) that are $\delta$-near, with respect to the maximum norm, to a feasible solution $\ssets{p_{j,\ell}^\ast,\tilde g_\ell^\ast, \tilde G_\ell^\ast}$ of (P). Such a solution can be found in time polynomial in $\log(1/\delta)$ and $(nkm)^{m^2}$ --- using, e.g., the results of \textcite[Remark, p. 38]{grigor1988solving}.

\paragraph{Rounding back the solution} We will now show how we can transform, in polynomial time, the probabilities $\tilde{\vec{p}}=\ssets{\tilde p_{j,\ell}}$ that we derived by (approximately) solving system (P), to another set of probabilities, let's denote them by $\hat{\vec{p}}=\ssets{\hat p_{j,\ell}}$, so that they give rise to an \emph{approximate} (mixed) equilibrium of our auction. That is, $\hat{\vec{p}}$ should satisfy constraints \eqref{eq:iid-symmetric-system-constraint-feasibility-probs-1}, \eqref{eq:iid-symmetric-system-constraint-feasibility-probs-2} and \eqref{eq:iid-symmetric-system-constraint-overbidding}, but we will allow~\eqref{eq:iid-symmetric-system-constraint-equilibrium} to be relaxed to 
\begin{align}
    \label{eq:iid-symmetric-system-constraint-equilibrium-approximate}
\sum_{\ell'=1}^m \hat p_{j,\ell'} u(\ell',j) &\geq u(\ell,j) - \varepsilon.
\end{align}

To that end, we define 
\begin{equation}
    \label{eq:idd-algorithm-rounding-normalization}
    \hat{p}_{j,\ell}:=\frac{T(\tilde{p}_{j,\ell})}{\sum_{\ell' = 1}^m T(\tilde{p}_{j,\ell'})},
\end{equation}
where operator $T$ is the following:
\begin{equation}
    \label{eq:idd-algorithm-rounding-truncation}
    T(x)= 
    \begin{cases}
        0, & \text{if}\;\; x \leq \delta, \\
        1, & \text{if}\;\; x \geq 1, \\
        x, & \text{otherwise}.
    \end{cases}
\end{equation}

From the fact that $T(x)\geq 0$ and the normalization in~\eqref{eq:idd-algorithm-rounding-normalization}, it is immediately apparent that our new probabilities $\hat{\vec{p}}$ indeed satisfy constraints \eqref{eq:iid-symmetric-system-constraint-feasibility-probs-1} and \eqref{eq:iid-symmetric-system-constraint-feasibility-probs-2}.

Additionally, consider how the probabilities $\vec{p}^*=\ssets{p^*_{j,\ell}}$ of
the \emph{exact} solution of (P) behave, depending on the value of
$(v_j-b_\ell)$ in \eqref{eq:iid-symmetric-system-constraint-overbidding}. If
$v_j-b_\ell \geq 0$ then, by the fact that all elements of $\hat{p}$ are
non-negative, \eqref{eq:iid-symmetric-system-constraint-overbidding} will be
satisfied. Otherwise, if $v_j-b_\ell \geq 0$  it should be the case that
$p^*_{j,\ell}=0$, and therefore it must be that $\tilde{p}_{j,\ell}\leq \delta$,
since $\norm{\vec{p}^*-\tilde{\vec{p}}}_{\infty} \leq \delta$. Therefore, due
to~\eqref{eq:idd-algorithm-rounding-truncation} it must be that
$\tilde{p}_{j,\ell}=0$, and
\eqref{eq:iid-symmetric-system-constraint-overbidding} is satisfied in this case
as well.

Finally, we need to establish~\eqref{eq:iid-symmetric-system-constraint-equilibrium-approximate}.
For a value $v_j\in V$ and bid $b_\ell\in B$ we will denote
\begin{equation}
    \label{eq:iid-symmetric-algorithm-rounding-5}
u_{j,\ell}(\vec{p}) = (v_j - b_\ell) \sum_{r=0}^{n-1}\frac{1}{r+1}\binom{n-1}{r} \left[g_\ell(\vec p)\right]^r \left[G_\ell(\vec p)\right]^{n-r-1},
\end{equation}
where
$$
g_{\ell}(\vec{p}) \coloneqq \sum_{j=1}^k p_{j,\ell} f(v_{j}),\quad\text{for}\; \ell\in[m]
\qquad\text{and}\qquad
G_{\ell}(\vec{p}) \coloneq \sum_{\ell'=1}^{\ell-1} g_{\ell}(\vec p).
$$
We will show that, for every $\varepsilon >0$, we can choose a $\delta = \Theta\left(\frac{\varepsilon}{mn}\right)$ such that  
\begin{align}
    \label{eq:iid-symmetric-algorithm-rounding-2}
\card{u_{j,\ell}(\hat{\vec{p}})- u_{j,\ell}(\vec{p}^*)} &\leq \frac{\varepsilon}{2}
&&\forall j\in[k],\; \forall  \ell\in[m].
\end{align}
Note that~\eqref{eq:iid-symmetric-algorithm-rounding-2} is enough in order to establish the desired 
validity of~\eqref{eq:iid-symmetric-system-constraint-equilibrium-approximate},
since $u(\vec{p})$ is simply a rewriting of the utility $u(\ell,j)$
(see~\eqref{eq:utilities-expression-wrt-g-iid-symmetrc}) directly as a function
of the probability variables $\vec{p}=\ssets{p_{j',\ell'}}$, and $\vec{p}^*$ is
a solution of (P) and thus
satisfies~\eqref{eq:iid-symmetric-system-constraint-equilibrium}.

We first bound the distance between $\hat{\vec{p}}$ and $\vec{p}^*$ in the following lemma:
\begin{lemma}
    \label{lemma:iid-symmetric-algorithm-rounding-1}
    If $\delta \leq \frac{1}{3m}$, then $\norm{\hat{\vec{p}}-\vec{p}^*}_\infty \leq 4\delta$. 
\end{lemma}
\begin{proof}
Since $p^*_{j,\ell}\in [0,1]$ for all $j\in[k]$, $\ell\in[m]$, and $\norm{\tilde{\vec{p}}-\vec{p}^*}_\infty \leq \delta$, we can deduce that
\begin{align*}
    -\delta &\leq \tilde{p}_{j,\ell} \leq 1+\delta
    && \forall  j\in[k],\; \forall \ell\in[m].
\end{align*}
Therefore, by the definition of operator $T$ in~\eqref{eq:idd-algorithm-rounding-truncation} we get that 
\begin{align}
    \label{eq:iid-symmetric-algorithm-rounding-helper-1}
\card{T(\tilde{p}_{j,\ell})-\tilde{p}_{j,\ell}} \leq \delta && \forall  j\in[k],\; \forall \ell\in[m].
\end{align}
Therefore
\begin{align}
    \card{\sum_{\ell=1}^m T(\tilde{p}_{j,\ell}) - 1} &= \card{\sum_{\ell=1}^m T(\tilde{p}_{j,\ell}) - \sum_{\ell=1}^m {p}^*_{j,\ell}}\notag\\
    &\leq \sum_{\ell=1}^m \card{ T(\tilde{p}_{j,\ell}) -  \tilde{p}_{j,\ell}}\notag\\
    &\leq \sum_{\ell=1}^m  \card{ T(\tilde{p}_{j,\ell}) -  \tilde{p}_{j,\ell}} + \sum_{\ell=1}^m \card{ p^*_{j,\ell} -  \tilde{p}_{j,\ell}} \notag\\
    &\leq m\cdot \delta + m\cdot \delta = 2m\delta \label{eq:iid-symmetric-algorithm-rounding-helper-2}
\end{align}

Now, for all $j\in[k]$ and $\ell\in[m]$ we can bound:
\begin{align}
    \card{\hat{p}_{j,\ell}- \tilde{p}_{j,\ell}}
    &\overset{\eqref{eq:idd-algorithm-rounding-normalization}}{=} \card{\frac{T(\tilde{p}_{j,\ell})}{\sum_{\ell' = 1}^m T(\tilde{p}_{j,\ell'})} - \tilde{p}_{j,\ell}}\\
    &\leq \max\sset{\frac{1}{\card{\sum_{\ell' = 1}^m T(\tilde{p}_{j,\ell'})}},1} \card{T(\tilde{p}_{j,\ell})-\tilde{p}_{j,\ell}}\\
    & \overset{\eqref{eq:iid-symmetric-algorithm-rounding-helper-1},\eqref{eq:iid-symmetric-algorithm-rounding-helper-2}}{\leq} \max\sset{\frac{1}{\card{1-2m\delta}},1} \delta\\
    &\leq 3\delta,
    \label{eq:iid-symmetric-algorithm-rounding-helper-3}
\end{align}
where for the last inequality we have used the fact that $0<\delta\leq \frac{1}{3m}$, and for the first one the fact that $\card{\lambda x -\mu y}\leq \max\ssets{\card{\lambda},\card{\mu}}\card{x-y}$ for all reals $\lambda,\mu, x, y$.

Finally, we can derive that:
$$
\norm{\hat{\vec{p}}-\vec{p}^*}_\infty\leq \norm{\tilde{\vec{p}}-\hat{\vec{p}}}_\infty + \norm{\tilde{\vec{p}}-\vec{p}^*}_\infty \overset{\eqref{eq:iid-symmetric-algorithm-rounding-helper-3}}{\leq} 3\delta + \delta =4\delta.
$$
\end{proof}

We are now ready to bound the distance of the utilities under the two profiles $\hat{\vec{p}}$ and $\vec{p}^*$ in~\eqref{eq:iid-symmetric-algorithm-rounding-2}.
Using that $\norm{\hat{\vec{p}}-\vec{p}^*}_\infty \leq 4\delta$ from~\cref{lemma:iid-symmetric-algorithm-rounding-1}, we can now provide the bounds:
\begin{align}
    \card{g(\hat{\vec{p}}) - g(\vec{p}^*)} &= \card{\sum_{j'=1}^k \hat p_{j',\ell} f(v_{j'}) - \sum_{j'=1}^k  p^*_{j',\ell} f(v_{j'})} 
    \leq \sum_{j'=1}^k \card{\hat p_{j',\ell} - p^*_{j',\ell}} f(v_{j'}) \leq 4\delta,\\
    \intertext{since $\sum_{j'=1}^kf(v_{j'})=1$, and so}
    \card{G(\hat{\vec{p}}) - G(\vec{p}^*)} & \leq \sum_{\ell'=1}^{\ell -1}\card{g(\hat{\vec{p}}) - g(\vec{p}^*)} 
    \leq 4\delta \card{\ell-1} \leq 4\delta m. \label{eq:iid-symmetric-system-constraint-feasibility-probs-6}
\end{align}
Also, note that for $\vec{p}\in\ssets{\hat{\vec{p}},\vec{p}^*}$, it is 
$$
G_\ell(\vec{p})=\sum_{\ell'=1}^{\ell-1} g_l(\vec{p}) \leq \sum_{\ell'=1}^{m} \sum_{j=1}^k p_{j,\ell'} f(v_j) = \sum_{j=1}^k\left(\sum_{\ell'=1}^{m}p_{j,\ell'}\right)f(v_j) = \sum_{j=1}^k f(v_j)=1,
$$
since both $\hat{\vec{p}}$ and $\vec{p}^*$ satisfy \eqref{eq:iid-symmetric-system-constraint-feasibility-probs-2}.

Next, we rewrite the utilities in~\eqref{eq:iid-symmetric-algorithm-rounding-5} as
\begin{align*}
    u_{j,\ell}(\vec{p}) &= \frac{v_j - b_\ell}{n \cdot g_\ell(\vec p)} \sum_{r=0}^{n-1}\binom{n}{r+1} \left[g_\ell(\vec p)\right]^{r+1} \left[G_\ell(\vec p)\right]^{n-r-1}\\
    &= \frac{v_j - b_\ell}{n \cdot g_\ell(\vec p)} \sum_{r=1}^{n}\binom{n}{r} \left[g_\ell(\vec p)\right]^{r} \left[G_\ell(\vec p)\right]^{n-r}\\
    &= \frac{v_j - b_\ell}{n \cdot g_\ell(\vec p)} \left[(G_{\ell}(\vec p)+g_{\ell}(\vec p))^n - G_{\ell}^n(\vec p)\right]\\
    &= \frac{v_j - b_\ell}{n \cdot g_\ell(\vec p)} \left[G^n_{\ell+1}(\vec p) - G^n_{\ell}(\vec p)\right]\\
    &= \frac{v_j - b_\ell}{n}\sum_{r=0}^{n-1}G^r_{\ell+1}(\vec p)G^{n-1-r}_{\ell}(\vec p).
\end{align*}
Using this, we can finally bound the distance of the utilities by
\begin{align*}
    \card{u_{j,\ell}(\hat{\vec{p}})- u_{j,\ell}(\vec{p}^*)} 
    &\leq \frac{1}{n}\sum_{r=0}^{n-1}\card{G^r_{\ell+1}(\hat{\vec p})G^{n-1-r}_{\ell}(\hat{\vec p})-G^r_{\ell+1}(\vec p^*)G^{n-1-r}_{\ell}(\vec p^*)}\\
    &\leq \max_{r=0,1,\dots,n-1}\card{G^r_{\ell+1}(\hat{\vec p})G^{n-1-r}_{\ell}(\hat{\vec p})-G^r_{\ell+1}(\vec p^*)G^{n-1-r}_{\ell}(\vec p^*)}\\
    &\leq \max_{r=0,1,\dots,n-1}\card{G^{r}_{\ell}(\hat{\vec p})-G^{r}_{\ell}({\vec p}^*)}\\
    &\leq \max_{r=0,1,\dots,n-1}r\card{G_{\ell}(\hat{\vec p})-G_{\ell}({\vec p}^*)}\\
    & \leq 4\delta m n,
\end{align*}
where the third and fourth inequalities hold due to the fact that $G_{\ell}(\hat{\vec p}),G_{\ell}(\vec p^*)\in[0,1]$ for all $\ell\in[m]$, and the last one due to~\eqref{eq:iid-symmetric-system-constraint-feasibility-probs-6}. 
Taking $\delta = \frac{\varepsilon}{8mn}$ we can indeed satisfy~\eqref{eq:iid-symmetric-algorithm-rounding-2}.
\end{proof}

We can now state our PTAS result.

\begin{theorem}
\label{th:ptas-symmetric-mixed-iid}
In every DFPA with iid priors, and any constant $\varepsilon>0$, a symmetric and monotone $\varepsilon$-approximate mixed equilibrium can be computed in polynomial time.
\end{theorem}

\begin{proof}
Fix an iid DFPA auction $\mathcal{A}$ with $n$ bidders, $k$ values, and bidding space $B$. Let $\varepsilon>0$ be an arbitrary constant. Choose integer $M\coloneqq \lceil 2/\varepsilon \rceil$, so that $\frac{1}{M} \leq \frac{\varepsilon}{2}$. Next, we construct a subspace $B'\subseteq B$ of our original bidding space $B$, with constant cardinality $\card{B'}=M$, as dictated in the proof of the shrinkage~\cref{lemma:shrinkage}. Notice that $B'$ can be constructed in polynomial time (in the description of $\mathcal A$).

Now, let $\mathcal A'$ denote the DFPA that results from $\mathcal A$ if we simply replace bidding space $B$ of $\mathcal A$ by $B'$ (everything else, namely bidders, value space and belief distributions, remaining the same). Obviously, $\mathcal A'$ is still an iid auction and thus, if we define $\varepsilon'\coloneqq \varepsilon/2$, we can deploy the algorithm presented in~\cref{th:symmetric-mixed-iid-algorithm} in order to compute a symmetric and monotone $\varepsilon'$-approximate MBNE of $\mathcal A'$, let's denote it by $\vec{\beta}^*$, in time polynomial in $\log(1/\varepsilon')=O(\log(1/\varepsilon))$ and in $(nkM)^M=\poly(n,k)$, since $M$ is a constant.
Due to~\cref{lemma:shrinkage}, again, we know that $\vec{\beta}^*$ is guaranteed to be an $(\varepsilon'+\frac{1}{M})$-approximate MBNE of the original auction $\mathcal A$. Our proof is concluded by noting that $\varepsilon'+\frac{1}{M} \leq \frac{\varepsilon}{2}+\frac{\varepsilon}{2}=\varepsilon$.  
\end{proof}

\section{Correlated Equilibria} \label{sec:correlated-equilibria}

Another well-studied solution concept in Game Theory is that of Correlated Equilibria (CE), first introduced by Aumann in his seminal paper \cite{aumann74}.
Intuitively, these capture the idea of each player choosing her action after being given a recommendation from a ``mediator'', ensuring that at a CE no player has an incentive to deviate from their recommended strategy. This is a broader solution concept which in several cases has proven to be amenable to polynomial-time algorithms, e.g., see \cite{Papadimitriou2008a,jiang-correlated}.
Given our PPAD-completeness result for computing MBNE in \cref{sec:PPAD-completeness-mixed}, it is natural to explore CE of the DFPA, in a quest for potential tractability results.  

Before we do that however, we have to come up with an appropriate definition: CE are typically defined in complete information, normal form games, whereas the DFPA is a Bayesian game of incomplete information. The appropriate extension of the notion of CE to Bayesian games has been discussed extensively in the literature \cite{aumann1987correlated, forges-5defs-ce, forges-revisited, forges2023correlated,bergemann2013robust,bergemann2016bayes}, and several natural notions have been proposed. In this paper, we will adopt the standard notion that defines a CE of the original DFPA by means of the induced \emph{type-agent} representation, adopted also by works in computer science, e.g., see \parencite{hartline2015no,ahunbay2024uniqueness}. 
This representation is a generic way of interpreting any Bayesian game as a traditional, strategic-form game; see, \cite[pp.~73--74,127--128]{myerson1997book} and~\cite[Definition~7.11]{Jehle2001a}. Historically, the representation has its roots directly in the seminal paper of~\cite[p.~1814]{harsanyi1967games}, where it is referred to as the \emph{Selten} or \emph{posterior-lottery} model. 
Formally, the type-agent representation is defined as follows:

\begin{definition}[Type-agent representation of the DFPA] \label{def: type-agent}
Fix a DFPA $\mathcal A$ with: bidders $N=\ssets{1,2,\dots,n}$; bidding set
$B=\sset{b_1,b_2,\dots,b_m}$; value sets
$V_i=\ssets{v_{i,1},v_{i,2},\dots,v_{i,k_i}}$, for $i\in N$; and value
distributions $F_{i,j}$ with probability density functions $$\{f_{i,j}(v_{j,1}),
f_{i,j}(v_{j,2}), \dots, f_{i,j}(v_{j,k_j})\},\qquad \text{for}\;\; i,j\in N,\; i\neq j,$$
given explicitly in the input. 

Let $\mathcal G$ be the normal form game with: players $\tilde{N}=
\sset{(i,v)\fwh{i\in N, v\in V_i}}$; for all players $(i,v)\in \tilde{N}$,
the same set of (pure) strategies $S_{(i,v)}=B$; and, for all outcomes
$\vec{s}\in B^{\tilde{N}}$ and players $(i,v)\in \tilde{N}$, a payoff 
\begin{equation}
    \label{eq:correspondence_payoffs_auction_normal-form-game}
\pi_{(i,v)}(\vec s) = u_i(\vec{s}(i,v),\hat{\vec{\beta}}_{-i};v),
\end{equation}
where $u_i$'s are the standard DFPA utilities (as given
in~\eqref{eq:DFPA-utility-interim-mixed}) and $\hat{\vec{\beta}}_{-i}$ are
\emph{pure} bidding strategies of $\mathcal{A}$ given by
\begin{equation}
    \label{eq:correspondence_equilibria_auction_normal-form-game}
\hat{\beta}_{j}(v_j) = \vec{s}(j,v_j),
\end{equation}
for all $j\in N\setminus{i}$ and $v_j\in V_j$. 
\end{definition}

Notice that \eqref{eq:correspondence_equilibria_auction_normal-form-game}
gives rise to a natural one-to-one correspondence between pure (and, by
extension, mixed as well) equilibria of auction $\mathcal A$ and the
normal-form game $\mathcal{G}$. Furthermore,
from~\eqref{eq:correspondence_payoffs_auction_normal-form-game} it is
apparent that this correspondence preserves the utility/payoffs between
games. In particular, this means that there is a one-to-one correspondence
between $\varepsilon$-MBNE of $\mathcal A$ and $\varepsilon$-approximate
(additive) mixed Nash equilibria of $\mathcal G$.
In fact, this has been the common approach for showing existence of equilibria of Bayesian games in general, by applying Nash's theorem to the type-agent game, which is indeed in normal form.

Given \cref{def: type-agent}, we can now define the correlated equilibria of the DFPA to be precisely the correlated equilibria of the corresponding type-agent game.

\begin{definition}[Correlated Equilibrium of a DFPA]
   Given a DFPA $\mathcal{A}$ and the corresponding type-agent representation game $\mathcal{G}$, a correlated equilibrium is a distribution $\sigma$ on the set of outcomes $B^{\tilde{N}}$ such that for every player $\tilde{n} \in \tilde{N}$, strategy $s_{\tilde{i}} \in S_{\tilde{n}}$, and every deviation $s_{\tilde{n}}' \in S_{\tilde{n}}$,

\begin{equation*}
    \expect[\vec{s} \sim \sigma]{\pi_{\tilde{n}}(\vec{s}) \fwh{s_{\tilde{n}}}} \geq \expect[\vec{s} \sim \sigma]{\pi_{\tilde{n}}(s_{\tilde{n}}',\vec{s}_{-\tilde{n}}) \fwh{s_{\tilde{n}}}}
\end{equation*}
\end{definition}

One of the most celebrated results on correlated equilibria is due to \cite{Papadimitriou2008a}, as well as follow-up work in \cite{jiang-correlated}, and it states the following:

\begin{theorem}[\cite{Papadimitriou2008a}] \label{thm:papadimitriou-roughgarden}
   In any succinct game of polynomial type for which the players' utilities can be computed in polynomial time, one can compute a correlated equilibrium in polynomial time. 
\end{theorem}

Informally, succinct games are games that can be represented without explicitly listing the payoffs of the players on all possible outcomes (that, in general, may be exponentially many), but rather via some alternative, efficient way. We refer the reader to \cite{Papadimitriou2008a} for a formal definition of the class.

Having established these important notions and results, the main theorem of this section follows rather straightforwardly, by applying \cref{lemma:DFPA-utilities-efficient}.

\begin{theorem}
    A correlated equilibrium of a DFPA can be computed in polynomial time.
\end{theorem}
\begin{proof}
    Given \cref{thm:papadimitriou-roughgarden}, it suffices to show that (1) the induced game $\mathcal{G}$ is succinct and (2) we can compute the (expected) utilities of the players in polynomial time.
    First, observe that the number of players and (pure) strategies of $\mathcal
    G$ are polynomial in the description size $\card{\mathcal{A}}$ of the original DFPA and, furthermore, due to~\eqref{eq:correspondence_payoffs_auction_normal-form-game} and \cref{lemma:DFPA-utilities-efficient}, the payoff functions $\pi_{(i,v)}$ are computable in polynomial (with respect to $\card{\mathcal{A}}$) time.
    Therefore, the class of games $\mathcal G$ that could arise through our reduction are succinctly representable (of \emph{polynomial type}).
    Additionally, \cref{lemma:DFPA-utilities-efficient} guarantees that the (expected) payoffs of $\mathcal G$ on \emph{mixed} strategy profiles are also polynomial-time computable, hence, the second condition is also satisfied.
\end{proof}

\section{Conclusion and Future Work}
In this paper, we considered the computational complexity of equilibrium computation in first-price auctions with discrete values and discrete bids. We established computational completeness results for deciding the existence of pure equilibria and for finding mixed equilibria in general. We also provided positive results for the natural case of iid bidders and symmetric equilibria, as well as for correlated equilibria. At the heart of several of our results lies a novel equivalence between mixed equilibria of the discrete auction and pure equilibria of the continuous auction. We believe that this result may prove to be quite useful in the future, as it allows one to choose which of the two settings to focus on, and translate results automatically to the other. 

Perhaps the most interesting, and seemingly rather challenging avenue for future work is to study the computational complexity of equilibrium computation in the discrete auction in the IPV setting. Another interesting question is to explore alternative notions of correlated equilibria defined in the literature for Bayesian games (e.g., see \cite{forges-5defs-ce,forges-revisited,forges2023correlated}) and explore the computational relation between them in the context of first-price auctions. Finally, one could explore mixed equilibria for auctions with continuous values, after coming up with an appropriate representation of the bidders' strategies in this setting.  

\section*{Acknowledgements} We are very grateful to Diogo Poças for useful ideas that served as a starting point for the NP-hardness proof.

\clearpage
\appendix

\addcontentsline{toc}{section}{Appendix}
\section*{Appendix}

\section{A Brief Discussion on Equilibrium Notions}\label{app:interim}

In the literature of the computational complexity of strategic games, the term ``$\varepsilon$-well-supported equilibrium'' has been used to refer to $\varepsilon$-approximate mixed equilibria in which every pure strategy in the support of the equilibrium strategy is utility-maximising, exactly like in \cref{def:well-supported-bayes-nash-equilibrium}. Recently, \textcite{chen2023complexity} ambiguously used the term ``$\varepsilon$-well-supported'' to refer to \emph{pure} \emph{interim} Bayes-Nash equilibria. An interim equilibrium is one as defined in our \cref{def:approx-mixed-bayes-nash-equilibrium}, in which a player chooses her strategy given her value $v_i$, which is known to her. In the terminology of auction theory, one could also define an alternative equilibrium notion called \emph{ex-ante} Bayes-Nash equilibrium. In such an equilibrium, each bidder also has some uncertainty about her own value, which she draws from a distribution $F_{i,i}$. We provide the formal definition below. 

\begin{definition}[$\varepsilon$-approximate \emph{ex-ante} mixed Bayes-Nash equilibrium of the DFPA]\label{def:ex-ante-mbne}
Let $\varepsilon \geq 0$.
A (mixed) strategy profile $\vec{\beta}=(\beta_1, \ldots, \beta_n)$ is an (interim) $\varepsilon$-approximate mixed Bayes-Nash equilibrium (MBNE) of the DFPA if for any bidder $i \in N$, 
\begin{equation*}
\label{eq:ex-ante-MBNE-def-condition-full}
\expect[\vec v_{i}\sim F_{i,i}]{u_i(\beta_i(v_i),\vec{\beta}_{-i};v_i)} \geq \expect[\vec v_{i}\sim F_{i,i}]{u_i(\vec{\gamma},\vec{\beta}_{-i};v_i)} - \varepsilon \qquad \text{for all}\;\; \vec{\gamma}\in \Delta(B).
\end{equation*}
\end{definition}

By \cref{def:well-supported-bayes-nash-equilibrium} and \cref{def:approx-mixed-bayes-nash-equilibrium}, it follows that an $\varepsilon$-well-supported equilibrium is also an $\varepsilon$-approximate Bayes-Nash equilibrium. Additionally, from \cref{def:ex-ante-mbne} above and \cref{def:approx-mixed-bayes-nash-equilibrium}, it follows that an interim Bayes-Nash equilibrium is also an ex-ante Bayes-Nash equilibrium. This implies that positive results are stronger for interim equilibria. In particular, the PTAS of \textcite{chen2023complexity} for the IPV setting, which is stated for ex-ante equilibria, does not translate to our setting via our computational equivalence result. 

We remark that we find the notion of interim equilibria more natural, and more in line with the standard definitions (e.g., see \parencite{krishna2009auction,myerson1997book}), and advocate the use of terminology which is consistent with the auction literature, e.g., see \parencite{maskin1985auction,maskin2000equilibrium,lebrun1996existence,Athey2001}.

\section{Non-Existence of Approximate Pure Equilibria}
\label{sec:existence_discrete_negative}

In this section we prove~\cref{th:tight_existence_approx}. We start by establishing a technical lemma that would be useful in our main proof:
\begin{lemma}
    \label{lemma:matching_existence_helper_1}
    For any integer $M\geq 10$,
    \[
    \frac{q}{2}M 
    \leq q (M-2)
    \leq \frac{q+1}{2}(M-3)
    \leq M-4
    \leq q(M-1)
    \leq \frac{q+1}{2}(M-2)
    \leq M-3
    \leq \frac{q+1}{2}(M-1),
    \]
    where 
    \[
    q = \frac{M-3}{M-1}-\frac{2}{3M}.
    \]
    \end{lemma}
    \begin{proof}
    First, we show that, for $M\geq 10$, the following holds:
    \begin{equation}
    \label{eq:matching_existence_helper}
    1-\frac{3}{M-1}\leq q \leq 1-\frac{2}{M-3}.
    \end{equation}
    Indeed, we can see that:
    $$
    1-\frac{3}{M-1}\leq q 
    \quad\ifif\quad 1-\frac{3}{M-1}\leq \frac{M-3}{M-1}-\frac{2}{3M} 
    \quad\ifif\quad M>1
    $$
    and
    $$
    q \leq 1-\frac{2}{M-3}
    \quad\ifif\quad \frac{M-3}{M-1}-\frac{2}{3M} \leq 1-\frac{2}{M-3}
    \quad\ifif\quad M^2-10 M+3\geq 0,
    $$
    which holds for $M\geq \sqrt{22}+5\approx 9.690$.
    
    To prove the main chain of inequalities in the statement of our lemma, again by performing basic algebraic operations, we get:
    \[
    \frac{q}{2}M 
    \leq q (M-2)
    \quad\ifif\quad
    M \geq 4;
    \]
    \[
    q (M-2)
    \leq \frac{q+1}{2}(M-3)
    \quad\ifif\quad
    q\leq \frac{M-3}{M-1}
    \]
    which holds since $q=\frac{M-3}{M-1}-\frac{2}{3M}$;
    \[
    \frac{q+1}{2}(M-3)
    \leq M-4
    \quad\ifif\quad
    q\leq 1-\frac{2}{M-3}
    \]
    which holds due to~\eqref{eq:matching_existence_helper};
    \[
    M-4
    \leq q(M-1)
    \quad\ifif\quad
    q\geq 1-\frac{3}{M-1}
    \]
    which holds due to~\eqref{eq:matching_existence_helper};
    \[
    \leq q(M-1)
    \leq \frac{q+1}{2}(M-2)
    \quad\ifif\quad
    q\leq 1-\frac{2}{M}
    \]
    which holds since, from~\eqref{eq:matching_existence_helper}, we have that
    $q\leq 1-\frac{2}{M-3}<1-\frac{2}{M}$;
    \[
    \frac{q+1}{2}(M-2)
    \leq M-3
    \quad\ifif\quad
    q\leq 1-\frac{2}{M-2}
    \]
    which holds since, from~\eqref{eq:matching_existence_helper}, we have that
    $q\leq 1-\frac{2}{M-3}<1-\frac{2}{M-2}$;
    \[
    M-3
    \leq \frac{q+1}{2}(M-1)
    \quad\ifif\quad
    q\geq 1-\frac{4}{M-1}
    \]
    which holds since, from~\eqref{eq:matching_existence_helper}, we have that
    $q\leq 1-\frac{3}{M-1}>1-\frac{4}{M-1}$.
    \end{proof}

We are now ready to proceed with the proof of our theorem:
\begin{proof}[Proof of~\cref{th:tight_existence_approx}]
Let $M\coloneqq 12$ and fix any $\varepsilon$ such that
\begin{equation}
\label{eq:tight_nonexistence_distance_helper_11}
0 \leq \varepsilon <\frac{1}{3M}-\frac{2}{M^2}=\frac{1}{72}.
\end{equation}

Consider a DFPA with two bidders with value spaces $V_1=V_2=\ssets{0,1}$, bidding space $B=\sset{0,\frac{1}{M},\frac{2}{M},\dots, 1}$, and identical prior distributions with probability mass function
\begin{equation}
    \label{eq:nonexistence_pure_distros}
f(v)=
\begin{cases}
    \frac{M-3}{M-1}-\frac{2}{3M}\eqqcolon q, &\text{if}\;\; v=0,\\
    1-q, &\text{if}\;\; v=1.
\end{cases}
\end{equation}
Notice that this distribution is well-defined, since $q\in (0,1)$ for $M\geq 10$ (see~\eqref{eq:matching_existence_helper} in the proof of~\cref{lemma:matching_existence_helper_1}).

First, we will show that no profile can be an $\varepsilon$-approximate equilibrium, unless $\beta_1(0)=\beta_2(0)=0$. Assume (without loss) that $\beta_1(0)\geq \beta_2(0)$, and let $\beta_1(0)=b\geq \frac{1}{M}$. We will prove that bidder $1$ can get an $\varepsilon$-improvement by deviating to $\beta_1(0)=0$. Indeed, the probability that bidder $1$ gets the item by reporting $b$ is at least
\[
\frac{1}{2}\prob{\beta_2(v_2) \leq b} \geq \frac{1}{2}\prob{v_2=0} = \frac{q}{2}
\]
so, keeping everything else fixed and conditioning on the true value of bidder $1$ being $v_1=0$, the improvement in the utility is:
\begin{align*}
u_1\left(\beta_1(0)=0\right) - u_1\left(\beta_1(0)=b\right)
	&= 0 - \prob{\text{bidder $1$ wins given report $b$}} \cdot (0-b)\\
 	& \geq \frac{q}{2}b \geq \frac{q}{2}\frac{1}{M}\\
 	&= \left[\frac{M-3}{2(M-1)}-\frac{1}{3M}\right]\frac{1}{M}, \qquad \text{by substituting $q$ from~\eqref{eq:nonexistence_pure_distros}},\\
 	&>\left[\frac{M-3}{3M}-\frac{1}{3M}\right]\frac{1}{M}\\
 	&=\frac{1}{3M}-\frac{4}{3}\frac{1}{M^2}>\frac{1}{3M}-2\frac{1}{M^2}=\varepsilon.
\end{align*}

Thus, from now on we fix $\beta_1(0)=\beta_2(0)=0$, and we use the shorthands $\beta_1=\beta_1(1)$, $\beta_2=\beta_2(1)$. We will now show that, for any bidder $i=1,2$:
\begin{itemize}
	\item[(a)] if $\beta_{-i}\in\sset{0,\frac{1}{M},\frac{2}{M}}$, then
	\[ 
		u_i\left(\beta_i=\beta_{-i}+\frac{1}{M};v_i=1\right)-u_i(\beta_i=b_i;v_i=1)>\varepsilon 
		\qquad\text{for all}\;\; b_i\neq \beta_{-i}+\frac{1}{M}
	\]
	\item[(b)] if $\beta_{-i}\in\sset{\frac{3}{M},\frac{4}{M},\dots,1}$, then
	\[ 
		u_i\left(\beta_i=\frac{1}{M};v_i=1\right)-u_i(\beta_i=b_i;v_i=1)>\varepsilon 
		\qquad\text{for all}\;\; b_i\neq \frac{1}{M}.
	\]
\end{itemize}
This is enough to establish the nonexistence of an $\varepsilon$-approximate equilibrium, since it shows that for \emph{any} choice of $\beta_1,\beta_2\in\sset{0,\frac{1}{M},\frac{2}{M},\dots,1}$ there will always exist an $\varepsilon$-improving deviation.
To simplify the notation, we will drop the conditioning on $v_i=1$ for the rest of the proof, and also denote
$$
\Delta(z)= \Delta(z,\beta_{-i})=u_i(\beta_i=z;v_i=1)-\max_{b_i\neq z}u_i(\beta_i=b_i;v_i=1),
$$
so that cases (a) and (b) above can be expressed simply as
\begin{itemize}
	\item[(a)] $\Delta\left(\beta_{-i}+\frac{1}{M}\right)>\varepsilon$ for $\beta_{-i}\in\sset{0,\frac{1}{M},\frac{2}{M}}$ 
	\item[(b)] $\Delta\left(\frac{1}{M}\right)>\varepsilon$ for $\beta_{-i}\in\sset{\frac{3}{M},\frac{4}{M},\dots,1}$.
\end{itemize}
First, observe that
$$
u_i(\beta_i=b_i)= 1-b_i\qquad \text{for all}\;\; b_i\geq \beta_{-i}+\frac{1}{M}
$$
Next, for the first case (a), we perform a case analysis depending on $\beta_{-i}=0,\frac{1}{M},\frac{2}{M}$:
\begin{itemize}
	\item[---] if $\beta_{-i}=0$, then $u_i(\beta_i=0)=\frac{1}{2}(1-0)=\frac{1}{2}$ and so
	\begin{align*}
	\Delta\left(\frac{1}{M}\right)
		= 1-\frac{1}{M}-\max\sset{\frac{1}{2},1-\frac{2}{M}}
		= \min\sset{\frac{1}{2}-\frac{1}{M},\frac{1}{M}}=\frac{1}{M}>\varepsilon.
	\end{align*}
	\item[---] if $\beta_{-i}=\frac{1}{M}$, then 
	$$u_i(\beta_i=0)=\frac{1}{2}q(1-0)=q\frac{1}{2}$$ 
	and 
	$$u_i\left(\beta_i=\frac{1}{M}\right)=\left[q+\frac{1}{2}(1-q)\right]\left(1-\frac{1}{M}\right)=\frac{q+1}{2}\left(1-\frac{1}{M}\right).$$
	So
	\begin{align*}
	\Delta\left(\frac{2}{M}\right)
	&= 1-\frac{2}{M}-\frac{1}{M}\max\sset{q\frac{M}{2},\frac{q+1}{2}(M-1),M-3}\\
	&= 1-\frac{2}{M}-\frac{1}{M}\frac{q+1}{2}(M-1), &&\text{by~\cref{lemma:matching_existence_helper_1}},\\
	&=\frac{1}{3M} - \frac{1}{3}\frac{1}{M^2}, &&\text{by substituting $q$ from~\eqref{eq:nonexistence_pure_distros}},\\
	&>\frac{1}{3M} - 2\frac{1}{M^2}=\varepsilon.
	\end{align*}
	\item[---] if $\beta_{-i}=\frac{2}{M}$, then 
	$$u_i(\beta_i=0)=\frac{1}{2}q(1-0)=q\frac{1}{2},$$ 
	$$
	u_i\left(\beta_i=\frac{1}{M}\right)=q\left(1-\frac{1}{M}\right)
	$$
	and 
	$$u_i\left(\beta_i=\frac{2}{M}\right)=\frac{q+1}{2}\left(1-\frac{2}{M}\right).$$
	So
	\begin{align*}
	\Delta\left(\frac{3}{M}\right)
	& = 1-\frac{3}{M}-\frac{1}{M}\max\sset{q\frac{M}{2},q(M-1),\frac{q+1}{2}(M-2),M-4}\\
	& = 1-\frac{3}{M} - \frac{1}{M}\frac{q+1}{2}(M-2), &&\text{by~\cref{lemma:matching_existence_helper_1}},\\
	& = \frac{M^2-6 M+2}{3 (M-1) M^2}, &&\text{by substituting $q$ from~\eqref{eq:nonexistence_pure_distros}},\\
	&> \frac{M^2-6 M}{3 M^3}=\frac{1}{3M}-\frac{2}{M^2}
	 = \varepsilon.
	\end{align*}
\end{itemize}

For the second case (b), it is
$$u_i(\beta_i=0)=\frac{1}{2}q(1-0)=q\frac{1}{2};$$ 
$$
u_i(\beta_i=b_i)=q(1-b_i)\qquad\text{for all}\;\; b_i=\frac{1}{M},\dots,\beta_{-i}-\frac{1}{M};
$$
$$u_i(\beta_i=\beta_{-i})=\frac{q+1}{2}(1-\beta_{-i}).$$
Therefore,
\begin{align*}
\Delta\left(\frac{1}{M}\right) &=
	 q\left(1-\frac{1}{M}\right)-\frac{1}{M}\max\sset{q\frac{M}{2},q(M-2),\frac{q+1}{2}(M-M\beta_{-i}), M-M\beta_{-i}-1}\\
	 &\geq q\left(1-\frac{1}{M}\right)-\frac{1}{M}\max\sset{q\frac{M}{2},q(M-2),\frac{q+1}{2}(M-3), M-4}, &&\text{since}\;\; \beta_{-i}\geq \frac{3}{M},\\
	 &= q\left(1-\frac{1}{M}\right)-\frac{1}{M}(M-4), &&\text{by~\cref{lemma:matching_existence_helper_1}},\\
	 &= \frac{1}{3M}+\frac{2}{3M^2}> \frac{1}{3M}-\frac{2}{M^2}= \varepsilon,
\end{align*}
where the penultimate equality follows by substituting $q$ from~\eqref{eq:nonexistence_pure_distros}.
\end{proof}

\section{NP-Hardness of Deciding Approximate Equilibria} \label{app:np-hardness}

In this section, we provide the full proof for the NP-hardness of the problem. In our reduction, we consider an instance of the problem with bidding space $B=\{0,1,2,3\}$ and a common value space for all bidders $V=\left\{0,\frac{9}{4},10\right\}$. Notice that in this case $V$ and $B$ do not lie inside $[0,1]$ as we originally defined. However, it is trivial to see, by the definition of the utility, that if we divide all elements of both sets by the maximum ($10$ in this case), then they are all in $[0,1]$ as required, and, additionally, the comparisons between the corresponding utilities in our proof still hold, as they have all been multiplied by $1/10$. This means that the value of $\varepsilon$ that we will compute for the hardness of approximation will also have to be multiplied by $1/10$ at the end of the proof.

Let a DFPA instance with $B=\{0,1,2,3\}$ and a common value space for all bidders, $V=\left\{0,\frac{9}{4},10\right\}$. In this context, the bidding strategy of a bidder $i$ can be efficiently represented by one of the $4^3=64$ functions of type $V\rightarrow B$. We construct a reduction from the Circuit Satisfiability problem (\textsc{Circuit-SAT}), which is NP-complete \parencite{Garey1979a}.

\begin{definition}
    Given a Boolean circuit $C$, \textsc{Circuit-SAT} is the problem of deciding whether there is an assignment to its inputs that makes its output \emph{true}.
\end{definition}

Our construction has 3 basic parts. Given an instance of \textsc{Circuit-SAT}, we transform it to an instance of the problem of calculating a PBNE in a DFPA as follows: 

\begin{itemize}

    \item[---] We map \emph{false} and \emph{true} to two specific strategies, which we will denote as $s_0$ and $s_1$.     
    
    \item[---] For each logic gate in the circuit, we introduce new bidders (along with their subjective prior distributions) to the DFPA, such that the best-response of the bidder corresponding to the gate's output given the gate's input bidder(s) reflects the type of gate (with respect to the mapping of \emph{false} and \emph{true} to $s_0$ and $s_1$). This ensures that a PBNE can only exist if the gates are correctly simulated.
    
    \item[---] We introduce two extra bidders such that the conditions of a PBNE are not satisfied between them unless the bidder corresponding to the output of the circuit represents the value \emph{true}.
    
\end{itemize}

\addtocontents{toc}{\protect\setcounter{tocdepth}{1}}
\subsection{Simulating Circuit Gates} \label{sec:simulation}
\addtocontents{toc}{\protect\setcounter{tocdepth}{2}}
We know that the NAND gate is functionally complete, i.e. all possible boolean functions can be represented by circuits using only NAND gates. Since we can construct a NAND from an OR and a NOT gate, we can transform (in polynomial time) any circuit of an instance of \textsc{Circuit-SAT} to one that contains only OR and NOT gates. Therefore, it suffices to reduce from these new instances, which are also hard to solve.
To simulate these in our reduction, we will introduce 2 corresponding gadgets, which will be used in combination with an extra gadget we define, the \emph{projection} gadget.
The \emph{projection} gadget has one input and one output bidder, and it is used for appropriately mapping any strategies different from $s_0$ and $s_1$ that might arise as outputs of the OR and NOT gadgets back to either $s_0$ or $s_1$. The values of \emph{false} and \emph{true} will be encoded by a bidder's strategy as follows:

\begin{itemize}
    \item \emph{false} is encoded by bidding strategies such that $\hat{\beta}(0)=0$, $\hat{\beta}\left(\frac{9}{4}\right)=1$, and $\hat{\beta}(10)=2$;
    \item \emph{true} is encoded by bidding strategies such that $\hat{\beta}(0)=0$, $\hat{\beta}\left(\frac{9}{4}\right)=2$, and $\hat{\beta}(10)=3$.
\end{itemize}

To simplify notation, we denote these strategies as $s_0=(0,1,2)$ for \emph{false} and $s_1=(0,2,3)$ for \emph{true}.

\subsubsection{Projection gadget}\label{sec:projection}
We begin by designing the \emph{projection} gadget, as we will need it to simulate the OR and NOT gates. Suppose bidder $i$ encodes a boolean value via her strategy $\hat{\beta}_i$. Consider a bidder $j$ whose subjective prior is given as follows: $F_{ji}$ is the discrete distribution that assumes value $0$ with probability $6/100$ and value $9/4$ with probability $94/100$; for any other $i'\neq i,j$, $F_{ji'}$ is the discrete distribution that assumes value $0$ with probability $1$. We will now perform a case analysis to calculate $j$'s best-response to $i$'s strategies. Whenever $i$ plays $s_0$ or $s_1$, $j$ will have to copy her strategy in order to best-respond. Additionally, all other strategies would result in $j$'s utility being at least $\frac{1}{18}$ worse. We also provide $j$'s best-response to strategies other than $s_0$ and $s_1$, as they will be required for the proof of \cref{lem:inpgate}. As $F_{ji}$ has probability mass only on $0$ and $9/4$, we can ignore $i$'s strategy when having value $10$; we will denote this by $x$ below. To begin with, $j$'s best-response must satisfy $\hat{\beta}_j(0)=0$ due to the no-overbidding assumption. In what follows, we allow a slight abuse of notation where $u_j(b;v)$ denotes the expected utility that bidder $j$ receives (from her point of view) when having value $v$ and choosing to bid $b$ - this is in contrast to the usual definition where we also include the other bidders' strategies, however here we explicitly list their strategies in the description of each different case. We consider the following ones:

\begin{enumerate}

    \item Suppose bidder $i$ plays the strategy $(0,0,x)$. We can compute the winning probabilities and utilities for bidder $j$ when bidding $0$, $1$, $2$, or $3$, and when having value $9/4$ or $10$:
    
    \begingroup
    \renewcommand{\arraystretch}{1.2}
    \begin{center}\begin{tabular}{c|cccc}
        $b$ & $0$ & $1$ & $2$ & $3$\\
        \hline
        $H_j(b)$ & $\frac{1}{n}$ & $1$ & $1$ & $1$ \\
        $u_j(b;9/4)$ & $\frac{9}{4n}$ & $\frac{5}{4}$ & $\frac{1}{4}$ & $-$\\
        $u_j(b;10)$ & $\frac{10}{n}$ & $9$ & $8$ & $7$
    \end{tabular}\end{center}
    \endgroup
    
    Since $n\geq 2$, we have the bounds $u_j(0;9/4)\leq\frac{9}{8}$, $u_j(0;10)\leq 5$. Comparing these values, we conclude that it is optimal for bidder $j$ to bid $\hat{\beta}_j(9/4)=1$, $\hat{\beta}_j(10)=1$. Also, we can see that there is no other $\varepsilon$-best-response for $\varepsilon<\frac{1}{18}$.

    \item Suppose bidder $i$ plays some strategy of the form $(0,1,x)$. We can compute the winning probabilities and utilities for bidder $j$ when bidding $0$, $1$, $2$, or $3$, and when having value $9/4$ or $10$:
    
    \begingroup
    \renewcommand{\arraystretch}{1.2}
    \begin{center}\begin{tabular}{c|cccc}
        $b$ & $0$ & $1$ & $2$ & $3$\\
        \hline
        $H_j(b)$ & $\frac{3}{50n}$ & $\frac{53}{100}$ & $1$ & $1$ \\
        $u_j(b;9/4)$ & $\frac{27}{200n}$ & $\frac{53}{80}$ & $\frac{1}{4}$ & $-$\\
        $u_j(b;10)$ & $\frac{3}{5n}$ & $\frac{477}{80}$ & $8$ & $7$
    \end{tabular}\end{center}
    \endgroup
    
    Since $n\geq 2$, we have the bounds $u_j(0;9/4)\leq\frac{27}{400}$, $u_j(0;10)\leq\frac{3}{10}$. Comparing these values, we conclude that it is optimal for bidder $j$ to bid $\hat{\beta}_j(9/4)=1$, $\hat{\beta}_j(10)=2$. Also, we can see that there is no other $\varepsilon$-best-response for $\varepsilon<\frac{1}{18}$. 
    
    \item Suppose bidder $i$ plays some strategy of the form $(0,2,x)$. We can compute the winning probabilities and utilities for bidder $j$ when bidding $0$, $1$, $2$, or $3$, and when having value $9/4$ or $10$:
    
    \begingroup
    \renewcommand{\arraystretch}{1.2}
    \begin{center}\begin{tabular}{c|cccc}
        $b$ & $0$ & $1$ & $2$ & $3$\\
        \hline
        $H_j(b)$ & $\frac{3}{50n}$ & $\frac{3}{50}$ & $\frac{53}{100}$ & $1$ \\
        $u_j(b;9/4)$ & $\frac{27}{200n}$ & $\frac{3}{40}$ & $\frac{53}{400}$ & $-$\\
        $u_j(b;10)$ & $\frac{3}{5n}$ & $\frac{27}{50}$ & $\frac{106}{25}$ & $7$
    \end{tabular}\end{center}
    \endgroup
    
    Since $n\geq 2$, we have the bounds $u_j(0;9/4)\leq\frac{27}{400}$, $u_j(0;10)\leq\frac{3}{10}$. Comparing these values, we conclude that it is optimal for bidder $j$ to bid $\hat{\beta}_j(9/4)=2$, $\hat{\beta}_j(10)=3$. Also, we can see that there is no other $\varepsilon$-best-response for $\varepsilon<\frac{1}{18}$.
    
    \item Suppose bidder $i$ plays the strategy $(0,3,x)$. We can compute the winning probabilities and utilities for bidder $j$ when bidding $0$, $1$, $2$, or $3$, and when having value $9/4$ or $10$:
    
    \begingroup
    \renewcommand{\arraystretch}{1.2}
    \begin{center}\begin{tabular}{c|cccc}
        $b$ & $0$ & $1$ & $2$ & $3$\\
        \hline
        $H_j(b)$ & $\frac{3}{50n}$ & $\frac{3}{50}$ & $\frac{3}{50}$ & $\frac{53}{100}$ \\
        $u_j(b;9/4)$ & $\frac{27}{200n}$ & $\frac{3}{40}$ & $\frac{3}{200}$ & $-$\\
        $u_j(b;10)$ & $\frac{3}{5n}$ & $\frac{27}{50}$ & $\frac{24}{50}$ & $\frac{371}{100}$
    \end{tabular}\end{center}
    \endgroup
    
    Since $n\geq 2$, we have the bounds $u_j(0;9/4)\leq\frac{27}{400}$, $u_j(0;10)\leq\frac{3}{10}$. Comparing these values, we conclude that it is optimal for bidder $j$ to bid $\hat{\beta}_j(9/4)=1$, $\hat{\beta}_j(10)=3$. Also, we can see that there is no other $\varepsilon$-best-response for $\varepsilon<\frac{1}{18}$.

\end{enumerate}

From the above case analysis, we obtain the following lemma:

\begin{lemma}\label{lem:projgate}
    If the bidder $i$ corresponding to the input of a \emph{projection} gadget plays either $s_0$ or $s_1$, then, in all $\varepsilon$-PBNE of the resulting DFPA (where $\varepsilon \in [0,\frac{1}{18})$), bidder $j$ (corresponding to the gadget's output) matches $i$'s strategy.
    Additionally, if bidder $i$ plays some strategy of the form $(0,1,x)$ or $(0,2,x)$, bidder $j$'s unique $\varepsilon$-best-response for $\varepsilon \in [0,\frac{1}{18})$ is $s_0$ and $s_1$ respectively.
\end{lemma}

Using the construction of the projection gadget, we can similarly define a \emph{split} gadget. Let $i$ be the bidder corresponding to the input of the split, and $j,k$ be the bidders corresponding to its outputs. We can then define the subjective priors of $j$ and $k$ as if each of them was the output of a projection gadget from $i$. It follows directly from our analysis above that in all $\varepsilon$-PBNE of the resulting DFPA, the output bidders of a split in the circuit have to match the strategy of the input bidder. Additionally, as these splits are used in order to repeat a value in multiple parts of the circuit, it is important to note that the number of bidders introduced remains polynomial on the number of gates of the circuit, since all outputs of a split will end up being inputs to some gate.

\subsubsection{OR gadget}
Next, we show how to simulate an OR gate. Suppose bidders $i$, $j$ encode boolean values via their strategies $\hat{\beta}_i$ and $\hat{\beta}_j$. Consider a bidder $k$ whose subjective prior is given as follows: $F_{ki}$ and $F_{kj}$ are discrete distributions that assume value $0$ with probability $6/100$ and value $9/4$ with probability $94/100$; for any other $i'\neq i,j$, $F_{ki'}$ is the discrete distribution that assumes value $0$ with probability $1$. Again, we assume that bidders other than $k$ adopt one of the two strategies $(0,1,2)$ or $(0,2,3)$. In particular, they bid $0$ at value $0$, and either $1$ or $2$ at value $9/4$; for our analysis in this part, what they bid at value $10$ does not make a difference, as $k$'s subjective prior puts 0 probability mass at value $10$. We can show that bidder $k$'s best-response when having value $9/4$ is to copy the highest among the strategies of bidders $i$ and $j$. Then, we can add an extra \emph{projection} gadget, mapping $k$'s strategy which will be of the form either $(0,1,x)$ or $(0,2,x')$ back to $s_0$ and $s_1$ respectively. Once again, note that $k$'s best-response must satisfy $\hat{\beta}_k(0)=0$, due to the no-overbidding assumption. In all the following tables we omit the row corresponding to $u_k(b;10)$, as this will be mapped back to one of $s_0, s_1$ with the use of the \emph{projection} gadget.

\begin{enumerate}

    \item Suppose bidders $i$ and $j$ both play the strategy $(0,1,2)$. We can compute the winning probabilities and utilities for bidder $k$ when bidding $0$, $1$, $2$, or $3$, and when having value $9/4$ or $10$:
    
    \begingroup
    \renewcommand{\arraystretch}{1.2}
    \begin{center}\begin{tabular}{c|cccc}
        $b$ & $0$ & $1$ & $2$ & $3$\\
        \hline
        $H_k(b)$ & $\frac{9}{2500n}$ & $\frac{2659}{7500}$ & $1$ & $1$ \\
        $u_k(b;9/4)$ & $\frac{81}{10000n}$ & $\frac{2659}{6000}$ & $\frac{1}{4}$ & $-$
    \end{tabular}\end{center}
    \endgroup
    
    Since $n\geq 3$, we have the bound $u_k(0;9/4)\leq\frac{27}{10000}$. Comparing these values, we conclude that it is optimal for bidder $j$ to bid $\hat{\beta}_k(9/4)=1$. Also, we can see that there is no other $\varepsilon$-best-response for $\varepsilon<\frac{1}{18}$.
    
    \item Suppose bidder $i$ plays the strategy $(0,1,2)$ and bidder $j$ plays the strategy $(0,2,3)$ (the same analysis holds if we swap the roles of $i$ and $j$). We can compute the winning probabilities and utilities for bidder $k$ when bidding $0$, $1$, $2$, or $3$, and when having value $9/4$ or $10$:
    
    \begingroup
    \renewcommand{\arraystretch}{1.2}
    \begin{center}\begin{tabular}{c|cccc}
        $b$ & $0$ & $1$ & $2$ & $3$\\
        \hline
        $H_k(b)$ & $\frac{9}{2500n}$ & $\frac{159}{5000}$ & $\frac{53}{100}$ & $1$ \\
        $u_k(b;9/4)$ & $\frac{81}{10000n}$ & $\frac{159}{4000}$ & $\frac{53}{400}$ & $-$
    \end{tabular}\end{center}
    \endgroup
    
    Since $n\geq 3$, we have the bound $u_k(0;9/4)\leq\frac{27}{10000}$. Comparing these values, we conclude that it is optimal for bidder $j$ to bid $\hat{\beta}_k(9/4)=2$. Also, we can see that there is no other $\varepsilon$-best-response for $\varepsilon<\frac{1}{18}$.
    
    \item Suppose bidders $i$ and $j$ both play the strategy $(0,2,3)$. We can compute the winning probabilities and utilities for bidder $k$ when bidding $0$, $1$, $2$, or $3$, and when having value $9/4$ or $10$:
    
    \begingroup
    \renewcommand{\arraystretch}{1.2}
    \begin{center}\begin{tabular}{c|cccc}
        $b$ & $0$ & $1$ & $2$ & $3$\\
        \hline
        $H_k(b)$ & $\frac{9}{2500n}$ & $\frac{9}{2500}$ & $\frac{2659}{7500}$ & $1$ \\
        $u_k(b;9/4)$ & $\frac{81}{10000n}$ & $\frac{9}{2000}$ & $\frac{2659}{30000}$ & $-$
    \end{tabular}\end{center}
    \endgroup
    
    Since $n\geq 3$, we have the bound $u_k(0;9/4)\leq\frac{27}{10000}$. Comparing these values, we conclude that it is optimal for bidder $j$ to bid $\hat{\beta}_k(9/4)=2$. Also, we can see that there is no other $\varepsilon$-best-response for $\varepsilon<\frac{1}{18}$.
    
\end{enumerate}

Using the construction of the OR gadget followed by a \emph{projection} gadget, we get the following lemma: 

\begin{lemma}\label{lem:orgate}

    Let $i$ and $j$ be the bidders corresponding to the inputs of an OR gate and assume each of them plays either $s_0$ or $s_1$. Additionally, let $k$ be the bidder introduced in the OR gadget, which we pass as an input to a \emph{projection} gadget with corresponding output bidder $k'$. Then, the only $\varepsilon$-best-response of bidder $k'$ (where $\varepsilon \in [0,\frac{1}{18})$) is $s_1$, if and only if at least one of $i$ and $j$ plays $s_1$, else it is $s_0$.
    
\end{lemma}

The above lemma implies that the OR gate is correctly simulated at all PBNE.

\subsubsection{Negation gadget}
To conclude this section, we show how to simulate a negation gate. Suppose bidder $i$ encodes a boolean value via her strategy $\hat{\beta}_i$. Consider a bidder $j$ whose subjective prior is given as follows: $F_{ji}$ is the discrete distribution that assumes value $0$ with probability $6/100$ and value $10$ with probability $94/100$; for any other $i'\neq i,j$, $F_{ji'}$ is the discrete distribution that assumes value $0$ with probability $1$. Again, assuming that every bidder other than $j$ adopts one of the two strategies $(0,1,2)$ or $(0,2,3)$, we can show the best-response of bidder $j$, when having value $9/4$, is to negate the strategy of bidder $i$. Afterwards, similarly to what we did for the OR gate, we add a \emph{projection} gadget, mapping the strategy of $j$ (which will be either of the form $(0,1,x)$ or $(0,2,x)$) back to $s_0$ and $s_1$ respectively. Similarly to the description of the OR gadget, it must be the case that $j$'s best-response satisfies $\hat{\beta}_j(0)=0$, due to the no-overbidding assumption. Once again, we omit the row of the table corresponding to $u_j(b;10)$.

\begin{enumerate}

    \item Suppose bidder $i$ plays the strategy $(0,1,2)$. We can compute the winning probabilities and utilities for bidder $j$ when bidding $0$, $1$, $2$, or $3$, and when having a value of $9/4$ or $10$:

    \begingroup
    \renewcommand{\arraystretch}{1.2}
    \begin{center}\begin{tabular}{c|cccc}
        $b$ & $0$ & $1$ & $2$ & $3$\\
        \hline
        $H_j(b)$ & $\frac{3}{50n}$ & $\frac{3}{50}$ & $\frac{53}{100}$ & $1$ \\
        $u_j(b;9/4)$ & $\frac{27}{200n}$ & $\frac{3}{40}$ & $\frac{53}{400}$ & $-$
    \end{tabular}\end{center}
    \endgroup
    
    Since $n\geq 2$, we have the bound $u_j(0;9/4)\leq\frac{27}{400}$. Comparing these values, we conclude that it is optimal for bidder $j$ to bid $\hat{\beta}_j(9/4)=2$. Also, we can see that there is no other $\varepsilon$-best-response for $\varepsilon<\frac{1}{18}$.
    
    \item Suppose bidder $i$ plays the strategy $(0,2,3)$. We can compute the winning probabilities and utilities for bidder $j$ when bidding $0$, $1$, $2$, or $3$, and when having a value of $9/4$ or $10$:
    
    \begingroup
    \renewcommand{\arraystretch}{1.2}
    \begin{center}\begin{tabular}{c|cccc}
        $b$ & $0$ & $1$ & $2$ & $3$\\
        \hline
        $H_j(b)$ & $\frac{3}{50n}$ & $\frac{3}{50}$ & $\frac{3}{50}$ & $\frac{53}{100}$ \\
        $u_j(b;9/4)$ & $\frac{27}{200n}$ & $\frac{3}{40}$ & $\frac{3}{200}$ & $-$
    \end{tabular}\end{center}
    \endgroup

     Since $n\geq 2$, we have the bound $u_j(0;9/4)\leq\frac{27}{400}$. Comparing these values, we conclude that it is optimal for bidder $j$ to bid $\hat{\beta}_j(9/4)=1$. Also, we can see that there is no other $\varepsilon$-best-response for $\varepsilon<\frac{1}{18}$.
     
\end{enumerate}

Let $j'$ be the bidder corresponding to the output of a \emph{projection} gadget, which takes $j$ as input. Again, assuming that every bidder other than $j, j'$ adopts one of the two strategies $(0,1,2)$ or $(0,2,3)$, we can show the best-response of bidder $j'$ is to negate the strategy of bidder $i$:

\begin{lemma}\label{lem:notgate}
    If the bidder $i$ corresponding to the input of a negation gadget plays either $s_0$ or $s_1$, then, in all $\varepsilon$-PBNE of the resulting DFPA (where $\varepsilon \in [0,\frac{1}{18})$), bidder $j'$ (corresponding to the output of the \emph{projection} gadget added after the negation gate) plays the opposite of $i$'s strategy.
\end{lemma}

\addtocontents{toc}{\protect\setcounter{tocdepth}{1}}
\subsection{Input gadget}\label{sec:inputgadget}
\addtocontents{toc}{\protect\setcounter{tocdepth}{2}}
We can see (and trivially prove by structural induction on the circuit) that, given that the bidders corresponding to the inputs of the circuit play either $s_0$ or $s_1$, all bidders in the DFPA are playing $\varepsilon$-best-responses if and only if the gates are correctly simulated. However, it remains to specify the subjective priors for these input bidders. To achieve this normalized behaviour, we can introduce a special gadget, similar to the \emph{projection} gadget, for each input to the circuit. Essentially, we add two bidders projecting to each other; our analysis will show that in all $\varepsilon$-PBNE they should both be playing either $s_0$ or $s_1$.

Suppose bidder $i$ encodes a boolean value via her strategy $\hat{\beta}_i$. Consider a bidder $j$, whose subjective prior is given as follows: $F_{ji}$ is the discrete distribution that assumes value $0$ with probability $6/100$ and value $9/4$ with probability $94/100$; for any other $i'\neq i,j$, $F_{ji'}$ is the discrete distribution that assumes value $0$ with probability $1$. At the same time, $i$'s subjective prior is given by: $F_{ij}$ is the discrete distribution that assumes value $0$ with probability $6/100$ and value $9/4$ with probability $94/100$; for any other $j'\neq i,j$, $F_{ij'}$ is the discrete distribution that assumes value $0$ with probability $1$. 

Using the analysis from \cref{sec:projection}, we can list all the best-responses of $j$ in a potential PBNE:

\begin{center}\begin{tabular}{c|c}
        $i$'s strategy & $j$'s best-response \\
        \hline
        $(0,0,x)$ & $(0,1,1)$ \\
        $(0,1,x)$ & $(0,1,2)$ \\
        $(0,2,x)$ & $(0,2,3)$ \\
        $(0,3,x)$ & $(0,1,3)$
\end{tabular}
\captionof{table}{Bidder $j$'s best-responses given $i$'s strategies}
\end{center}

Also, from previous sections, we know that there are no other strategies that result in utility within $1/18$ of that of the best-responses. It is evident that, due to the symmetric nature of the input gadget, the best-responses of $i$ to the strategies of $j$ are symmetrical to the ones in the above table. Therefore, in all PBNE $i$ and $j$ would have to play the same strategy, which would be either $s_0$ or $s_1$, thus yielding the intended behaviour of the circuit's input gates.

\begin{lemma}\label{lem:inpgate}
    If $i$ and $j$ are the bidders corresponding to some input of the circuit (as described in the construction of the input gadget), then, in all $\varepsilon$-PBNE for $\varepsilon\in[0,1/18)$, $i$ and $j$ both play the same strategy, which can be either $s_0$ or $s_1$.
\end{lemma}

Combining \crefrange{lem:projgate}{lem:inpgate}, we deduce the following:

\begin{corollary}\label{cor:gate-simulation}
    Using the above mapping from a circuit to a DFPA, all induced bidders are simultaneously $\varepsilon$-best-responding (for $\varepsilon\in[0,1/18)$) to each other if and only if the values that their strategies map to satisfy the gates of the circuit.
\end{corollary}

\addtocontents{toc}{\protect\setcounter{tocdepth}{1}}
\subsection{Output Gadget} \label{sec:output-gadget}
\addtocontents{toc}{\protect\setcounter{tocdepth}{2}}
Thus far, our construction ensures that at an $\varepsilon$-PBNE it is necessary that all the gates are correctly simulated. As we are reducing from the \textsc{Circuit-SAT} problem, we want an $\varepsilon$-PBNE to exist only in cases where the output of the circuit can take \emph{true} as a value. This leads us to the last gadget of our construction - the \emph{output} gadget. This introduces two new bidders to the auction, $k$ and $\ell$, which (given their subjective priors, which assume value $0$ with probability $1$ for every bidder other than each other) cannot play best-responses simultaneously. Then, we also add $k$'s prior distribution for bidder $i$ (the bidder corresponding to the output of the circuit), such that now the bidders can simultaneously play $\varepsilon$-best-responses if and only if bidder $i$'s strategy represents the value \emph{true}.

Next, we provide the details of the construction of the output gadget. We introduce the two new bidders $k$ and $\ell$ with the following prior distributions: $F_{k\ell}$ and $F_{\ell k}$ are discrete distributions that assume value $0$ with probability $8/11$ and value $10$ with probability $3/11$. Bidder $k$'s prior for bidder $i$ (corresponding to the circuit's output) is $F_{ki}$, the discrete distribution that assumes value $0$ with probability $1/2$ and value $9/4$ with probability $1/2$. For any other $k'\neq k,\ell$ and $\ell' \neq k,\ell,i$, $F_{\ell k'}$ and $F_{k\ell'}$ are discrete distributions that assume value $0$ with probability $1$.

We will begin our analysis by calculating bidder $\ell$'s best-responses depending on $k$'s strategies. Note that, since $F_{\ell k}$ only has mass at $0$ and $10$, it doesn't matter what $\hat{\beta}_k(9/4)$ is for strategy $\hat{\beta}_k$, hence we will denote this as $x$ in our analysis. As our aim is to compare how $k$ and $\ell$ best-respond to each other, and the subjective prior of each puts positive probability only on $0$ and $10$, we will only analyse what $u_{\ell}(b;10)$ will be. We begin by noting that all of $\ell$'s best-responses must satisfy $\hat{\beta}_{\ell}(0)=0$, due to the no-overbidding assumption. Then:

\begin{enumerate}

    \item Suppose bidder $k$ plays some strategy of the form $(0,x,0)$. We can compute the winning probabilities and utilities for bidder $\ell$ when bidding $0$, $1$, $2$, or $3$ and having value $10$:
    
    \begingroup
    \renewcommand{\arraystretch}{1.2}
    \begin{center}\begin{tabular}{c|cccc}
        $b$ & $0$ & $1$ & $2$ & $3$\\
        \hline
        $H_{\ell}(b)$ & $\frac{1}{n}$ & $1$ & $1$ & $1$ \\
        $u_{\ell}(b;10)$ & $\frac{10}{n}$ & $9$ & $8$ & $7$
    \end{tabular}\end{center}
    \endgroup
    
    Since $n\geq 2$, we have the bound $u_{\ell}(0;10)\leq 5$. Comparing the values in the table, we conclude that it is optimal for bidder $\ell$ to bid $\hat{\beta}_{\ell}(10)=1$. Also, we can see that there is no other $\varepsilon$-best-response for $\varepsilon<\frac{1}{18}$.
    
    \item Suppose bidder $k$ plays some strategy of the form $(0,x,1)$. We can compute the winning probabilities and utilities for bidder $\ell$ when bidding $0$, $1$, $2$, or $3$ and having value $10$:
    
    \begingroup
    \renewcommand{\arraystretch}{1.2}
    \begin{center}\begin{tabular}{c|cccc}
        $b$ & $0$ & $1$ & $2$ & $3$\\
        \hline
        $H_{\ell}(b)$ & $\frac{8}{11n}$ & $\frac{19}{22}$ & $1$ & $1$ \\
        $u_{\ell}(b;10)$ & $\frac{80}{11n}$ & $\frac{171}{22}$ & $8$ & $7$
    \end{tabular}\end{center}
    \endgroup
    
    Since $n\geq 2$, we have the bound $u_{\ell}(0;10)\leq \frac{40}{11}$. Comparing the values in the table, we conclude that it is optimal for bidder $\ell$ to bid $\hat{\beta}_{\ell}(10)=2$. Also, we can see that there is no other $\varepsilon$-best-response for $\varepsilon<\frac{1}{18}$.

    \item Suppose bidder $k$ plays some strategy of the form $(0,x,2)$. We can compute the winning probabilities and utilities for bidder $\ell$ when bidding $0$, $1$, $2$, or $3$ and having value $10$:
    
    \begingroup
    \renewcommand{\arraystretch}{1.2}
    \begin{center}\begin{tabular}{c|cccc}
        $b$ & $0$ & $1$ & $2$ & $3$\\
        \hline
        $H_{\ell}(b)$ & $\frac{8}{11n}$ & $\frac{8}{11}$ & $\frac{19}{22}$ & $1$ \\
        $u_{\ell}(b;10)$ & $\frac{80}{11n}$ & $\frac{72}{11}$ & $\frac{76}{11}$ & $7$
    \end{tabular}\end{center}
    \endgroup
    
    Since $n\geq 2$, we have the bound $u_{\ell}(0;10)\leq \frac{40}{11}$. Comparing these values, we conclude that it is optimal for bidder $l$ to bid $\hat{\beta}_l(9/4)=1$, $\hat{\beta}_l(10)=3$. Also, we can see that there is no other $\varepsilon$-best-response for $\varepsilon<\frac{1}{18}$.

    \item Suppose bidder $k$ plays some strategy of the form $(0,x,3)$. We can compute the winning probabilities and utilities for bidder $\ell$ when bidding $0$, $1$, $2$, or $3$ and having value $10$:
    
    \begingroup
    \renewcommand{\arraystretch}{1.2}
    \begin{center}\begin{tabular}{c|cccc}
        $b$ & $0$ & $1$ & $2$ & $3$\\
        \hline
        $H_{\ell}(b)$ & $\frac{8}{11n}$ & $\frac{8}{11}$ & $\frac{8}{11}$ & $\frac{19}{22}$ \\
        $u_{\ell}(b;10)$ & $\frac{80}{11n}$ & $\frac{72}{11}$ & $\frac{64}{11}$ & $\frac{133}{22}$
    \end{tabular}\end{center}
    \endgroup
    
    Since $n\geq 2$, we have the bound $u_{\ell}(0;10)\leq \frac{40}{11}$. Comparing these values, we conclude that it is optimal for bidder $\ell$ to bid $\hat{\beta}_{\ell}(10)=1$. Also, we can see that there is no other $\varepsilon$-best-response for $\varepsilon<\frac{1}{18}$.

\end{enumerate}

We can note that $k$'s strategy at an equilibrium (from $\ell$'s point of view) can be represented by just $\hat{\beta}_k(10)$, since $\hat{\beta}_k(0)=0$ at all PBNE due to the no-overbidding assumption. We can do the same thing for $\ell$, since $k$'s prior for her has mass at a point other than $0$, and that is at $10$, so we get the following table representing $\ell$'s best-responses to $k$:

\begin{center}\begin{tabular}{c|cccc}
    $\hat{\beta}_k(10)$ & $0$ & $1$ & $2$ & $3$\\
    \hline
    $\hat{\beta}_{\ell}(10)$ & $1$ & $2$ & $3$ & $1$ 
\end{tabular}
\captionof{table}{Bidder $\ell$'s best-responses to bidder $k$'s strategies}
\end{center}

We will now look at $k$'s best-response. Since $k$ has a non-zero prior for both $i$ and $\ell$, her best-response will depend on two other bidders, resulting in more cases. We begin our analysis with the ones where $i$ plays $(0,1,2)$. Our goal is to show that none of them can lead to a PBNE of the DFPA. It is also safe to assume that, at all PBNE, $i$ plays either $(0,1,2)$ or $(0,2,3)$, using the analysis of the simulation of the circuit's gates, as $i$ corresponds to the output of the circuit. Start by fixing $i$'s strategy to $s_0=(0,1,2)$. By the no-overbidding assumption, $k$'s best-response must satisfy $\hat{\beta}_k(0)=0$. Then, we consider the following cases:

\begin{enumerate}

    \item Suppose bidder $\ell$ plays some strategy of the form $(0,x,0)$. We can compute the winning probabilities and utilities for bidder $k$ when bidding $0$, $1$, $2$, or $3$ and having value $10$:
    
    \begingroup
    \renewcommand{\arraystretch}{1.2}
    \begin{center}\begin{tabular}{c|cccc}
        $b$ & $0$ & $1$ & $2$ & $3$\\
        \hline
        $H_k(b)$ & $\frac{1}{2n}$ & $\frac{3}{4}$ & $1$ & $1$ \\
        $u_k(b;10)$ & $\frac{5}{n}$ & $\frac{27}{4}$ & $8$ & $7$
    \end{tabular}\end{center}
    \endgroup
    
    Since $n\geq 3$, we have the bound $u_k(0;10)\leq \frac{5}{3}$. Comparing the values in the table, we conclude that it is optimal for bidder $k$ to bid $\hat{\beta}_k(10)=2$. Also, we can see that there is no other $\varepsilon$-best-response for $\varepsilon<\frac{1}{18}$.

    \item Suppose bidder $\ell$ plays some strategy of the form $(0,x,1)$. We can compute the winning probabilities and utilities for bidder $k$ when bidding $0$, $1$, $2$, or $3$ and having value $10$:
    
    \begingroup
    \renewcommand{\arraystretch}{1.2}
    \begin{center}\begin{tabular}{c|cccc}
        $b$ & $0$ & $1$ & $2$ & $3$\\
        \hline
        $H_k(b)$ & $\frac{4}{11n}$ & $\frac{29}{44}$ & $1$ & $1$ \\
        $u_k(b;10)$ & $\frac{40}{11n}$ & $\frac{261}{44}$ & $8$ & $7$
    \end{tabular}\end{center}
    \endgroup
    
    Since $n\geq 3$, we have the bound $u_k(0;10)\leq \frac{40}{33}$. Comparing the values in the table, we conclude that it is optimal for bidder $k$ to bid $\hat{\beta}_k(10)=2$. Also, we can see that there is no other $\varepsilon$-best-response for $\varepsilon<\frac{1}{18}$.

    \item Suppose bidder $\ell$ plays some strategy of the form $(0,x,2)$. We can compute the winning probabilities and utilities for bidder $k$ when bidding $0$, $1$, $2$, or $3$ and having value $10$:
    
    \begingroup
    \renewcommand{\arraystretch}{1.2}
    \begin{center}\begin{tabular}{c|cccc}
        $b$ & $0$ & $1$ & $2$ & $3$\\
        \hline
        $H_k(b)$ & $\frac{4}{11n}$ & $\frac{6}{11}$ & $\frac{19}{22}$ & $1$ \\
        $u_k(b;10)$ & $\frac{40}{11n}$ & $\frac{72}{11}$ & $\frac{76}{11}$ & $7$
    \end{tabular}\end{center}
    \endgroup
    
    Since $n\geq 3$, we have the bound $u_k(0;10)\leq \frac{40}{33}$. Comparing the values in the table, we conclude that it is optimal for bidder $k$ to bid $\hat{\beta}_k(10)=3$. Also, we can see that there is no other $\varepsilon$-best-response for $\varepsilon<\frac{1}{18}$.

    \item Suppose bidder $\ell$ plays some strategy of the form $(0,x,3)$. We can compute the winning probabilities and utilities for bidder $k$ when bidding $0$, $1$, $2$, or $3$ and having value $10$:
    
    \begingroup
    \renewcommand{\arraystretch}{1.2}
    \begin{center}\begin{tabular}{c|cccc}
        $b$ & $0$ & $1$ & $2$ & $3$\\
        \hline
        $H_k(b)$ & $\frac{4}{11n}$ & $\frac{6}{11}$ & $\frac{8}{11}$ & $\frac{19}{22}$ \\
        $u_k(b;10)$ & $\frac{40}{11n}$ & $\frac{54}{11}$ & $\frac{64}{11}$ & $\frac{133}{22}$
    \end{tabular}\end{center}
    \endgroup
    
    Since $n\geq 3$, we have the bound $u_k(0;10)\leq \frac{65}{33}$. Comparing the values in the table, we conclude that it is optimal for bidder $k$ to bid $\hat{\beta}_k(10)=3$. Also, we can see that there is no other $\varepsilon$-best-response for $\varepsilon<\frac{1}{18}$.

\end{enumerate}

Hence, $k$'s unique ($\varepsilon$-) best-responses for $\varepsilon<\frac{1}{18}$ to $\ell$'s strategies when $i$ plays $(0,1,2)$ are the ones demonstrated in the following table:

\begin{center}\begin{tabular}{c|cccc}
    $\hat{\beta}_{\ell}(10)$ & $0$ & $1$ & $2$ & $3$\\
    \hline
    $\hat{\beta}_k(10)$ & $2$ & $2$ & $3$ & $3$ 
\end{tabular}
\captionof{table}{Bidder $k$'s best-responses given $\ell$'s strategies when $i$ plays $(0,1,2)$}
\end{center} 

Comparing this table with the one describing $\ell$'s best-responses, we can easily see that there is no equilibrium when $i$ plays $(0,1,2)$. We now examine the case where $i$ plays $(0,2,3)$. It suffices to show that there is some strategy of $\ell$ for which $k$'s best-response leads to both playing best-responses simultaneously. To this end, we consider the case where $\ell$ plays some strategy $(0,x,1)$. Then, $k$'s utility for each possible bid is the following:

\begingroup
\renewcommand{\arraystretch}{1.2}
\begin{center}\begin{tabular}{c|cccc}
    $b$ & $0$ & $1$ & $2$ & $3$\\
    \hline
    $H_k(b)$ & $\frac{4}{11n}$ & $\frac{19}{44}$ & $\frac{3}{4}$ & $1$ \\
    $u_k(b;10)$ & $\frac{40}{11n}$ & $\frac{171}{44}$ & $6$ & $7$
\end{tabular}\end{center}
\endgroup

Since $n\geq 3$, we have the bound $u_k(0;10)\leq \frac{40}{33}$. Comparing the values in the table, we conclude that it is optimal for bidder $k$ to bid $\hat{\beta}_k(10)=3$. Also, we can see that there is no other $\varepsilon$-best-response for $\varepsilon<\frac{1}{18}$. Consequently, there is a strategy profile with $k$ and $\ell$ simultaneously playing best-responses, where $\ell$ plays some strategy $(0,x,1)$ and $k$ plays some strategy $(0,x',3)$. This concludes our proof of the following lemma:

\begin{lemma}\label{lem:outputgate}
    Given the construction of the output gadget, agents $k$ and $\ell$ can simultaneously $\varepsilon$-best-respond to every agent (for $\varepsilon\in[0,1/18)$) if and only if the bidder corresponding to the output of the circuit plays the strategy $s_1=(0,2,3)$.
\end{lemma}

Combining the results of \crefrange{sec:simulation}{sec:output-gadget}, we complete the proof of the NP-hardness result. Note that this will hold for $1/10$ of the value of $\varepsilon$ we used in the reduction, due to the scaling argument that we mentioned in the beginning of \cref{app:np-hardness}.
We will demonstrate that there is a satisfying assignment to the circuit's inputs if and only if there is a PBNE in the resulting DFPA.

 Assume the circuit is satisfiable and let $\vec{\alpha}$ be a satisfying assignment to it inputs. Then, using \cref{cor:gate-simulation} and \cref{lem:outputgate}, the resulting DFPA has an exact PBNE where for each input $\alpha_i$ to the circuit, the corresponding bidders $i$ and $j$ (as described in \cref{sec:inputgadget}) both play strategy $s_0$ if ${\alpha_i}=\text{false}$, and $s_1$ otherwise. 
 
 If on the other hand the circuit is unsatisfiable, then the output bidder would always play $s_0$ in order to best-respond to the previous bidders while satisfying the circuit's gates, so there can be no $\varepsilon$-PBNE in the resulting DFPA, due to \cref{lem:outputgate}. This concludes our proof of NP-hardness in the auction we introduced for $\varepsilon<\frac{1}{18}$, thus yielding NP-hardness for $\varepsilon<\frac{1}{180}$ in the original auction.

\section{PPAD-Hardness of Finding Mixed Equilibria} \label{app:ppad-hardness}

In this section, we prove the following result:

\begin{theorem}
For any $\varepsilon < 1/36$, it is \ppad-hard to compute an $\varepsilon$-well-supported MBNE in a DFPA, even when the auction has interaction degree bounded by $d = 2$.
\end{theorem}

By \cref{lem:NE2WSNE} we also obtain the following.

\begin{corollary}
For any $\varepsilon < 1/(3^4 \cdot 2^8)$, it is \ppad-hard to compute an $\varepsilon$-approximate MBNE in a DFPA.
\end{corollary}

\addtocontents{toc}{\protect\setcounter{tocdepth}{1}}
\subsection{The \pcircuit Problem}
\addtocontents{toc}{\protect\setcounter{tocdepth}{2}}

An instance of the \pcircuit problem is given by a node set $\variables$ and a set $G$ of gates. Each gate $g \in G$ is of the form $g = (T,\uvar,\vvar,\wvar)$ where $\uvar,\vvar,\wvar \in \variables$ are distinct nodes, and $T \in \{\NOT, \AND, \PURE\}$ is the type of the gate, with the following interpretation.
\begin{itemize}
	\item[---] If $T=\NOT$, then $\uvar$ is the input of the gate, and $\vvar$ is its output. ($\wvar$ is unused)
	\item[---] If $T=\AND$, then $\uvar$ and $\vvar$ are the inputs of the gate, and $\wvar$ is its output.
	\item[---] If $T=\PURE$, then $\uvar$ is the input of the gate, and $\vvar$ and $\wvar$ are its outputs.
\end{itemize}
We require that each node is the output of exactly one gate.

A solution to instance $(\variables,G)$ is an assignment $\valonly: \variables \to \{0,1,\garbo\}$ that satisfies all the gates, i.e., for each gate $g=(T,\uvar,\vvar,\wvar) \in G$ we have the following. 
\begin{itemize}
	\item[---] If $T=\NOT$ in $g=(T,\uvar,\vvar)$, then $\valonly$ satisfies
	\begin{align*}
		&\val{\uvar} = 0 \implies \val{\vvar} = 1\\
		&\val{\uvar} = 1 \implies \val{\vvar} = 0.
	\end{align*}
	\item[---] If $T=\AND$ in $g = (T,\uvar,\vvar,\wvar)$, then $\valonly$ satisfies
	\begin{align*}
		\val{\uvar} = \val{\vvar} = 1 & \implies \val{\wvar} = 1\\
		\val{\uvar}=0 \lor \val{\vvar}=0 &\implies \val{\wvar}=0.
	\end{align*}
	\item[---] If $T=\PURE$, then $\valonly$ satisfies
	\begin{align*}
		& \{\val{\vvar}, \val{\wvar}\} \cap \{0,1\} \neq \emptyset\\
		& \val{\uvar} \in \{0,1\} \implies \val{\vvar} = \val{\wvar} = \val{\uvar}.
	\end{align*}
\end{itemize}

\begin{theorem}[\cite{deligkas2022pure}]\label{thm:purecircuit}
The \pcircuit problem is \ppad-complete.
\end{theorem}

\addtocontents{toc}{\protect\setcounter{tocdepth}{1}}
\subsection{Construction of the Auction}
\addtocontents{toc}{\protect\setcounter{tocdepth}{2}}

Consider an instance $(\variables,G)$ of \pcircuit. We now describe how to construct an instance of the discrete FPA that encodes the \pcircuit instance. In the next section, we will then argue that any approximate equilibrium of the auction must yield a solution to $(\variables,G)$.

The auction consists of one bidder for each of the variables in $\variables$, along with a number of additional auxiliary bidders. Each bidder $i$ has a value space consisting of two values, $V_i = \{0, v_i\}$, where $v_i \in (0,1]$ depends on $i$. The bidding space is $B = \{b_0,b_1,b_2,b_3\} = \{0,1/4,1/2,3/4\}$. We now describe the bidders in the auction.

\paragraph{\textbf{Constant bidders}} The purpose of this part of the construction is to have a bidder \cbidder who bids $b_2$ with probability $1$ when its value is $v_\cbidder$ in any approximate equilibrium. In order to achieve this, we introduce a set $C$ of two additional bidders and ensure that they always bid $b_1$ when their value is non-zero. In more detail, each bidder $i \in C$ has value space $V_i = \{0,v_i\}$ with $v_i = 1/2$, and we set $f_{i,j}(0) = 1, f_{i,j}(v_j) = 0$ for all other bidders $j$ in the auction. The bidder \cbidder has value space $V_\cbidder = \{0,v_\cbidder\}$ with $v_\cbidder = 1$, and we set $f_{\cbidder,j}(0) = 0, f_{\cbidder,j}(v_j) = 1$ for all $j \in C$, and $f_{i,j}(0) = 1, f_{i,j}(v_j) = 0$ for all other bidders $j$ in the auction. We show that this construction indeed satisfies the desired properties in the next section.

\paragraph{\textbf{Variable bidders}} For every variable $\uvar \in \variables$ of the \pcircuit instance, we introduce a corresponding bidder $\uvar$ in the auction. The value space $V_\uvar$ and distributions $f_{\uvar,j}$ depend on the type of gate that uses $\uvar$ as output.

\paragraph{\textbf{\AND gate}} For every variable $\wvar \in \variables$ that is the output of an \AND gate with inputs $\uvar$ and $\vvar$, we set the value space to be $V_\wvar = \{0, v_\wvar\}$, where $v_\wvar = 7/12$. The distributions are $f_{\wvar,j}(0)=0, f_{\wvar,j}(v_j)=1$ for $j \in \{\uvar, \vvar\}$, and $f_{\wvar,j}(0) = 1, f_{\wvar,j}(v_j) = 0$ for all other bidders $j$ in the auction.

\paragraph{\textbf{\PURE gate}} For every variable $\vvar \in \variables$ that is the \emph{first} output of a \PURE gate with input $\uvar$, we set the value space to be $V_\vvar = \{0,v_\vvar\}$, where $v_\vvar = 9/16$. The distributions are $f_{\vvar,\uvar}(0)=0, f_{\vvar,\uvar}(v_\uvar)=1$, and $f_{\vvar,j}(0) = 1, f_{\vvar,j}(v_j) = 0$ for all other bidders $j$ in the auction.

For every variable $\wvar \in \variables$ that is the \emph{second} output of a \PURE gate with input $\uvar$, we set the value space to be $V_\wvar = \{0,v_\wvar\}$, where $v_\wvar = 11/16$. The distributions are $f_{\wvar,\uvar}(0)=0, f_{\wvar,\uvar}(v_\uvar)=1$, and $f_{\wvar,j}(0) = 1, f_{\wvar,j}(v_j) = 0$ for all other bidders $j$ in the auction.

\paragraph{\textbf{\NOT gate}} For every variable $\vvar \in \variables$ that is the output of a \NOT gate with input $\uvar$, we introduce an additional auxiliary bidder $\vvar'$ to help with the implementation of the gate. The value space of bidder $\vvar'$ is $V_{\vvar'} = \{0,v_{\vvar'}\}$, where $v_{\vvar'} = 13/14$. The distributions are $f_{\vvar',\uvar}(0)=0, f_{\vvar',\uvar}(v_\uvar)=1$, $f_{\vvar',\cbidder}(0)=0, f_{\vvar',\cbidder}(v_\cbidder)=1$, and $f_{\vvar',j}(0) = 1, f_{\vvar',j}(v_j) = 0$ for all other bidders $j$ in the auction.

The value space of bidder $\vvar$ is $V_{\vvar} = \{0,v_{\vvar}\}$, where $v_{\vvar} = 5/8$. The distributions are $f_{\vvar,\vvar'}(0)=1/9, f_{\vvar,\vvar'}(v_{\vvar'})=8/9$, and $f_{\vvar,j}(0) = 1, f_{\vvar,j}(v_j) = 0$ for all other bidders $j$ in the auction.

\paragraph{\textbf{Properties}} This instance of the auction problem can clearly be constructed in polynomial time given an instance $(\variables,G)$ of \pcircuit. Furthermore, by construction the auction has interaction degree bounded by $d = 2$.

\addtocontents{toc}{\protect\setcounter{tocdepth}{1}}
\subsection{Analysis}
\addtocontents{toc}{\protect\setcounter{tocdepth}{2}}

Fix any constant $\varepsilon < 1/36$. Without loss of generality, we can assume that $|\variables|$ is sufficiently large so that the number of bidders in the auction, denoted by $n$, satisfies $n \geq \max\{1000,1/(1/36-\varepsilon)\}$.

Consider any $\varepsilon$-well-supported equilibrium $\vec{\beta}$ of the auction. In this section, we show how to extract a solution to the \pcircuit instance $(\variables, G)$ from $\vec{\beta}$ in polynomial time.

\paragraph{\textbf{Extraction}} We construct an assignment $\valonly: \variables \to \{0,1,\garbo\}$ to the variables of \pcircuit by letting
\begin{itemize}
    \item[---] $\val{\uvar} = 1$ if $\beta_{\uvar}(v_\uvar)(b_1) = 1$,
    \item[---] $\val{\uvar} = 0$ if $\beta_{\uvar}(v_\uvar)(b_1) = 0$,
    \item[---] $\val{\uvar} = \garbo$ otherwise.
\end{itemize}
In the remainder of this section, we argue that this assignment satisfies all the gates of the \pcircuit instance $(\variables, G)$. Note that by construction we have $v_\uvar < b_3$ for all $\uvar \in \variables$, so we will necessarily have $\beta_{\uvar}(v_\uvar)(b_3) = 0$. For a bidder $\uvar \in \variables$, we say that $\uvar$ is \emph{valid}, if, in addition, $\beta_{\uvar}(v_\uvar)(b_0) = 0$, or, in other words,
$$\beta_{\uvar}(v_\uvar)(b_1) + \beta_{\uvar}(v_\uvar)(b_2) = 1.$$
Ultimately, we will argue that all bidders $\uvar \in \variables$ must be valid.

\paragraph{\textbf{Constant bidders}}
We start by showing that the constant bidders have the desired behaviour.

\begin{lemma}\label{lem:constant-bidder}
We have $\beta_{\cbidder}(v_\cbidder)(b_2) = 1$.
\end{lemma}

\begin{proof}
We first analyse the behaviour of the bidders in set $C$. Consider any bidder $i \in C$ and note that since $f_{i,j}(0) = 1$ for all $j \neq i$, bidder $i$ believes that all other bidders have value $0$. As a result, by the no-overbidding assumption, bidder $i$ believes that all other bidders will bid $b_0 = 0$. This yields the following utilities for bidder $i$ at value $v_i = 1/2$:
\begin{equation*}
\begin{split}
u_i(b_0,\vec{\beta}_{-i};v_i) &= (v_i-b_0) \cdot H_i(b_0,\vec{\beta}_{-i}) = \frac{1}{2} \cdot \frac{1}{n}\\
u_i(b_1,\vec{\beta}_{-i};v_i) &= (v_i-b_1) \cdot H_i(b_1,\vec{\beta}_{-i}) = \frac{1}{4}\\
u_i(b_2,\vec{\beta}_{-i};v_i) &= (v_i-b_2) \cdot H_i(b_2,\vec{\beta}_{-i}) = 0
\end{split}
\end{equation*}
and bid $b_3$ is not playable because $b_3 > v_i$. Since $n \geq 1000$, we have $1/4 > 1/2n + \varepsilon$, and thus $b_1$ is the only $\varepsilon$-best-response, i.e., $\beta_i(v_i)(b_1) = 1$.

Now we turn our attention to the bidder \cbidder. Recall that $f_{\cbidder,j}(0) = 0, f_{\cbidder,j}(v_j) = 1$ for all $j \in C$, and $f_{i,j}(0) = 1, f_{i,j}(v_j) = 0$ for all other bidders $j$ in the auction. As a result, since $\beta_j(v_j)(b_1) = 1$ for all $j \in C$, bidder \cbidder believes that all bidders $j \in C$ bid $b_1$. Furthermore, by the same argument as above, bidder \cbidder believes that all other bidders outside $C$ bid $b_0 = 0$. This yields the following utilities for bidder $i = \cbidder$ at value $v_i = v_\cbidder = 1$:
\begin{equation*}
\begin{split}
u_i(b_0,\vec{\beta}_{-i};v_i) &= (v_i-b_0) \cdot H_i(b_0,\vec{\beta}_{-i}) = 0\\
u_i(b_1,\vec{\beta}_{-i};v_i) &= (v_i-b_1) \cdot H_i(b_1,\vec{\beta}_{-i}) = \frac{3}{4} \cdot \frac{1}{|C|+1}\\
u_i(b_2,\vec{\beta}_{-i};v_i) &= (v_i-b_2) \cdot H_i(b_2,\vec{\beta}_{-i}) = \frac{1}{2}\\
u_i(b_3,\vec{\beta}_{-i};v_i) &= (v_i-b_3) \cdot H_i(b_3,\vec{\beta}_{-i}) = \frac{1}{4}\\
\end{split}
\end{equation*}
Recalling that $|C| = 2$, we see that $b_2$ is the only $\varepsilon$-best-response, since $\varepsilon < 1/4$.
\end{proof}

\paragraph{\textbf{\AND gate}}
Next, we consider bidders corresponding to an \AND gate.

\begin{lemma}\label{lem:pcircuit-AND}
For any \AND gate with inputs $\uvar$ and $\vvar$, and output $\wvar$, we have that
\begin{itemize}
    \item[---] $\wvar$ is valid,
    \item[---] if $\uvar$ and $\vvar$ are valid, then $\valonly$ satisfies the conditions of the \AND gate.
\end{itemize}
\end{lemma}

\begin{proof}
Let us first show that $\wvar$ is valid. By construction, bidder $\wvar$ believes that all other bidders, except $\uvar$ and $\vvar$, will bid $b_0$. Furthermore, $\wvar$ believes that $\uvar$ and $\vvar$ will bid according to $\beta_{\uvar}(v_\uvar)$ and $\beta_{\vvar}(v_\vvar)$, respectively. Recalling that $\beta_{\uvar}(v_\uvar)(b_3) = \beta_{\vvar}(v_\vvar)(b_3) = 0$, and using the notation $p_j = \beta_{\uvar}(v_\uvar)(b_j)$, $q_j = \beta_{\vvar}(v_\vvar)(b_j)$, for $j \in \{0,1,2\}$, we can write the utilities of bidder $i=\wvar$ at value $v_i = v_\wvar = 7/12$:
\begin{equation*}
\begin{split}
u_i(b_0,\vec{\beta}_{-i};v_i) &= (v_i-b_0) \cdot H_i(b_0,\vec{\beta}_{-i}) = \frac{7}{12}p_0q_0/n\\
u_i(b_1,\vec{\beta}_{-i};v_i) &= (v_i-b_1) \cdot H_i(b_1,\vec{\beta}_{-i}) = \left(\frac{7}{12}-\frac{1}{4}\right)(p_0q_0 + p_1q_0/2 + p_0q_1/2 + p_1q_1/3) \\
u_i(b_2,\vec{\beta}_{-i};v_i) &= \left(\frac{7}{12}-\frac{1}{2}\right)((1-p_2)(1-q_2) + (1-p_2)q_2/2 + p_2(1-q_2)/2 + p_2q_2/3)
\end{split}
\end{equation*}
and bid $b_3$ is not playable because $b_3 > v_i$. Note that we always have $u_i(b_2,\vec{\beta}_{-i};v_i) \geq (7/12-1/2)/3 = 1/36$, because $(1-p_2)(1-q_2) + (1-p_2)q_2/2 + p_2(1-q_2)/2 + p_2q_2/3 \geq 1/3$. Furthermore, $u_i(b_0,\vec{\beta}_{-i};v_i) \leq 7/12n < 1/n \leq 1/36 - \varepsilon$, since $n \geq 1/(1/36-\varepsilon)$.
As a result, $b_0$ can never be an $\varepsilon$-best-response for $\wvar$, and thus $\wvar$ is valid.

Now consider the case where $\uvar$ and $\vvar$ are valid, i.e., $p_0 = q_0 = 0$. In that case, the utilities can be simplified to
\begin{equation*}
\begin{split}
u_i(b_0,\vec{\beta}_{-i};v_i) &= 0\\
u_i(b_1,\vec{\beta}_{-i};v_i) &= \left(\frac{7}{12}-\frac{1}{4}\right) p_1q_1/3 \\
u_i(b_2,\vec{\beta}_{-i};v_i) &= \left(\frac{7}{12}-\frac{1}{2}\right)(p_1q_1 + p_1(1-q_1)/2 + (1-p_1)q_1/2 + (1-p_1)(1-q_1)/3)
\end{split}
\end{equation*}
If $\val{\uvar} = \val{\vvar} = 1$, i.e., $p_1=q_1=1$, then
$$u_i(b_1,\vec{\beta}_{-i};v_i) - u_i(b_2,\vec{\beta}_{-i};v_i) = \left(\frac{7}{12}-\frac{1}{4}\right) \frac{1}{3} - \left(\frac{7}{12}-\frac{1}{2}\right) = \frac{1}{36} > \varepsilon$$
and so $b_1$ is the only $\varepsilon$-best-response, i.e., $\val{\wvar} = 1$.

Finally, consider the case where $\val{\uvar} = 0$ or $\val{\vvar} = 0$. Without loss of generality, given the symmetry of the construction, assume that $\val{\uvar} = 0$, i.e., $p_1 = 0$. In that case, we have
$$u_i(b_2,\vec{\beta}_{-i};v_i) - u_i(b_1,\vec{\beta}_{-i};v_i) = \frac{1}{12}(q_1/2+(1-q_1)/3) - 0 \geq \frac{1}{36} > \varepsilon$$
and so $b_2$ is the only $\varepsilon$-best-response, i.e., $\val{\wvar} = 0$.
\end{proof}

\paragraph{\textbf{\PURE gate}}
Next, we consider bidders corresponding to a \PURE gate.

\begin{lemma}\label{lem:pcircuit-PURE}
For any \PURE gate with input $\uvar$ and outputs $\vvar$ and $\wvar$, we have that
\begin{itemize}
    \item[---] $\vvar$ and $\wvar$ are valid,
    \item[---] if $\uvar$ is valid, then $\valonly$ satisfies the conditions of the \PURE gate.
\end{itemize}
\end{lemma}

\begin{proof}
Let us first show that $\vvar$ and $\wvar$ are valid. By construction, bidder $\vvar$ believes that all other bidders except $\uvar$ will bid $b_0$. Furthermore, $\vvar$ believes that $\uvar$ will bid according to $\beta_\uvar(v_\uvar)$. Recalling that $\beta_\uvar(v_\uvar)(b_3) = 0$, and using the notation $p_j = \beta_\uvar(v_\uvar)(b_j)$ for $j \in \{0,1,2\}$, we can write the utilities of bidder $i = \vvar$ at value $v_i = v_\vvar = 9/16$:
\begin{equation*}
\begin{split}
u_i(b_0,\vec{\beta}_{-i};v_i) &= (v_i-b_0) \cdot H_i(b_0,\vec{\beta}_{-i}) = \frac{9}{16}p_0/n\\
u_i(b_1,\vec{\beta}_{-i};v_i) &= (v_i-b_1) \cdot H_i(b_1,\vec{\beta}_{-i}) = \left(\frac{9}{16}-\frac{1}{4}\right)(p_0+p_1/2) \\
u_i(b_2,\vec{\beta}_{-i};v_i) &= (v_i-b_2) \cdot H_i(b_2,\vec{\beta}_{-i}) = \left(\frac{9}{16}-\frac{1}{2}\right)(p_0+p_1+p_2/2)
\end{split}
\end{equation*}
and bid $b_3$ is not playable because $b_3 > v_i$. Since $p_0+p_1+p_2/2 \geq 1/2$, we can write
$$u_i(b_2,\vec{\beta}_{-i};v_i) - u_i(b_0,\vec{\beta}_{-i};v_i) \geq \left(\frac{9}{16}-\frac{1}{2}\right) \frac{1}{2} - \frac{9}{16}p_0/n = \frac{1}{32} - \frac{9}{16}p_0/n > \varepsilon$$
because $n \geq 1000$. As a result, $b_0$ can never be an $\varepsilon$-best-response for $\vvar$, and thus $\vvar$ is valid. The same reasoning for $i = \wvar$, with $v_i = v_\wvar = 11/16$, yields
$$u_i(b_2,\vec{\beta}_{-i};v_i) - u_i(b_0,\vec{\beta}_{-i};v_i) \geq \left(\frac{11}{16}-\frac{1}{2}\right) \frac{1}{2} - \frac{11}{16}p_0/n = \frac{3}{32} - \frac{11}{16}p_0/n > \varepsilon$$
because $n \geq 1000$. Thus, $b_0$ can never be an $\varepsilon$-best-response for $\wvar$, and $\wvar$ is also valid.

Now consider the case where $\uvar$ is valid, i.e., $p_0 = 0$. In that case, the utilities for $i \in \{\vvar, \wvar\}$ at value $v_i$ can be simplified to
\begin{equation*}
\begin{split}
u_i(b_0,\vec{\beta}_{-i};v_i) &= 0\\
u_i(b_1,\vec{\beta}_{-i};v_i) &= \left(v_i - \frac{1}{4}\right)p_1/2 \\
u_i(b_2,\vec{\beta}_{-i};v_i) &= \left(v_i - \frac{1}{2}\right)(1/2+p_1/2)
\end{split}
\end{equation*}
We can write
$$\Delta := u_i(b_1,\vec{\beta}_{-i};v_i) - u_i(b_2,\vec{\beta}_{-i};v_i) = \left(v_i - \frac{1}{4}\right)p_1/2 - \left(v_i - \frac{1}{2}\right)(1/2+p_1/2) = p_1/8 - v_i/2 + 1/4.$$
Now for $i=\vvar$, and thus $v_i = 9/16$, we obtain that
\begin{itemize}
    \item[---] if $\val{\uvar} = 0$, i.e., $p_1 = 0$, then $\Delta = p_1/8 - 1/32 = -1/32 < - \varepsilon$, and $b_2$ is the only $\varepsilon$-best-response, i.e., $\val{\vvar} = 0$.
    \item[---] if $\val{\uvar} \geq 1/2$, i.e., $p_1 \geq 1/2$, then $\Delta = p_1/8 - 1/32 \geq 1/32 > \varepsilon$, and $b_1$ is the only $\varepsilon$-best-response, i.e., $\val{\vvar} = 1$.
\end{itemize}
Similarly, for $i=\wvar$, and thus $v_i = 11/16$, we obtain that
\begin{itemize}
    \item[---] if $\val{\uvar} \leq 1/2$, i.e., $p_1 \leq 1/2$, then $\Delta = p_1/8 - 3/32 \leq -1/32 < - \varepsilon$, and $b_2$ is the only $\varepsilon$-best-response, i.e., $\val{\wvar} = 0$.
    \item[---] if $\val{\uvar} = 1$, i.e., $p_1 = 1$, then $\Delta = p_1/8 - 3/32 = 1/32 > \varepsilon$, and $b_1$ is the only $\varepsilon$-best-response, i.e., $\val{\wvar} = 1$.
\end{itemize}
As a result, the conditions of the \PURE gate are indeed satisfied. In particular, we always have that at least one of $\val{\vvar}$ or $\val{\wvar}$ lies in $\{0,1\}$.
\end{proof}

\paragraph{\textbf{\NOT gate}}
Finally, we consider bidders corresponding to a \NOT gate.

\begin{lemma}\label{lem:pcircuit-NOT}
For any \NOT gate with input $\uvar$ and output $\vvar$, we have that
\begin{itemize}
    \item[---] $\vvar$ is valid,
    \item[---] if $\uvar$ is valid, then $\valonly$ satisfies the conditions of the \NOT gate.
\end{itemize}
\end{lemma}

\begin{proof}
We first show that $\vvar$ is valid. By construction, bidder $\vvar$ believes that all other bidders except $\vvar'$ will bid $b_0$. Furthermore, $\vvar$ believes that $\vvar'$ will bid according to $\beta_{\vvar'}(0)$ with probability $1/9$ (namely when $\vvar'$ has value $0$), and according to $\beta_{\vvar'}(v_{\vvar'})$ with probability $8/9$ (namely when $\vvar'$ has value $v_{\vvar'}$). By the no-overbidding assumption, $\beta_{\vvar'}(0)$ is the distribution with $\beta_{\vvar'}(0)(b_0) = 1$. Using the notation $q_j = \beta_{\vvar'}(v_{\vvar'})(b_j)$ for $j = 0,1,2,3$, we can write the utilities of bidder $i = \vvar$ at value $v_i = v_\vvar = 5/8$:
\begin{equation*}
\begin{split}
u_i(b_0,\vec{\beta}_{-i};v_i) &= \frac{5}{8}\left(\frac{1}{9} \cdot \frac{1}{n} + \frac{8}{9} \cdot \frac{q_0}{n}\right)\\
u_i(b_1,\vec{\beta}_{-i};v_i) &= \left(\frac{5}{8} - \frac{1}{4}\right)\left(\frac{1}{9} + \frac{8}{9} (q_0 + q_1/2) \right) \\
u_i(b_2,\vec{\beta}_{-i};v_i) &= \left(\frac{5}{8} - \frac{1}{2}\right)\left(\frac{1}{9} + \frac{8}{9} (q_0 + q_1 + q_2/2) \right)
\end{split}
\end{equation*}
and bid $b_3$ is not playable because $b_3 > v_i$. We always have
$$u_i(b_1,\vec{\beta}_{-i};v_i) - u_i(b_0,\vec{\beta}_{-i};v_i) \geq \frac{3}{8} \cdot \frac{1}{9} - \frac{5}{8} \cdot \left(\frac{1}{9} \cdot \frac{1}{n} + \frac{8}{9} \cdot \frac{1}{n} \right) = \frac{3}{72} - \frac{5}{8n} > \varepsilon$$
because $n \geq 1000$. As a result, $b_0$ is never an $\varepsilon$-best-response for $\vvar$ at value $v_\vvar$, and thus $\vvar$ is valid. Furthermore, we note the following, which will be useful below
\begin{itemize}
    \item[---] if $\beta_{\vvar'}(v_{\vvar'})(b_2) = 1$, i.e., $q_2 = 1$, then
    $$u_i(b_2,\vec{\beta}_{-i};v_i) - u_i(b_1,\vec{\beta}_{-i};v_i) = \frac{1}{8}\left(\frac{1}{9} + \frac{8}{9} \cdot \frac{1}{2}\right) - \frac{3}{8} \cdot \frac{1}{9} = \frac{1}{36} > \varepsilon$$
    and thus $b_2$ is the only $\varepsilon$-best-response, i.e., $\val{\vvar} = 0$.
    \item[---] if $\beta_{\vvar'}(v_{\vvar'})(b_3) = 1$, i.e., $q_3 = 1$, then
    $$u_i(b_1,\vec{\beta}_{-i};v_i) - u_i(b_2,\vec{\beta}_{-i};v_i) = \frac{3}{8} \cdot \frac{1}{9} - \frac{1}{8} \cdot \frac{1}{9} = \frac{1}{36} > \varepsilon$$
    and thus $b_1$ is the only $\varepsilon$-best-response, i.e., $\val{\vvar} = 1$.
\end{itemize}

From now on, we assume that $\uvar$ is valid, and aim to show that the constraints of the \NOT gate are satisfied. By construction, bidder $\vvar'$ believes that all other bidders except $\uvar$ and \cbidder will bid $b_0$. By \cref{lem:constant-bidder} we know that bidder $\vvar'$ believes that bidder \cbidder will bid $b_2$. Finally, bidder $\vvar'$ believes that $\uvar$ will bid according to $\beta_\uvar(v_\uvar)$, where $\beta_\uvar(v_\uvar)(b_0) = \beta_\uvar(v_\uvar)(b_3) = 0$, since $\uvar$ is valid. Using the notation $p_1 = \beta_\uvar(v_\uvar)(b_1)$, and thus $1-p_1 = \beta_\uvar(v_\uvar)(b_2)$, we can write the utilities of bidder $i =\vvar'$ at value $v_i = v_{\vvar'} = 13/14$:
\begin{equation*}
\begin{split}
u_i(b_0,\vec{\beta}_{-i};v_i) &= \frac{13}{14} \cdot 0 = 0\\
u_i(b_1,\vec{\beta}_{-i};v_i) &= \left(\frac{13}{14} - \frac{1}{4}\right) \cdot 0 = 0\\
u_i(b_2,\vec{\beta}_{-i};v_i) &= \left(\frac{13}{14} - \frac{1}{2}\right)(p_1/2 + (1-p_1)/3)\\
u_i(b_3,\vec{\beta}_{-i};v_i) &= \left(\frac{13}{14} - \frac{3}{4}\right) \cdot 1
\end{split}
\end{equation*}
Now, we have
\begin{itemize}
    \item[---] if $\val{\uvar} = 0$, i.e., $p_1 = 0$, then
    $$u_i(b_3,\vec{\beta}_{-i};v_i) - u_i(b_2,\vec{\beta}_{-i};v_i) = \frac{5}{28} - \frac{3}{7} \cdot \frac{1}{3} = \frac{1}{28} > \varepsilon$$
    and thus $b_3$ is the only $\varepsilon$-best-response, i.e., $\beta_{\vvar'}(v_{\vvar'})(b_3) = 1$. As proved above in the analysis of bidder $\vvar$, this implies $\val{\vvar} = 1$.
    \item[---] if $\val{\uvar} = 1$, i.e., $p_1 = 1$, then
    $$u_i(b_2,\vec{\beta}_{-i};v_i) - u_i(b_3,\vec{\beta}_{-i};v_i) = \frac{3}{7} \cdot \frac{1}{2} - \frac{5}{28} = \frac{1}{28} > \varepsilon$$
    and thus $b_2$ is the only $\varepsilon$-best-response, i.e., $\beta_{\vvar'}(v_{\vvar'})(b_2) = 1$. As proved above in the analysis of bidder $\vvar$, this implies $\val{\vvar} = 0$.
\end{itemize}
As a result, the conditions of the \NOT gate are indeed satisfied.
\end{proof}

To complete the proof, note that by
\cref{lem:pcircuit-AND,lem:pcircuit-PURE,lem:pcircuit-NOT} all bidders $\uvar
\in \variables$ are valid, and thus, again by these lemmas, the constructed
assignment $\valonly$ satisfies all the gates of the \pcircuit instance
$(\variables,G)$.

\section{Existence of Symmetric Equilibria in IID Auctions}
\label{sec:existence-symmetric-MBNE-iid-DFPA}
This section is dedicated to proving ~\cref{th:existence-symmetric-equilibria-iid}. En route, we establish two
results that could be of independent interest, namely that (a) symmetric
continuous first-price auctions (with discrete bidding space) always have pure
symmetric equilibria (see~\cref{th:existence-symmetric-PBNE-iid-CFPA}), and (b)
the structure and approximability of (mixed) equilibria is preserved, in the
limit, in discrete auctions
(see~\cref{lemma:convergence-approx-equilibria-DFPA}).

\begin{theorem}
    \label{th:existence-symmetric-PBNE-iid-CFPA} Any first-price auction with continuous iid priors and discrete bids has a symmetric and monotone pure equilibrium.
\end{theorem}

\begin{proof}
    Our proof follows closely that of~\textcite{Athey2001} for the existence of (pure) equilibria in continuous auctions for general (i.e., not necessarily symmetric) priors. We need to take care in handling correctly the symmetry of the underlying strategies, but this turns out to be straightforward.
    
    Fix a CFPA $\mathcal A$ with iid priors over $[0,1]$ and discrete bidding space $B=\ssets{b_1\leq b_2 \leq \dots \leq b_m}$. It is known that, in this continuous setting, any monotone strategy $\beta:[0,1]\rightarrow B$ can be represented by a non-decreasing sequence of break points $\alpha=\alpha(\beta)=(\alpha_1,\alpha_2,\dots,\alpha_m)$ such that $\alpha_\ell=\sup\ssets{v\in[0,1]\fwh{\beta(v) \leq b_\ell}}$, for $\ell\in[m]$ (for a more detailed discussion, see~\parencite[Sec.~2]{fghlp2021_sicomp}). For simplicity, in the following we will slightly abuse notation, and sometimes use directly the jump-point representation $\alpha=\alpha(\beta)$ instead of a bidding strategy $\beta$. 
    
    Then, let us also denote the corresponding space of all feasible, no-overbidding (jump-point representation) strategies by
    $$
        \Omega = \sset{\alpha\in [0,1]^m\fwh{\alpha_{\ell}\leq \alpha_{\ell+1}\;\; \forall \ell\in[m-1]\quad\land\quad b_\ell \leq \alpha_\ell\;\;\forall \ell\in[m]}}.
    $$
    Notice that the linearity of the inequalities defining $\Omega$ imply that it is a polytope, and thus $\Omega$ is convex and compact subspace of the Euclidean space $\R^{m}$. Also, it is easy to see that $\Omega$ is nonempty, since $(1,1,\dots,1)\in \Omega$ (corresponding to the strategy of always bidding $\min B = b_1$).

    Next, given a symmetric and monotone strategy profile $\vec{\alpha} = (\alpha,\alpha,\dots,\alpha)$, we define (the set of) its best-responses by
    $$\mathsf{BR}(\alpha)\coloneqq \sset{\alpha'\in\Omega\fwh{u_i\left(\alpha'(v_i),\vec{\alpha}_{-i};v_i\right) \geq \max_{b\in B} u_i\left(b,\vec{\alpha}_{-i};v_i\right)\quad\forall v_i\in V }}$$
    Notice that, since auction $\mathcal A$ has iid priors, the specific choice of the bidder $i$ in the above definition is irrelevant.
    This gives rise to a correspondence $\mathsf{BR}:\Omega\rightarrow 2^{\Omega}$ with the property that: a symmetric and monotone strategy profile, where all bidders play a strategy $\alpha$, is an (exact) PBNE of $\mathcal A$, if and only if $\alpha\in\mathsf{BR}(\alpha)$. 

    Using this property, we can establish the existence of a symmetric PBNE by means of Kakutani's fixed-point theorem (see, e.g., \parencite[p.~583]{Aliprantis2006a}). Towards that end, we need to establish the following properties:
    \begin{itemize}
        \item[---] $\mathsf{BR}(\alpha)\neq \emptyset$ for all $\alpha\in\Omega$. 
        \item[---] $\mathsf{BR}(\alpha)$ is a convex set, for all $\alpha\in\Omega$. 
        \item[---] $\mathsf{BR}$ has a closed graph.
    \end{itemize}
    The above properties can be proved in an identical way to that of the existing results for general priors, and we refer the interested reader to~\parencite[Lemma~3]{Athey2001} or~\parencite[p.~302]{krishna2009auction}.
\end{proof}

\begin{lemma}[Convergence Lemma]
    \label{lemma:convergence-approx-equilibria-DFPA} Consider a DFPA $\mathcal
    A$ with iid priors. Assume that $\mathcal A$ has a symmetric
    monotone $\varepsilon$-approximate MBNE, for all $\varepsilon >0$. Then,
    $\mathcal A$ has an \emph{exact}, symmetric and monotone, MBNE.
\end{lemma}
\begin{proof}
    Fix a DFPA $\mathcal{A}$ with iid priors over a common value space $V$,
    and bidding space $B$. By the assumptions of our lemma, there exists a
    sequence of strategy profiles
    $\ssets{\vec{\beta}_n=(\beta_n,\beta_n,\dots,\beta_n)}$, $n\in \N$, where
    each bidder plays according to the same, monotone, strategy
    $\beta_n:V\rightarrow \Delta(B)$, such that $\vec{\beta}_n$ is a
    $\frac{1}{n}$-approximate MBNE of $\mathcal A$ for all positive integers
    $n$. For our proof, we will show that
    \begin{itemize}
        \item[(a)] there exists a subsequence $\ssets{\bar\beta_n}$ of $\ssets{\beta_n}$, such that $\bar\beta_n \to \beta^*$ (as $n\to \infty$), where $\beta^*$ is a monotone strategy.
        \item[(b)] $\vec{\beta}^*=(\beta^*,\beta^*,\dots,\beta^*)$ is an \emph{exact} MBNE.
    \end{itemize}

    Starting on (a), first observe that since we are at a discrete auction,
    there are finitely many different supports that a bidding strategy can have.
    Formally, since sets $V,B$ are finite, the set $\sset{\times_{v\in
    V}\support{\beta(v)}}_{\beta\in \Delta(B)^V}$ is also finite. Therefore,
    there has to exist a subsequence $\ssets{\tilde{\beta_n}}$ of
    $\ssets{\beta_n}$ such that all its elements have exactly the same support.

    Let $\tilde p_n=\left(\tilde p_n(v,b)\right)_{v\in V, b\in B}$ denote the
    canonical representation of strategy $\tilde\beta_n$; that is, $\tilde p_n(v,b)$
    is the probability of submitting bid $b$ when having true value $v$, under
    strategy $\tilde \beta_n$ (see~\cref{sec:model}), for all $n\in \N$. Observe
    that all these probabilities must satisfy the feasibility constraints
    \begin{align*}
        \tilde{p}_n(v,b) &\geq 0, \quad
        \sum_{b'\in B} \tilde{p}_n(v,b') = 1, \quad\text{and}\quad
        \tilde{p}_n(v,b) (v-b) \geq 0
        &&\quad\forall{v\in V, b\in B},
    \end{align*}
    which define a polytope, and thus a compact subspace, of the Euclidean
    metric space $\R^{V\times B}$. Therefore, there has to exist a subsequence
    $\ssets{\bar p_n}$ of $\ssets{\tilde p_n}$ that converges to a feasible
    representation $\bar p^*$ of bidding strategy $\beta^*\in \Delta(B)^V$.

    The only thing remaining in order to establish point (a) is the
    monotonicity of $\beta^*$. Indeed, since all strategies of
    $\ssets{\tilde{\beta}_n}$ has the same support, it must be that \emph{all}
    elements of the sequence $\ssets{\bar{p}_n}$ have exactly the same zero
    components. Therefore, due to compactness, the only way in which the
    supports of limiting strategy $\beta^*$ may differ from the supports of all
    strategies in $\ssets{\bar{\beta}_n}$ is if for a probability component it
    is $p^*(v,b)=0$ but $\bar{p}_n(v,b)\neq 0$ (for all $n\in\N$). But then,
    this implies that the supports may only \emph{shrink} in the limit; that is,
    $\support{\beta^*(v)} \subseteq \support{\bar{\beta}_n(v)}$ for all $n\in
    \N$. So, given that $\bar{\beta}_n$'s are monotone, $\beta^*$ must also be
    monotone; see~\cref{def:monotonicty-discrete}.
    
    Moving now to point (b), for a symmetric strategy profile
    $\vec{\beta}=(\beta,\beta,\dots,\beta)$, let
    \begin{equation}
        \label{eq:improvement-function-Z}
    Z(\beta) \coloneqq 
    \max_{v\in V, b\in B} u_i(b,\vec{\beta}_{-i};v_i) - u_i(\beta(v_i),\vec{\beta}_{-i};v_i)
    \end{equation}
    denote the maximum utility improvement that a bidder $i$ can achieve by
    unilaterally deviating, from the bid $\beta(v_i)$ dictated by $\beta$, to
    any other bid, while all other bidders play according to the symmetric
    profile $\vec{\beta}=(\beta,\beta,\dots,\beta)$, across all possible true
    values $v_i$. Note that, due to symmetry, the specific choice of the bidder
    $i$ in the above definition~\eqref{eq:improvement-function-Z} is irrelevant.
    Then, $\vec{\beta}$ is an $\varepsilon$-approximate MBNE of $\mathcal A$, if
    and only if $Z(\beta) \leq \varepsilon$;
    see~\cref{def:approx-mixed-bayes-nash-equilibrium} Obviously, $Z(\beta)\geq
    0$ for any strategy profile $\vec{\beta}$.

    Observe that, as defined by~\eqref{eq:improvement-function-Z},
    $Z:\Delta(B)^V\rightarrow [0,1]$ is a continuous function: bidder utilities
    are continuous (see~\eqref{eq:DFPA-utility-interim-mixed}) and (finitely
    many) max-operators preserve continuity. So,
    $$
    Z(\beta^*)= \lim_{n\to\infty} Z(\bar{\beta}_n) \leq  \lim_{n\to\infty} \frac{1}{n} = 0,
    $$
    meaning that indeed $\vec{\beta}^*$ is an exact MBNE of $\mathcal{A}$.
\end{proof}

Now we have all the necessary machinery in order to prove~\cref{th:existence-symmetric-equilibria-iid}:
\begin{proof}[Proof of~\cref{th:existence-symmetric-equilibria-iid}]
    Fix a DFPA $\mathcal A$ with iid priors and choose an arbitrary $\varepsilon \in (0,1)$. Then, by making use of our computational equivalence\footnote{As a matter of fact, in this proof we make use of the reduction from discrete to continuous auctions by deploying~\cref{lem:DFPA2CFPA} with $\varepsilon \gets 0$ and $\delta \gets \varepsilon$ in its statement.} between discrete and continuous auctions from~\cref{sec:equivalence}, we can construct a CFPA $\mathcal{A}'$ with iid priors (and the same, discrete, bidding space as $\mathcal A$) such that, any $\varepsilon$-approximate symmetric (and monotone) PBNE of $\mathcal A'$ can be translated back to a symmetric and monotone $\varepsilon$-approximate MBNE of $\mathcal A$.

    Thus, given that such $\varepsilon$-approximate symmetric PBNE are guaranteed to exist, due to~\cref{th:existence-symmetric-PBNE-iid-CFPA}, we also get guaranteed existence of an $\varepsilon$-approximate, symmetric and monotone MBNE in our original discrete auction $\mathcal A$. Since we chose $\varepsilon>0$ arbitrarily, our Convergence~\cref{lemma:convergence-approx-equilibria-DFPA} gives us the desired existence of a monotone \emph{exact} MBNE in $\mathcal A$.
\end{proof}

\printbibliography

\end{document}
\typeout{get arXiv to do 4 passes: Label(s) may have changed. Rerun}